\documentclass[a4paper,11pt]{article}

\usepackage[english]{babel}
\usepackage[utf8]{inputenc}
\usepackage[T1]{fontenc}

\usepackage{amsmath,amsfonts,amsthm,amssymb}
\usepackage{manfnt}
\usepackage{stmaryrd}
\usepackage{subcaption}
\usepackage{graphicx}
\usepackage{proof}
\usepackage{tikz}
\usepackage{mathtools}

\usepackage{hyperref}
\usepackage{cite}

\newtheorem{theorem}{Theorem}
\newtheorem{definition}{Definition}
\newtheorem{proposition}{Proposition}
\newtheorem{lemma}{Lemma}

\renewcommand{\phi}{\varphi}
\renewcommand{\epsilon}{\varepsilon}
\renewcommand{\leq}{\leqslant}
\renewcommand{\leq}{\leqslant}
\renewcommand{\exp}{e}

\newcommand{\tiling}{\mathcal{T}}
\newcommand{\tileset}{\mathbf{T}}

\newcommand{\slope}{\mathcal{E}}
\newcommand{\ea}{\mathcal{E}}

\newcommand{\e}[1]{\mathcal{E}_{#1}}
\newcommand{\R}[1]{\mathbb{R}^{#1}}

\newcommand{\circulant}[1]{\text{\rm circulant}#1}

\newcommand{\nn}{\mathbb{N}}
\newcommand{\diam}{\text{diam}}

\newcommand{\trans}{^\mathsf{T}}
\newcommand{\imag}{\mathrm{i}}
\newcommand{\tile}[2]{T_{#1,#2}}
\newcommand{\intsquare}[3]{S_{#1,#2,#3}}
\newcommand{\boundary}{\partial\sigma}

\newcommand{\fst}[1]{I(#1)}
\newcommand{\scd}[1]{J(#1)}
\newcommand{\ceil}[1]{\left\lceil #1 \right\rceil}
\newcommand{\floor}[1]{\left\lfloor #1 \right\rfloor}
\newcommand{\candidate}[1]{\sigma_{(#1)}}
\newcommand{\eigenmatrix}[1]{N_{#1}}

\newcommand{\ie}{\emph{i.e.}}

\title{Substitution discrete plane tilings with $2n$-fold rotational symmetry for odd $n$}
\author{Jarkko \textsc{Kari} $^{1}$, Victor H. \textsc{Lutfalla} $^{2\ast}$\\
  \normalsize{$^{1}$ Department of Mathematics, FI-20014 University of Turku, Finland}\\
  \normalsize{$^{2}$ Laboratoire d'Informatique de Paris Nord, UMR CNRS 7030, } \\ \normalsize{Institut Galilée - Université Paris 13, 99 avenue J.B. Clément,} \\ \normalsize{ 93430 Villetaneuse, France}\\
\normalsize{$^\ast$ Corresponding author, email: lutfalla@lipn.univ-paris13.fr, }
}
\date{\today}
 
\begin{document}

\maketitle

\subsection*{Abstract}
\label{sec:abstract}

We study substitution tilings that are also discrete plane tilings, that is, satisfy a relaxed version of cut-and-projection.
We prove that the Sub Rosa substitution tilings with a $2n$-fold rotational symmetry for odd $n>5$ defined by Kari and Rissanen are not discrete planes --
and therefore not cut-and-project tilings either. We then define new Planar Rosa substitution tilings  with a $2n$-fold rotational symmetry for any odd $n$, and show that
these satisfy the discrete plane condition. The tilings we consider are edge-to-edge rhombus tilings. We give an explicit construction for the 10-fold case, and provide
a construction method for the general case of any odd $n$.

 Our methods are to lift the tilings  and substitutions to $\R{n}$ using the lift operator first defined by Levitov,
 and to study the planarity of substitution tilings in $\R{n}$ using mainly linear algebra, properties of circulant matrices, and trigonometric sums.
For the construction of the Planar Rosa substitutions we additionally use the Kenyon criterion  and a result on De Bruijn multigrid dual tilings.

\paragraph{Keywords:}substitution tilings, discrete planes, cut-and-project tiling, n-fold symmetric tiling, quasiperiodic tilings, rhombus tiling

\paragraph{Acknowledgments.}
We wish to thank Thierry Monteil and Nicolas Bédaride for their help regarding billiard words, and also  Thomas Fernique for his help and proofreading.

\tableofcontents

\section{Introduction}
\label{sec:intro}

A \emph{tiling} is an exact covering of the plane by tiles,
meaning that the union of tiles is the whole plane and the tiles do not overlap,  \emph{i.e.}, they only intersect on their boundaries.
We are interested in aperiodic tilings with long-range order. These are important as models of quasicrystals \cite{baake2013}. From a physicists's point of view, long-range order is characterized by a ``sharp'' diffraction pattern, which mathematically corresponds to a discrete Fourier transform of the tiling \cite{baake2013, baake2017}.

Long-range order is present in
\emph{cut-and-project} tilings \cite{baake2013}, which are tilings that can be seen as the projection of a discrete 2D plane in some higher dimensional space.
A \emph{discrete 2D-surface} in $\R{n}$ is a collection of adjacent \emph{squares} which are the translates of the 2D-facets of the unit hypercube. Such a surface is called \emph{planar} if it approximates a plane, meaning that there exists a two-dimensional real plane of
$\R{n}$ to which the surface has a bounded distance.
Any edge-to-edge rhombus tiling with $n$ edge directions can be lifted into a discrete surface in $\R{n}$. This rhombus tiling is said to be a \emph{discrete plane} tiling if its lifted version is planar.
Let us emphasize that the plane in which the tiling is defined (before lifting it in $\R{n}$) and the 2D plane that is approximated by the discrete surface (after it is lifted in $\R{n}$) are not the same plane.
This property of planarity is less restrictive than the cut-and-project property, which puts a more strict
condition on the distance of the lifted discrete surface to a 2D plane.
In this article we focus on discrete plane tilings instead of the more restrictive cut-and-project case.
Note that the discrete plane condition is enough for so called \emph{essentially discrete} diffraction patterns. This means that from a physicist's viewpoint, both the discrete plane property and the cut-and-project property capture long-range order and are in this sense similar.

Historically, the discrete plane condition was used in \cite{levitov1988} in connection to tilings under local matching rules. In that context, the
terms ``strong'' and ``weak'' local rules
were used for matching rules that force cut-and-project and discrete plane conditions on the admitted tilings, respectively.
The discrete plane condition has also been studied under the name ``planar tilings'' and ``planarity'', see for example \cite{bedaride2015}. We have chosen to use the name ``discrete plane tilings'' or simply ``discrete planes'' to remove the possible confusion between ``planar tilings'' and ``tilings of the plane''. Remark that the name discrete planes is widely used for the discrete approximation of 2D planes in $\R{n}$.

Tilings with long-range order particularly interest us when they have forbidden symmetries, that is, symmetries that are incompatible with a lattice structure.
For example the first quasicrystal observed \cite{shechtman1984} had 5-fold rotational symmetry.
A tiling has \emph{local n-fold rotational symmetry} when the tiling and its image by the rotation of angle $\tfrac{2\pi}{n}$ have the exact same finite patterns up to translations. A tiling has \emph{global n-fold rotational symmetry} when there exist a point such that the tiling is invariant under the rotation of angle $\tfrac{2\pi}{n}$ around that point.

Sometimes a tiling can be generated by a \emph{substitution}: a local inflation-subdivision rule which allows one to replace individually each tile by a set of tiles, so-called
\emph{metatile}, that has the same shape as the initial tile but is larger by a constant scaling factor. Iterating the substitution generates a tiling with a self-similar
structure in different scales.
In addition to providing a way to generate a tiling, substitutions give methods to study the properties of the tiling or a tiling space, making them attractive for mathematical analysis \cite{frank2008}.

The famous Penrose rhombus tilings have all the properties discussed above.
They are discrete plane -- even cut-and-project -- tilings projected from $\R{5}$, but they are also generated by
 a substitution. Penrose tilings are non-periodic: they have local 10-fold rotational symmetries incompatible with periodicity. Penrose tilings
 can be also enforced by local matching rules and, as discussed below, the primitivity of the generating substitution guarantees that Penrose tilings are quasiperiodic.

A natural question arises whether tilings with all the nice properties of Penrose tilings are possible for other orders of rotational symmetries instead of order 10.
In \cite{kari2016} Kari and Rissanen presented the Sub Rosa substitution tiling family. These tilings have $2n$-fold rotational symmetry but are in general not discrete plane tilings, as stated below and proved in Section \ref{sec:subrosa}.

\begin{theorem}
  \label{th:subrosa}
  The Sub Rosa tilings for odd $n>5$ are not discrete plane tilings.
\end{theorem}

The main result of our paper (Section \ref{sec:planar-rosa}) is the construction of a new Planar Rosa family of tilings which are defined by substitutions and have $2n$-fold rotational symmetries for
odd $n\geq 5$. The important new feature is that these tilings are discrete plane tilings.

\begin{theorem}
  \label{th:main}
  For any odd $n\geq 5$, the Planar Rosa $n$ tiling is a  discrete plane substitution tiling with global $2n$-fold rotational symmetry.
\end{theorem}

The article is organized as follows. In Section \ref{sec:planar} we introduce specific definitions for discrete plane substitution tilings.
To illustrate these definitions and to familiarize the reader with discrete plane substitutions we present in detail a 10-fold planar substitution tiling. This is the Sub Rosa tiling with $n=5$.
The two key elements to prove Theorems \ref{th:subrosa} and \ref{th:main} are a correspondence between the boundary of the metatiles and the planarity of the substitution, and the tileability of metatiles given their boundary.
The correspondence between the boundary of the metatiles and the planarity of the substitution is achieved using mainly linear algebra and
properties of circulant matrices as introduced in Section \ref{sec:planar} on the 10-fold example, and later proved in Sections \ref{sec:subrosa} and \ref{sec:planar-rosa}.
The tileability of the metatiles given their boundary is discussed in Section \ref{sec:tileability} where we adapt the Kenyon criterion for tileability of a polygon by parallelogram \cite{kenyon1993} to our specific case. The results of Section \ref{sec:tileability} are used in Section \ref{sec:planar-rosa} to prove the tileability of the Planar Rosa metatiles.

We point out that the relation between the substitution construction and the cut-and-project method has already been studied from other perspectives.
In \cite{harriss2004,harriss2004canonical} Harriss proved that canonical cut-and-project tilings are canonical substitution tilings only if the plane which is approximated is an eigenspace of a matrix with quadratic eigenvalues. This is very restrictive for the approximated plane of canonical cut-and-project substitution tilings. In particular, for cut-and-project substitution tilings with $n$-fold rotational symmetry for odd $n\geq 7$ the eigenvalues of a matrix that would admit the approximated plane as its eigenspace are of degree more than 2.

Another approach for cut-and-project substitution tilings are through \emph{generalized substitutions} as introduced in \cite{arnoux2001pisot}. These are rather well understood for the
case of codimension one, e.g., 2D projections of discrete planes in $\mathbb{R}^3$ \cite{fernique2006, jolivet2013}. For tilings with $n$-fold rotational symmetry we need the higher dimensional version defined in \cite{arnoux2001higher}. This formalism was used in \cite{arnoux2011} to study in detail tilings in the codimension two case. However these works were motivated by the efficient coding of dynamical systems and not by geometrical tilings so there were no considerations of rotational symmetry. Natural questions would be whether this formalism can be used to define a family of tilings with $n$-fold rotational symmetry, and if so, how this approach would compare to the method we present here. Possibly this formalism could provide substitutions that have smaller scaling factors but have more complicated shapes of metatiles.

Let us also mention the recent work on generalized self-similarities of cut-and-project sets \cite{masakova2019generalized} where the authors consider the linear maps that preserve cut-and-project sets. They face much the same linear algebra questions as we do, but do not consider the issue of tileability.

\section{Settings}
\label{sec:settings}
\paragraph{Rhombus Tiling.}
Let $\vec{v}_0,\dots,\vec{v}_{n-1}$ be $n$ pairwise non-collinear unit vectors of the Euclidean plane. We call these \emph{edge directions}.
In this work we only consider the case were $\vec{v}_0,\dots,\vec{v}_{n-1}$ are the $n$-th roots of unity, \ie, $\vec{v}_k=e^{\imag\frac{2k\pi}{n}}$.
Here, and frequently in the rest of the paper, we identify the real plane $\mathbb{R}^2$ and the complex plane $\mathbb{C}$ in the standard manner. Remark that we write all vectors
as row vectors, and we use the transpose operator if column vectors are needed.

The $n$ edge directions define $\binom{n}{2}$ rhombus \emph{prototiles} which we denote by $T_{j,k}$
for $0\leq j<k< n$. We then denote by $\mathbf{T}$ the set of the prototiles and call it a \emph{tileset}:
$$T_{j,k} := \{ \lambda \vec{v}_j + \mu \vec{v}_k\ |\ 0\leq \lambda,\mu\leq 1\}, \qquad \mathbf{T} := \{T_{j,k}\ |\ 0\leq j < k < n\}.$$
We call \emph{$\mathbf{T}$-tiling} an edge-to-edge tiling of the plane where the tiles are translates of the prototiles $T_{j,k}$.
Recall that a tiling is a covering of the plane with no overlap and that edge-to-edge means that any two tiles of the tiling either intersect on a full common edge, on a single common vertex, or not at all.
\begin{figure}[b]
  \center  \includegraphics[width=10cm]{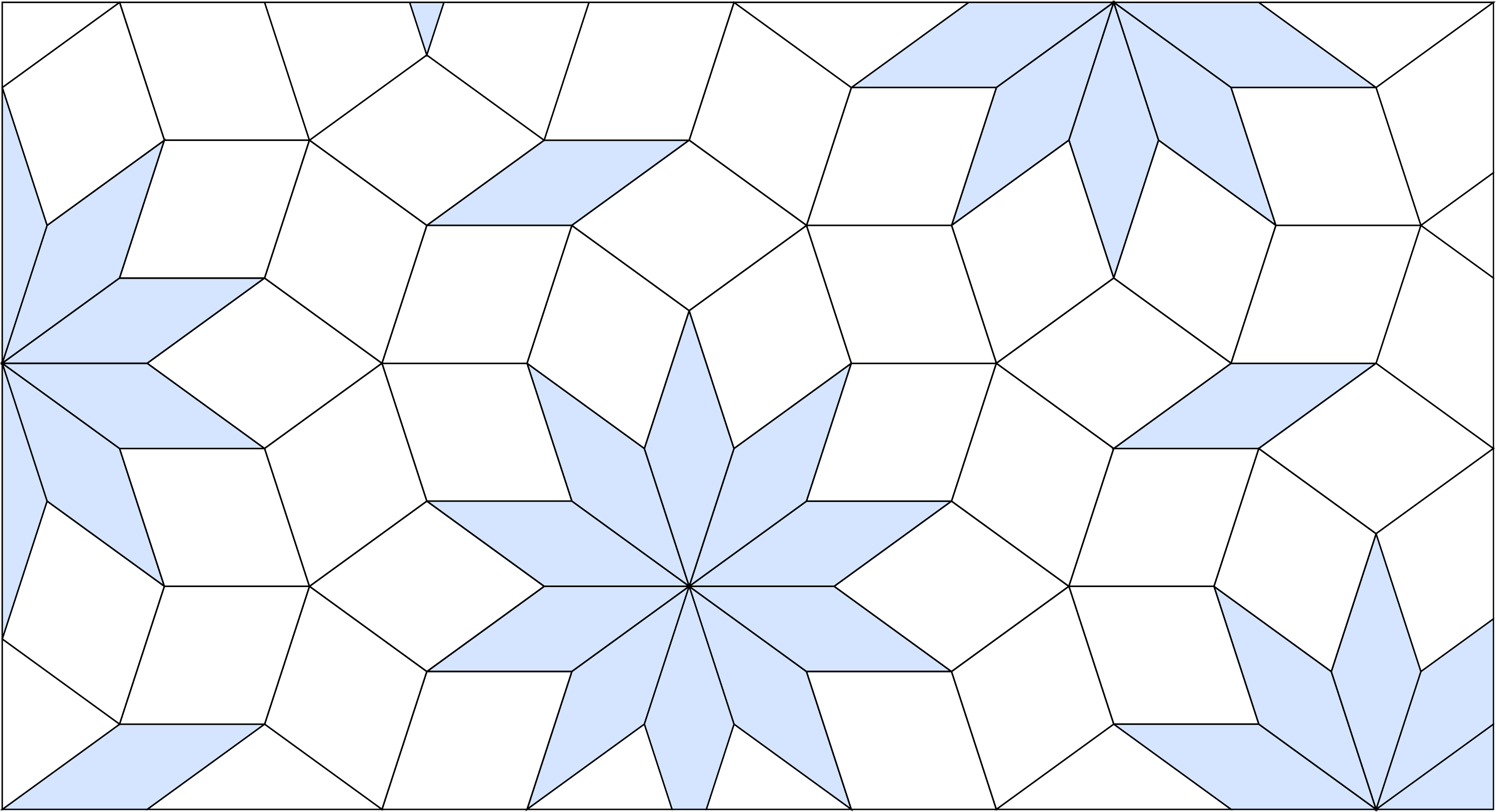}
  \caption{Example of rhombus tiling.}
  \label{fig:tilingexample}
\end{figure}

For example, take $n=5$. We have  $$\vec{v}_k:=\left( \cos\frac{2k\pi}{5}, \sin\frac{2k\pi}{5}\right) = e^{\imag\frac{2k\pi}{n}} \text{, for } k\in\{0,1,2,3,4\}.$$ We then have $\binom{5}{2} = 10$ rhombuses, but up to $2n$-fold rotations and translations we only have two rhombuses: a narrow one with angles $\frac{\pi}{5}$ and $\frac{4\pi}{5}$ (in blue/shaded in Figure \ref{fig:tilingexample}), and a wide one with angles $\frac{2\pi}{5}$ and $\frac{3\pi}{5}$ (in white in Figure \ref{fig:tilingexample}).

We call a \emph{$\mathbf{T}$-patch} a simply-connected edge-to-edge set of tiles which are translates of the prototiles of $\mathbf{T}$. We may also call this simply a \emph{patch} if the tileset $\mathbf{T}$ is known. We denote by $V(P)$ the set of vertices of a patch $P$. We call a \emph{pattern} a patch up to translations. We say that a pattern $\mathcal{P}$ appears in a tiling if a subset of the tiling is a patch in $\mathcal{P}$.

A tiling is called \emph{uniformly repetitive} or \emph{uniformly recurrent} when, for any pattern that appears in the tiling, there exists a radius $r$ such that in any disk of radius $r$ in the tiling this pattern appears.

A tiling is called \emph{periodic} if there exists a non-trivial translation for which it is invariant, and \emph{non-periodic} if there exists no such non-trivial translation. We call a tiling \emph{quasiperiodic} if it is both non-periodic and uniformly recurrent.

\begin{figure}[!b]
  \includegraphics[width=\textwidth]{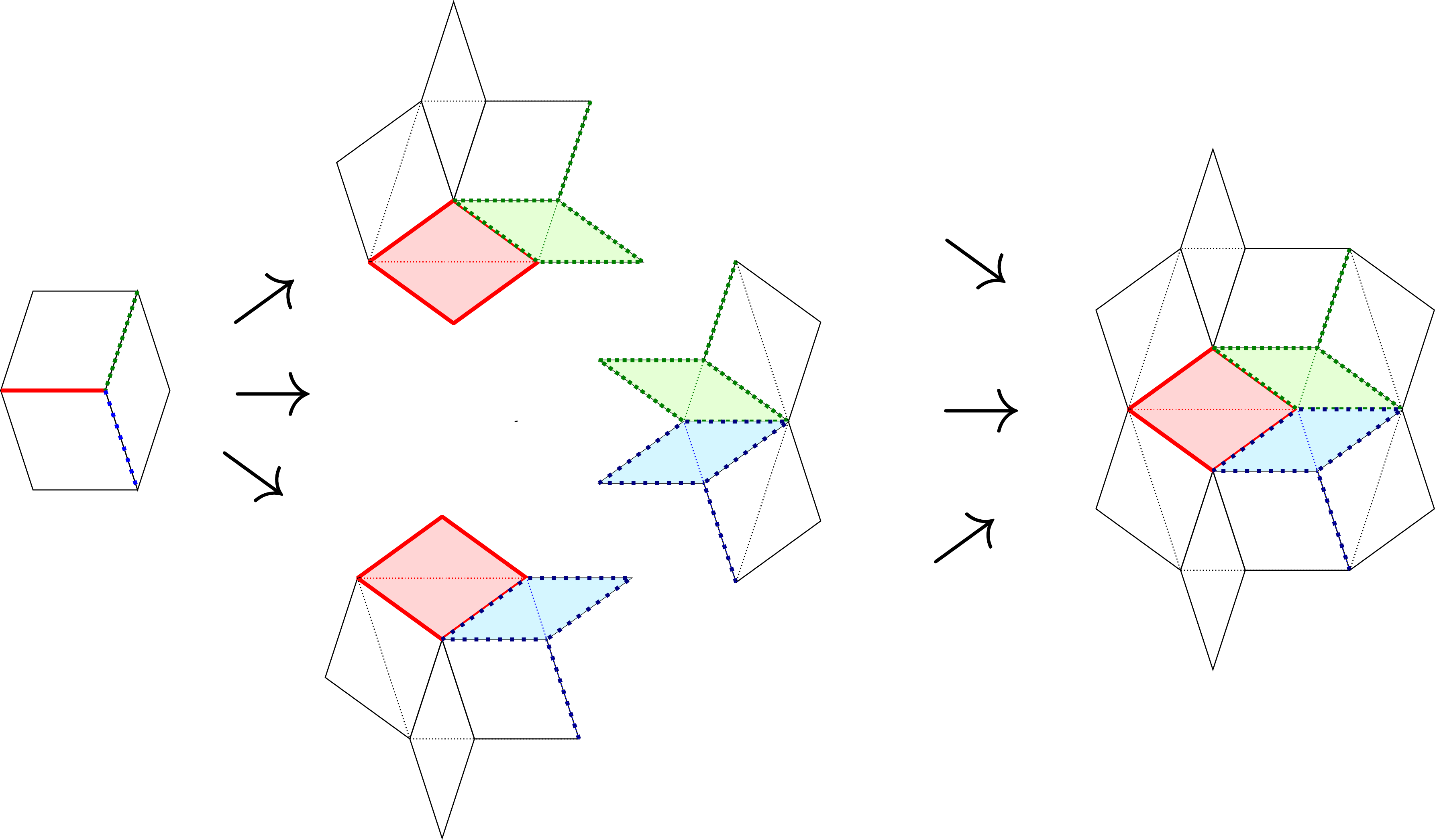}
  \caption{An example of combinatorial substitution. On the left the initial patch with the three internal edges coloured. In the middle the images of the three initial tiles with the images of the internal edges coloured. On the right the patch obtained by gluing the images of the three tiles along the images of the shared edges.}
  \label{fig:combinatorialsubstitution}
\end{figure}

\paragraph{Substitution.}
A \emph{combinatorial substitution} on a tileset $\mathbf{T}$ is a pair of functions $(\sigma,\partial\sigma)$ where $\sigma$, called a \emph{substitution}, is a function that associates to each prototile $t$ a finite patch of tiles $\sigma(t)$, and $\partial\sigma$, called the \emph{boundary of the substitution}, is a function that associates a set of external edges and/or tiles of $\sigma(t)$ to each pair $(t,e)$ where $t$ is a prototile and $e$ is an edge of $t$.

The substitution $\sigma$ is extended to a function on patches of tiles by applying the substitution separately to each tile and gluing the obtained patches in such a way that it preserves the combinatorial structure. In other words, for any two tiles $t_0$ and $t_1$ of a patch $P$ that are adjacent along an edge $e$, the sets $\partial\sigma(t_0,e)$ and $\partial\sigma(t_1,e)$ have to be equal, and in the patch $\sigma(P)$ the two patches $\sigma(t_0)$ and $\sigma(t_1)$ are glued along the set of edges and/or tiles in $\partial\sigma(t_0,e)=\partial\sigma(t_1,e)$. See Figure \ref{fig:combinatorialsubstitution} for an illustration.

We call \emph{metatiles of order $k$} of $\sigma$ the patterns $\sigma^k(t)$ for $t\in\mathbf{T}$. We simply call \emph{metatiles} the first order metatiles.
We say that a combinatorial substitution $(\sigma,\partial\sigma)$ is \emph{well-defined} when the substitution $\sigma$ can be applied on all metatiles of all orders, \ie, there is no metatile $\sigma^k(t)$ on which $(\sigma,\boundary)$ cannot be applied in a way that respects the combinatorial structure. In the following we always assume that the substitutions are well defined and we usually omit the boundary $\partial\sigma$ from the notation.

A finite pattern is said to be \emph{legal} for a substitution $\sigma$ if it appears in some $\sigma^k(t)$ with $t\in\tileset$ and $k\in\mathbb{N}$. A tiling $\tiling$ is said \emph{legal} for $\sigma$ if every finite pattern of $\tiling$ is legal for $\sigma$.

Let us now define two families of substitutions that are much easier to work with.
An \emph{edge-hierarchic substitution} (or stone substitution) is a combinatorial substitution $(\sigma,\partial\sigma)$ such that there exists an expansion $\phi$ (orientation preserving expanding similitude of the plane) such that for any tile $t$ and any edge $e$ of $t$, the metatile $\sigma(t)$ spans exactly the expanded tile $\phi(t)$ and the image $\boundary(t,e)$  of an edge $e$ is exactly the expanded edge $\phi(e)$, \ie,
$$ \bigcup\limits_{t' \in \sigma(t)} t' = \phi(t)\qquad \bigcup\limits_{e' \in \partial\sigma(t,e)} e' = \phi(e).$$
Note that this implies that $\partial\sigma(t,e)$ only contains edges for any $t$ and $e$. For an example, see the Chair substitution in Figure \ref{fig:substitutionexample}.

\begin{figure}[b]
  \center  \includegraphics[width=5cm]{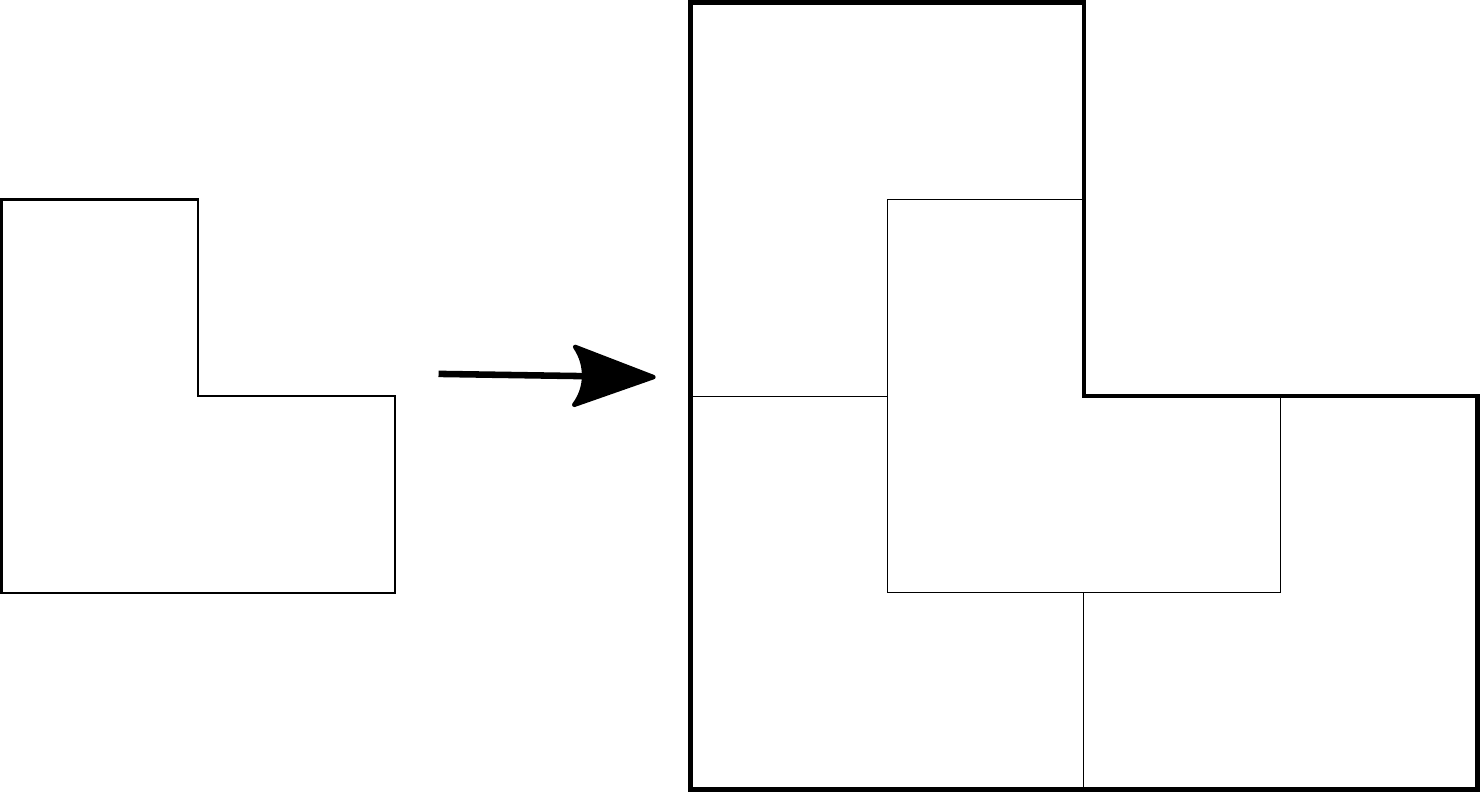}
  \caption{The Chair substitution, an example of an edge-hierarchic substitution on only one prototile up to translations and rotations.}
  \label{fig:substitutionexample}
\end{figure}

A \emph{vertex-hierarchic substitution} is a combinatorial substitution $(\sigma,\partial\sigma)$ such that there exists an expansion $\phi$ (orientation preserving expanding similitude of the plane) such that for any tile $t$ and any edge $e$ of $t$, the area spanned by $\sigma(t)$ is equal to the area of the expanded tile $\phi(t)$ and the vertices of the expanded edge $\phi(e)$ are vertices of the boundary $\boundary(t,e)$, \ie,
$$ Area\left(\bigcup\limits_{t' \in \sigma(t)} t'\right) = Area\left(\phi(t)\right),\qquad V\left(\phi(e)\right)\subset V\left(\bigcup\limits_{x \in \partial\sigma(t,e)} x\right). $$
If $\partial\sigma(t,e)$ contains tiles we take the convention that these boundary tiles count only for half in the computation of the area. See for example the Penrose substitution in Figure \ref{fig:penrosesubst} where the expanded tiles are represented in thick lines.
\begin{figure}[t]
  \center  \includegraphics[width=\textwidth]{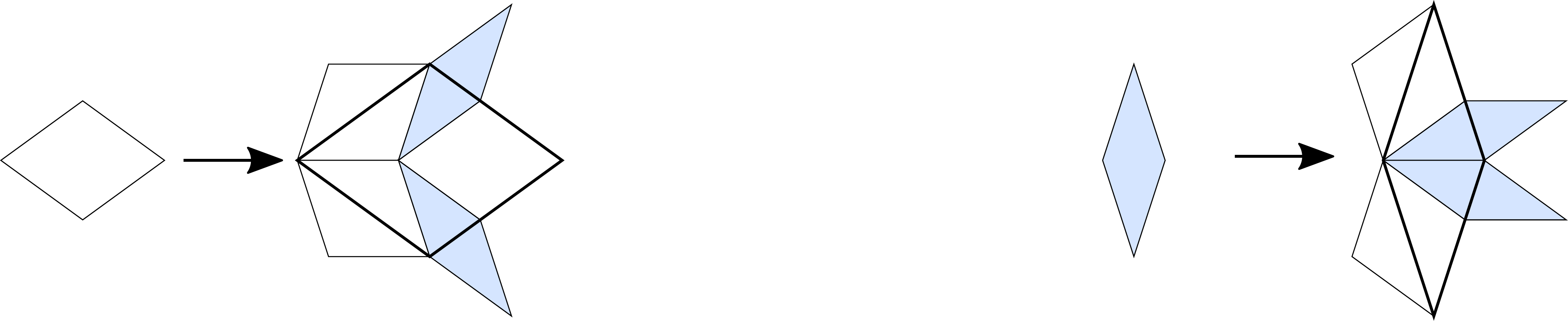}
  \caption{The Penrose substitution, an example of a vertex-hierarchic substitution on two rhombus tiles up to translations and rotations.}
  \label{fig:penrosesubst}
\end{figure}

The idea of these two families of substitutions is that they are both an inflation-subdivision process meaning that the substitution first inflates a prototile $t$ to $\phi(t)$ and then subdivides $\phi(t)$ to obtain a patch of tiles.
In the case of an edge-hierarchic substitution the subdivision is exact. In the case of a vertex-hierarchic substitution the subdivision can differ from $\phi(t)$ but the vertices of $\phi(t)$ must be boundary vertices of $\sigma(t)$, and the area of $\phi(t)$ and $\sigma(t)$ must be equal.
Note that substitution tilings are also sometimes called inflation tilings, or self-similar tilings.

\begin{figure}[b]
  \center  \includegraphics[width=10cm]{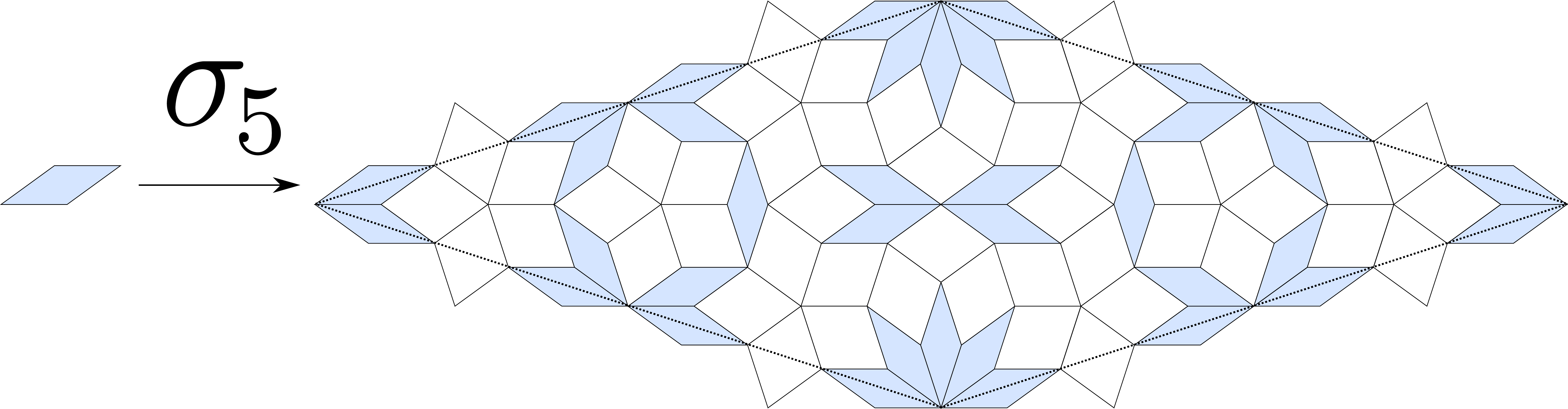}
  \caption{The Sub Rosa 5 substitution on the $\tfrac{\pi}{5}$ rhombus, an example of a vertex-hierarchic substitution where the images of all edges are identical up to rotations and translations. The edgeword is 131131 where 1 and 3  indicate the narrow and the wide rhombuses, respectively.}
  \label{fig:mysubst}
\end{figure}

In this paper we study a specific case of vertex-hierarchic substitutions where the image $\partial\sigma(t,e)$ of all edges $e$ of all tiles $t$ is the same up to rotations and translations. See for example the Sub Rosa 5 substitution in Figure \ref{fig:mysubst}. This choice reduces greatly the possibilities for the shape of the boundary. For example the well-known Penrose substitution \cite{penrose1974} and the Ammann-Beenker substitution \cite{beenker1982, grunbaum1987} are not in this class, but it simplifies conditions on the substitution, ensure the well-definedness of the substitution, and makes it easy to lift the substitution in $\mathbb{R}^n$. In this case we call the \emph{edgeword} the sequence of rhombuses and/or edges in the image of an edge up to translations and rotations. Our edgewords are always palindromes.

A substitution is called \emph{primitive} when there exists a $k$ such that for every prototile $t$, the patch $\sigma^k(t)$ contains all the different prototiles of the tileset.
\begin{proposition}
  If a tiling $\tiling$ is legal for some primitive substitution $\sigma$ then $\tiling$ is uniformly recurrent.
\end{proposition}

This result is well-known and can be found in \cite[\S 5]{baake2013}. We use this result to prove that the substitution tilings we consider in Sections \ref{sec:planar} to \ref{sec:planar-rosa} are quasiperiodic.

\paragraph{Lifting to $\mathbb{R}^n$.}
Let $\vec{e}_0,\dots \vec{e}_{n-1}$ be the standard basis of $\mathbb{R}^n$.
We define the \emph{integer square} $S_{x,j,k}$ where $x\in\mathbb{Z}^n$ and $j,k\in\{0,.. n-1\}$ as  $$S_{x,j,k} := \{x+ \lambda \vec{e}_j + \mu \vec{e}_k\ |\ 0\leq \lambda, \mu \leq 1\}.$$
\begin{figure}[b]
  \center  \includegraphics[width=10cm]{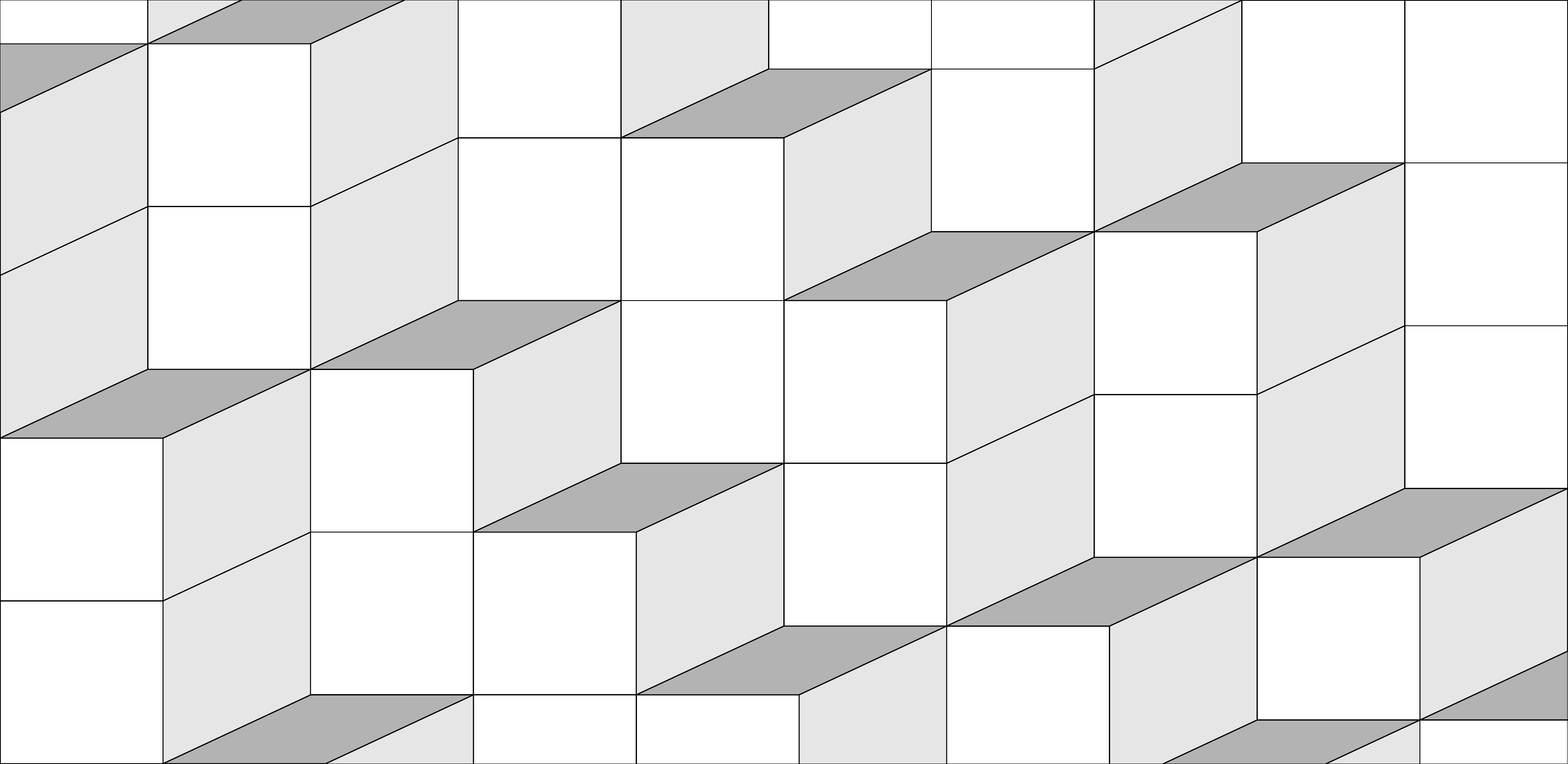}
  \caption{Example of an edge-to-edge rhombus tiling with 3 edge directions. We can intuitively see its lifted discrete surface in $\mathbb{R}^3$.}
\end{figure}

A \emph{discrete surface} is a simply-connected set of integer squares such that no more than two squares intersect on any edge. Discrete surfaces are sometimes called ``stepped surfaces''.
A $\mathbf{T}$-tiling is \emph{lifted} in $\mathbb{R}^n$ as follows:
\begin{itemize}
\item take an arbitrary \emph{origin} vertex in the tiling and map it to $0 \in \mathbb{R}^n$.
\item each tile of type $T_{j,k}$ is mapped to a some $S_{x,j,k}$ such that if two tiles share an edge $\vec{v}_j$ (resp. a vertex) then their image will share an edge $\vec{e}_j$ (resp. a vertex).
\end{itemize}
This lift operation was introduced by Levitov in \cite{levitov1988}. It lifts any edge-to-edge rhombus tiling $\mathcal{T}$ with $n$ edge directions to a discrete surface $\widehat{\mathcal{T}}$ in $\mathbb{R}^n$ that is unique up to the choice of the origin vertex. The reason why the tiling $\mathcal{T}$ can be lifted and that it is unique up to the choice of the origin vertex is that in an edge-to-edge rhombus tiling any two edge paths from the origin to a vertex $x$ are identical up to reordering and cancellation (\ie, $\vec{v}_i -\vec{v}_i = 0$). In particular, there is a unique abelianized path from the origin to $x$~\cite{levitov1988}. The lifted vertex $\hat{x}$ is characterized uniquely by
\begin{align*}
  \hat{x} &= (k_0,\dots k_{n-1}) \in \mathbb{Z}^n \Leftrightarrow x=\sum\limits_{0\leq i < n} k_i \vec{v}_i.
\end{align*}
Note that the lift of the tile $\tile{j}{k}$ at position $x$ is the square $\intsquare{\hat{x}}{j}{k}$ \emph{i.e. }
$$ \widehat{x+\tile{j}{k}}:= \intsquare{\hat{x}}{j}{k}.$$

We mostly use the same names and symbols for the objects in the Euclidean plane and their lifted counterparts in $\mathbb{R}^n$. However, when we want to emphasize the difference we denote $\hat{x}$ for the lifted version of an object $x$.

\paragraph{Discrete plane.}
We call a \emph{discrete plane tiling}, or simply a \emph{discrete plane}, an edge-to-edge rhombus tiling $\tiling$ with $n$ edge directions (or a discrete surface of $\R{n}$) such that there exists a 2D-plane $\slope$ of $\mathbb{R}^n$ called \emph{slope} such that the lifted tiling $\widehat{\tiling}$ stays within bounded distance of $\slope$ in $\mathbb{R}^n$, \ie, there exists $\delta\in\mathbb{R}^+$ such that $d(\slope, \hat{\tiling})\leq \delta$. The smallest such $\delta$ is called \emph{thickness} of the tiling.

We call the \emph{window} of a discrete plane $\tiling$, denoted by $W$, the orthogonal projection of the vertex set $V(\tiling)$ onto $\slope^\bot$
where, as usual, $\slope^\bot$ denotes the orthogonal complement of $\slope$ in $\mathbb{R}^n$. Note that for cut-and-projection there are additional conditions on the window.

Note that \emph{discrete planes} are sometimes called \emph{planar tilings} \cite{bedaride2015}, and that in the scope of tilings with local matching conditions they are called \emph{weak local rules tilings} \cite{levitov1988} or \emph{weak matching rules tiling} \cite{socolar1990}.

\paragraph{$n$-fold symmetric tilings.}
A $\mathbf{T}$-tiling has \emph{local $n$-fold rotational symmetry} if, for every pattern $P$ of the tiling, the image of $P$ by the rotation of angle $\frac{2\pi}{n}$ is also a pattern of the tiling (see the patches in Figure \ref{fig:localsym}).
A $\mathbf{T}$-tiling $\mathcal{T}$ has \emph{global $n$-fold rotational symmetry} if there exist a point $p$ such that $\mathcal{T}$ is invariant under the rotation of angle $\frac{2\pi}{n}$ around center $p$.

\begin{figure}[t]
  \center  \includegraphics[width=\textwidth]{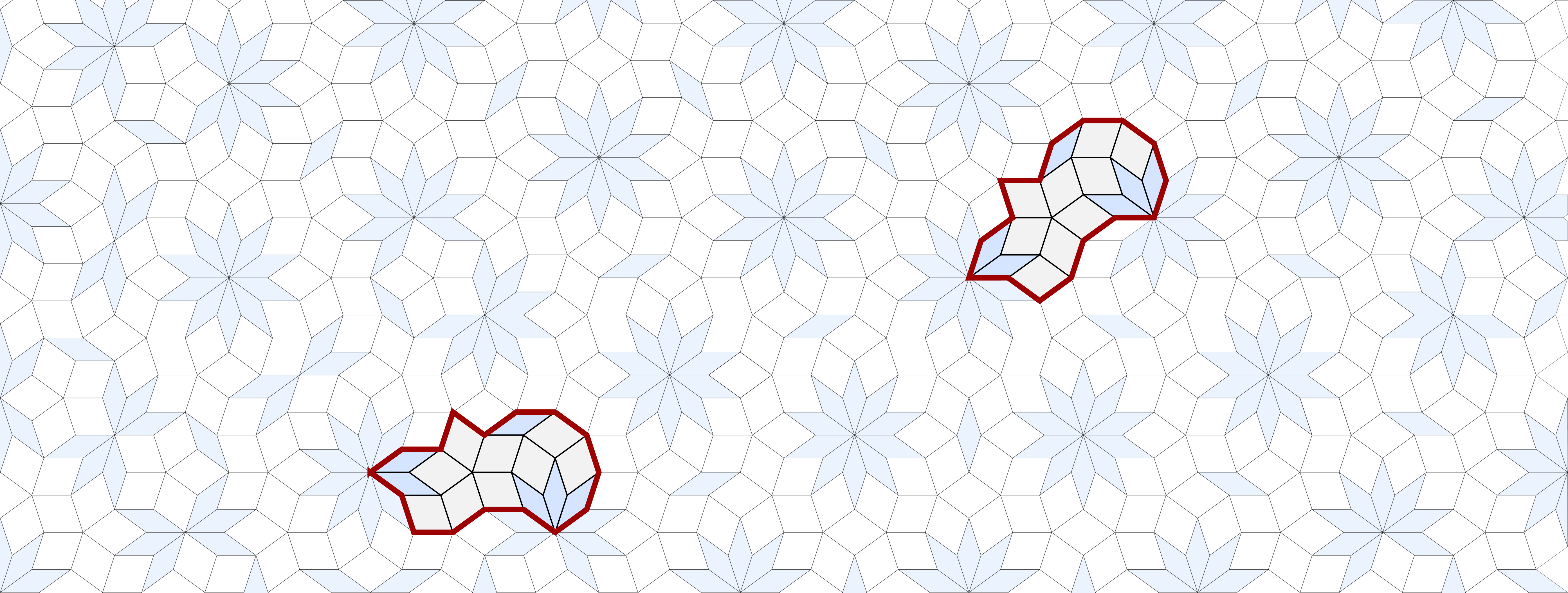}
  \caption{Example of a patch of a tiling with a 10-fold local symmetry. In the bottom left a patch has been highlighted, and in the top right a second patch appears which is the same up to a translation and a $\tfrac{2\pi}{10}$ rotation.}
  \label{fig:localsym}
\end{figure}

These two similar looking notions are, in fact, quite different. For example the canonical Penrose rhombus tilings has a local 10-fold symmetry and ``only'' a global 5-fold symmetry.
Note that in the context of quasicrystal tilings, cut-and-project or discrete planes tilings with $n$ edge directions (and usually with global $n$-fold or $2n$-fold rotational symmetry) are sometimes just called $n$-fold tilings.



\section{Substitution discrete planes}
\label{sec:planar}
In this Section we define how the substitutions we use are lifted to $\R{n}$ in order to study tilings that are both substitution tilings and discrete planes. To illustrate this we first present in detail an example of a substitution discrete plane with a 10-fold rotational symmetry.

Let us take a vertex-hierarchic substitution $\sigma$ and its expansion $\phi$ such that the image by the substitution of any two parallel edges $e$ and $e'$ is the same up to a translation, \emph{i.e.}, for any two tiles $t$ and $t'$, and for any edge $e$ of $t$ and any edge $e'$ of $t'$ we have
$$ e \parallel e' \Rightarrow \partial\sigma(t,e) \equiv \partial\sigma(t',e').$$

We define the lifted substitution $\hat{\sigma}$ from discrete surfaces (of $\R{n}$) to discrete surfaces (of $\R{n}$), and the lifted expansion $\hat{\phi}$ from $\R{n}$ to $\R{n}$ as follows.
\begin{itemize}
  \item Let us first define the lifted substitution $\hat{\sigma}$ on the prototiles as
$$\hat{\sigma}(\hat{r}):=\widehat{\sigma(r)},$$
    for any rhombus prototile $r$. Since the prototiles are the tiles $\tile{j}{k}$, this defines $\hat{\sigma}$ for the integer squares $\intsquare{0}{j}{k}$.
  \item Let us next define the lifted expansion $\hat{\phi}$ so that we can define the lifted substitution on all integer squares, at all positions.
    We define $\hat{\phi}$ as a linear function of $\R{n}$, so we only need to define it on the canonical basis of $\R{n}$.
    For $i\in{0,..,n-1}$ we define $\hat{\phi}(\vec{e}_i)$ as
    $$\hat{\phi}(\vec{e}_i) := \widehat{\phi(\vec{v}_i)},$$
    where $\widehat{\phi(\vec{v}_i)}$ is the lifted version of the abelianized path from the origin to the $\vec{v}_i$ corner vertex in the patch $\sigma(\tile{i}{j})$ for some $j\neq i$. Note that $\phi(\vec{v}_i)$ is uniquely defined due to the condition that the image by the substitution of any two parallel edges is the same up to a translation.
  \item Let us finally define the lifted substitution $\hat{\sigma}$ on any integer square as
    $$ \hat{\sigma}(\intsquare{x}{i}{j}) := \hat{\phi}(x) + \hat{\sigma}(\intsquare{0}{i}{j}) = \hat{\phi}(x) + \widehat{\sigma(\tile{i}{j})},$$ and extend it on any discrete surface by linearity.
\end{itemize}

Note that as $\hat{\phi}$ is defined as a linear function of $\R{n}$ we quite often consider the expansion of a tile that is an integer square. In that case the tile is considered as a set of points of $\R{n}$ and its image is the set of images of its points. Note also that the image by $\hat{\phi}$ of an integer square is a rhombus of $\R{n}$ with vertices in $\mathbb{Z}^n$.
Remark that the image by $\hat{\sigma}$ of any integer square $\intsquare{0}{i}{j}$ has the same shape as $\hat{\phi}(\intsquare{0}{i}{j})$ in the sense that the vertices of $\hat{\phi}(\intsquare{0}{i}{j})$ are extremal boundary vertices of $\hat{\sigma}(\intsquare{0}{i}{j})$.

\begin{figure}[h]
  \center
  \includegraphics[width=5cm]{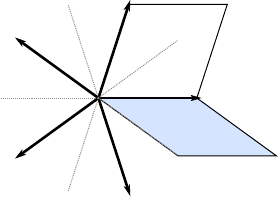}
  \caption{Vectors and rhombuses for $n=5$.}
  \label{fig:vectrhomb5}
\end{figure}

Since $\sigma$ and $\hat{\sigma}$ have exactly the same behaviour, we will write $\sigma$ for both, and the context makes it clear if we are considering the $\R{n}$ versions or the $\R{2}$ versions.

To study the behaviour of the linear function $\phi$ in $\R{n}$ we consider its matrix,
which we denote by $M_\phi$ or simply $M$. The matrix operates on column vectors from the right, \ie, $\phi(\vec{r})\trans = M_\phi\cdot \vec{r}\trans$.
We use the notation $\phi$ when we apply the expansion as a function, and the notation $M_\phi$ or $M$ when we study it as a matrix.

To ensure that a substitution tiling is planar along the plane $\mathcal{E}$, we want $\phi$ to be expanding along the plane $\mathcal{E}$ and non-expanding along $\mathcal{E}^\bot$. Indeed, $\phi$ expanding along $\mathcal{E}$ is necessary so that the substitution on the plane has a scaling factor greater than one, and if $\phi$ were expanding along $\mathcal{E}^\bot$ then iterating $\sigma$ on any initial patch of tiles would go farther and farther away from the plane $\mathcal{E}$. This comes from the fact that our substitutions are vertex-hierarchic.

 In this work we only study substitutions which are strictly expanding along $\mathcal{E}$ and strictly contracting along $\mathcal{E}^\bot$. This condition, though not necessary, is sufficient to ensure planarity, \emph{i.e.}, to ensure that the tilings legal for this substitution are discrete planes. This is stated and proved in
 Proposition \ref{prop:eigenvalue_planarity}.

To illustrate these definitions and the ideas, let us consider a specific substitution tiling with 10-fold global and local symmetry: the Sub Rosa 5 substitution. The general case for odd dimension $n$ and $2n$-fold rotational symmetry is quite similar, but the case  $n=5$ will allow for nicer pictures and easier notations while illustrating all the important ideas.

We take $n=5$ and the unit vectors $\vec{v}_0,\vec{v}_1,\vec{v}_2,\vec{v}_3,\vec{v}_4$ where
$$\vec{v}_j = \left(\cos\frac{2j\pi}{5}, \sin\frac{2j\pi}{5}\right)= e^{\imag\frac{2j\pi}{n}}.$$
We define the tiles $\tile{i}{j}$ and the tileset $\tileset$ as presented in Section \ref{sec:settings} and the $\tileset$-tilings are then naturally lifted in $\R{5}$ by the lifting operator presented in Section \ref{sec:settings}.

Since we impose that the boundary substitutions are all identical up to rotation, the effects of the substitution on the vectors of the canonical basis are also identical up to rotation. This implies that the expansion is a cyclic function, \emph{i.e.}, a linear function that commutes with a cyclic permutation of the basis vectors. So the expansion can be written as a circulant matrix \[M_\phi = \circulant{(m_0,m_1,m_2,m_3,m_4)} = \begin{pmatrix}
    m_0 & m_1 & m_2 & m_3 & m_4 \\
    m_4 & m_0 & m_1 & m_2 & m_3 \\
    m_3 & m_4 & m_0 & m_1 & m_2 \\
    m_2 & m_3 & m_4 & m_0 & m_1 \\
    m_1 & m_2 & m_3 & m_4 & m_0 \\
\end{pmatrix}\]

The eigenvalues and eigenvectors of circulant matrices are well-known. See for example \cite{Davis1979}. All the substitutions studied in the later parts
this of work also have circulant expansion matrices, so we state the following classical result in a  general form. We denote as $\circulant{(m_0,m_1,\dots, m_{n-1})}$
the circulant $n\times n$ matrix whose first row is $(m_0,m_1,\dots,m_{n-1})$.

\begin{proposition}
\label{prop:circulant}
Let $\zeta\in\mathbb{C}$ be an $n$'th root of unity, \ie, $\zeta^n=1$. For any circulant $n\times n$ matrix $M = \circulant{(m_0,m_1,\dots, m_{n-1})}$  holds
$$
M\cdot (1,\zeta,\dots,\zeta^{n-1})\trans =  \left(\sum_{j=0}^{n-1} m_j\zeta^j \right)\cdot (1,\zeta,\dots,\zeta^{n-1})\trans,
$$
and
$$
(1,\zeta,\dots,\zeta^{n-1})\cdot M =  \left(\sum_{j=0}^{n-1} m_j\zeta^{n-j} \right)\cdot (1,\zeta,\dots,\zeta^{n-1}).
$$
\end{proposition}
\noindent
In particular, note that all $n\times n$ circular matrices have the
common complex eigenvectors $(1,\zeta,\dots,\zeta^{n-1})\trans$ for all $n$'th roots $\zeta$ of unity.

Returning to the specific case of $n=5$,
we define a decomposition of $\R{5}$ into a line and two planes as follows:
\begin{itemize}
\item $\Delta := \left\langle\left(1,1,1,1,1\right)\right\rangle$
\item $\e{0} := \left\langle \left(\cos\frac{2k\pi}{5}\right)_{k=0..4}, \left(\sin\frac{2k\pi}{5}\right)_{k=0..4} \right\rangle$
\item $\e{1} :=  \left\langle \left(\cos\frac{6k\pi}{5}\right)_{k=0..4}, \left(\sin\frac{6k\pi}{5}\right)_{k=0..4} \right\rangle$
\end{itemize}
Here, and elsewhere, we denote by  $\langle\cdot \rangle$ the subspace generated by given vectors.
Note that
$\Delta,\ \e{0}$ and $\e{1}$ are orthogonal and $\Delta \oplus \e{0} \oplus \e{1} = \R{5}$.
For convenience, e.g., to use Proposition~\ref{prop:circulant} we rephrase the real two-dimensional subspaces $\e{0}$ and $\e{1}$ as
one-dimensional complex subspaces
$$ \e{0} := \left\langle \left(e^{\imag\frac{2k\pi}{5}}\right)_{k=0..4}\right\rangle,\qquad \e{1} := \left\langle \left(e^{\imag\frac{6k\pi}{5}}\right)_{k=0..4}\right\rangle$$
so that the vectors of real planes are precisely the vectors formed by the real and the imaginary parts of the complex vectors.
Now we can deduce from Proposition~\ref{prop:circulant} that $\Delta$, $\e{0}$ and $\e{1}$ are one-dimensional (complex) eigenspaces of any
$5\times 5$ circulant matrix $M_\phi$,
which means that as real planes  $\e{0}$ and $\e{1}$ are invariant spaces of dimension two.

Let us define the Sub Rosa 5 substitution $\sigma_5$ by Figure \ref{fig:sigmafive}.
\begin{figure}[b]
  \center \includegraphics[width=\textwidth]{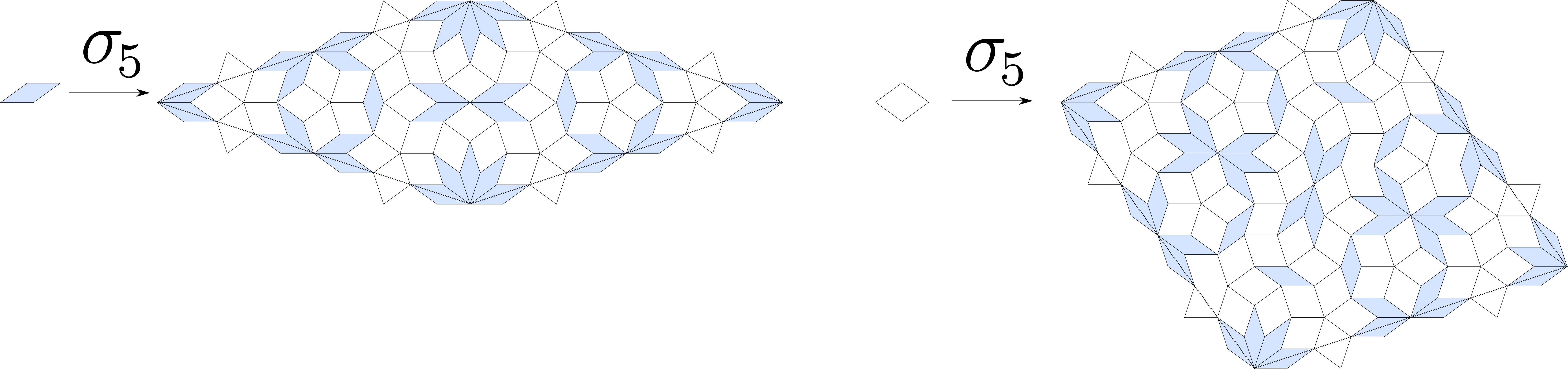}
  \caption{The Sub Rosa 5 substitution $\sigma_5$ up to translation and rotation.}
  \label{fig:sigmafive}
\end{figure}
The expansion matrix $M_5$ of $\sigma_5$ is
\[M_5= \begin{pmatrix}
  4 & 2 & -2 & -4 & 0 \\
  0 & 4 & 2 & -2 & -4 \\
  -4 & 0 & 4 & 2 & -2 \\
  -2 & -4 & 0 & 4 & 2 \\
  2 & -2 & -4 & 0 & 4 \\
  \end{pmatrix}\\
  .\]
The eigenvalues of $M_5$, corresponding to the complex eigenspaces $\Delta$, $\e{0}$ and $\e{1}$ as provided by Proposition~\ref{prop:circulant}, are
\begin{equation}
\label{eq:eigenvalues5}
\begin{aligned}
  \lambda_\Delta &= 0\\
  \lambda_0 &= \left(8\cos\frac{\pi}{10}+4\cos\frac{3\pi}{10}\right)e^{-\imag\frac{\pi}{10}}\\
  \lambda_1 &= \left( 8\cos\frac{3\pi}{10} - 4\cos\frac{\pi}{10} \right)e^{-\imag\frac{3\pi}{10}}.
\end{aligned}
\end{equation}
Note that on the plane $\e{0}$ the expansion is the multiplication by $\lambda_0$, which means that the scaling factor of the expansion on this \emph{tiling plane} is $|\lambda_0|$,
where we use the standard notation $|\cdot |$ for the modulus of a complex number.

The edgeword of $\sigma_5$ is $131131$. Recall that the edgeword is the sequence of rhombuses that appears on the image of every edge by the substitution. Here we encode the narrow rhombus which has angle $\tfrac{\pi}{5}$ by symbol $1$ and the wide rhombus which has angle $\tfrac{3\pi}{5}$ by symbol $3$.

Let us now consider the more general setting where we have a similar vertex-hierarchic substitution $\sigma$ on the same tileset and with an edgeword $u$. We assume that $u$ is a palindromic word on alphabet $\{1,3\}$ and that $\sigma$ is a substitution such that the image of any edge is the succession of rhombuses coded by $u$. This substitution is lifted in $\R{5}$ in the same way as $\sigma_5$ and its expansion is a linear function that admits $\Delta$, $\e{0}$ and $\e{1}$ as eigenspaces with some eigenvalues $\lambda_\Delta$, $\lambda_0$ and $\lambda_1$.

We define the \emph{abelianized edgeword}  $[u]$ as $$[u]:= (|u|_1, |u|_3),$$
 where $|\cdot |_x$ denotes the number of occurrences of letter $x$ in a word. Remark that the expansion $\phi$ and its matrix $M$ only depend on the abelianized edgeword. More precisely, we can decompose $M$ as a linear combination of two \emph{elementary matrices} $M_0(5)$ and $M_1(5)$ with
\[M_0(5):= \begin{pmatrix}
  1 & 0 & 0 & -1 & 0 \\
  0 & 1 & 0 & 0 & -1 \\
  -1 & 0 & 1 & 0 & 0 \\
  0 & -1 & 0 & 1 & 0 \\
  0 & 0 & -1 & 0 & 1 \\
\end{pmatrix}\  M_1(5):= \begin{pmatrix}
  0 & 1 & -1 & 0 & 0 \\
  0 & 0 & 1 & -1 & 0 \\
  0 & 0 & 0 & 1 & -1 \\
  -1 & 0 & 0 & 0 & 1 \\
  1 & -1 & 0 & 0 & 0 \\
\end{pmatrix}. \]
The matrix $M_0(5)$ is the expansion matrix of a substitution with edgeword $1$, \emph{i.e.}, with only a narrow rhombus on every side of every metatile. Similarly, $M_1(5)$ is the expansion matrix of a substitution with edgeword $3$. We have
$$ M = [u]_0M_0(5) + [u]_1M_1(5). $$
This decomposition gives us a formula to easily compute the eigenvalues of the expansion of the substitution associated to any edgeword $u$.
Note that the order of the tiles in the edgeword does not influence the expansion matrix $M$ as it is determined by the abelianized edgeword. Therefore also the eigenvalues of the
expansion are indifferent to the order of letters in $u$.

\begin{proposition}
  \label{prop:eigenvalues_5fold}
  Let $u$ be a palindromic word on alphabet $\{1,3\}$. Let $\sigma$ be a substitution with the edgeword $u$ and the corresponding expansion $\phi$.
  Let $\lambda_\Delta$, $\lambda_0$ and $\lambda_1$ be the eigenvalues of $\phi$ on eigenspaces $\Delta$, $\e{0}$ and $\e{1}$, respectively.

  We have $\lambda_\Delta = 0$ and
  \[ |\lambda|^{\trans} = \left|
  \begin{pmatrix}2\cos(\frac{\pi}{10}) & 2\cos(\frac{3\pi}{10}) \\ 2\cos(\frac{3\pi}{10}) & -2\cos(\frac{\pi}{10})\end{pmatrix}
  \cdot [u]^{\trans}
  \right|\]
  where $\lambda = (\lambda_0,\lambda_1)$ and where $|\cdot |$ is understood as taking the modulus of each element of a vector.

\end{proposition}

\begin{proof}
  The key point of this proof is the decomposition of $M$ as $[u]_0M_0(5)+[u]_1M_1(5)$. By Proposition~\ref{prop:circulant} the matrices
  $M_0(5)$ and $M_1(5)$ have eigenspaces $\Delta$, $\e{0}$ and $\e{1}$ with eigenvalues

  \[ \begin{matrix*}[l]
     \lambda_{\Delta,0} = 0,  &\hspace*{1cm}  & \lambda_{\Delta,1} =0,  \\
     \lambda_{0,0} = 2\cos\tfrac{\pi}{10}e^{-\imag\frac{\pi}{10}},  &  & \lambda_{0,1} = 2\cos\tfrac{3\pi}{10}e^{-\imag\frac{\pi}{10}}, \\
     \lambda_{1,0} = 2\cos\tfrac{3\pi}{10}e^{-\imag\frac{3\pi}{10}},  &  & \lambda_{1,1} = -2\cos\tfrac{\pi}{10}e^{-\imag\frac{3\pi}{10}},
  \end{matrix*}\]
  where $\lambda_{i,j}$ (resp. $\lambda_{\Delta,j}$) is the eigenvalue of $M_j(5)$ on eigenspace $\e{i}$ (resp. on $\Delta$).
  Since the two elementary matrices $M_0(5)$ and $M_1(5)$ have the same eigenspaces we have
  \begin{align*}
    \lambda_\Delta &= [u]_0\lambda_{\Delta,0} + [u]_1 \lambda_{\Delta_1} =0,\\
    \lambda_0 &= [u]_0\lambda_{0,0} + [u]_1\lambda_{0,1} = [u]_0\cdot 2\cos\tfrac{\pi}{10}e^{-\imag\frac{\pi}{10}} + [u]_1\cdot 2\cos\tfrac{3\pi}{10}e^{-\imag\frac{\pi}{10}}, \\
    \lambda_1 &= [u]_0\lambda_{1,0} + [u]_1\lambda_{1,0} = [u]_0\cdot 2\cos\tfrac{3\pi}{10}e^{-\imag\frac{3\pi}{10}} + [u]_1 \cdot (-2\cos\tfrac{\pi}{10}e^{-\imag\frac{3\pi}{10}}).
  \end{align*}
  Furthermore, since on the eigenspace $\e{0}$ (resp. on $\e{1}$) the complex eigenvalues of the elementary matrices have the same argument, the modulus of $\lambda_0$ (resp. of $\lambda_1$) is a simple linear combination, \emph{i.e.},
  \begin{align*}
    |\lambda_0| &= \left|[u]_0\cdot 2\cos\tfrac{\pi}{10}+ [u]_1\cdot 2\cos\tfrac{3\pi}{10}\right|\\
    |\lambda_1| &= \left|[u]_0\cdot 2\cos\tfrac{3\pi}{10} + [u]_1 \cdot (-2\cos\tfrac{\pi}{10})\right|
  \end{align*}
  We can then reformulate it as a matrix-vector product to obtain the formula in the proposition.
\end{proof}

Note that in Sections \ref{sec:subrosa} and \ref{sec:planar-rosa} we use the term \emph{eigenvalue matrix} for the matrix $N_\lambda$ that links the abelianized edgeword $[u]$ to the eigenvalue vector $|\lambda|$. Here
$$N_\lambda := \begin{pmatrix} 2\cos(\frac{\pi}{10}) & 2\cos(\frac{3\pi}{10}) \\ 2\cos(\frac{3\pi}{10}) & -2\cos(\frac{\pi}{10})\end{pmatrix}.$$
 An idea that helps to understand this eigenvalue matrix is to note that the length of the diagonal of a unit rhombus of angle $\theta$ is $2\cos\frac{\theta}{2}$, so on the eigenspace $\e{0}$ a narrow rhombus will add $2\cos\frac{\pi}{10}$ and a wide rhombus will add $2\cos\frac{3\pi}{10}$ to the eigenvalue. On the eigenspace $\e{1}$ they will weight $2\cos\frac{3\pi}{10}$ and $-2\cos\frac{\pi}{10}$, respectively, because the projection is different and the rhombuses are deformed.

Consider again our specific example $\sigma_5$ with the edgeword $131131$. We have $M_5 = 4M_0(5) + 2M_1(5)$. The eigenvalues of $M_5$
were calculated in (\ref{eq:eigenvalues5}). From these, or from
the formulation in Proposition~\ref{prop:eigenvalues_5fold} for their moduli, we obtain that $|\lambda_0|>1$ and $|\lambda_1| < 1$. We also have $\lambda_\Delta=0$. It turns out, and will be proved in a more general setting in Proposition~\ref{prop:eigenvalue_planarity} of Section~\ref{sec:planar-rosa},
that from these bounds we can conclude that any tiling legal for $\sigma_5$ is a discrete plane of slope $\e{0}$. In fact, this is then true for any
substitution with the same abelianized edgeword $[u_5] = \begin{pmatrix} 4, & 2 \end{pmatrix}$.

Let us briefly introduce the idea that we use in Section \ref{sec:planar-rosa} to find a suitable substitution. For the substitution $\sigma$ to be planar of slope $\e{0}$, \ie, for $\sigma$ to generate discrete planes of slope $\e{0}$, we want $|\lambda _0|> 1$ and $|\lambda_1|<1$. Let us now remark that
$$ \begin{pmatrix} 2\cos(\frac{\pi}{10}) & 2\cos(\frac{3\pi}{10}) \\ 2\cos(\frac{3\pi}{10}) & -2\cos(\frac{\pi}{10})\end{pmatrix} \cdot\begin{pmatrix} \frac{2}{5}\cos\frac{\pi}{10}, & \frac{2}{5}\cos\frac{3\pi}{10}\end{pmatrix}\trans= (1,0)\trans.$$
If we take an edgeword $[u] \approx \alpha\begin{pmatrix} \cos\tfrac{\pi}{10}, & \cos\tfrac{3\pi}{10}\end{pmatrix}$ for some positive real number $\alpha$, we have $|\lambda| \approx \tfrac{5\alpha}{2}\begin{pmatrix}1, & 0 \end{pmatrix}$ so for a large enough $\alpha$ and a good enough approximation we have $|\lambda_0|>1>|\lambda_1|$.
$\mathbb{Z}^2$) approximates the real line $\left\langle \begin{pmatrix} \cos\tfrac{\pi}{10}, & \cos\tfrac{3\pi}{10}\end{pmatrix} \right\rangle$. Note that
$[u_5] = \begin{pmatrix} 4, & 2 \end{pmatrix}$ is not very far from this line.

\begin{figure}[h]
  \center
  \begin{subfigure}[c]{0.48\textwidth}
    \center
    \includegraphics[width=0.6\textwidth]{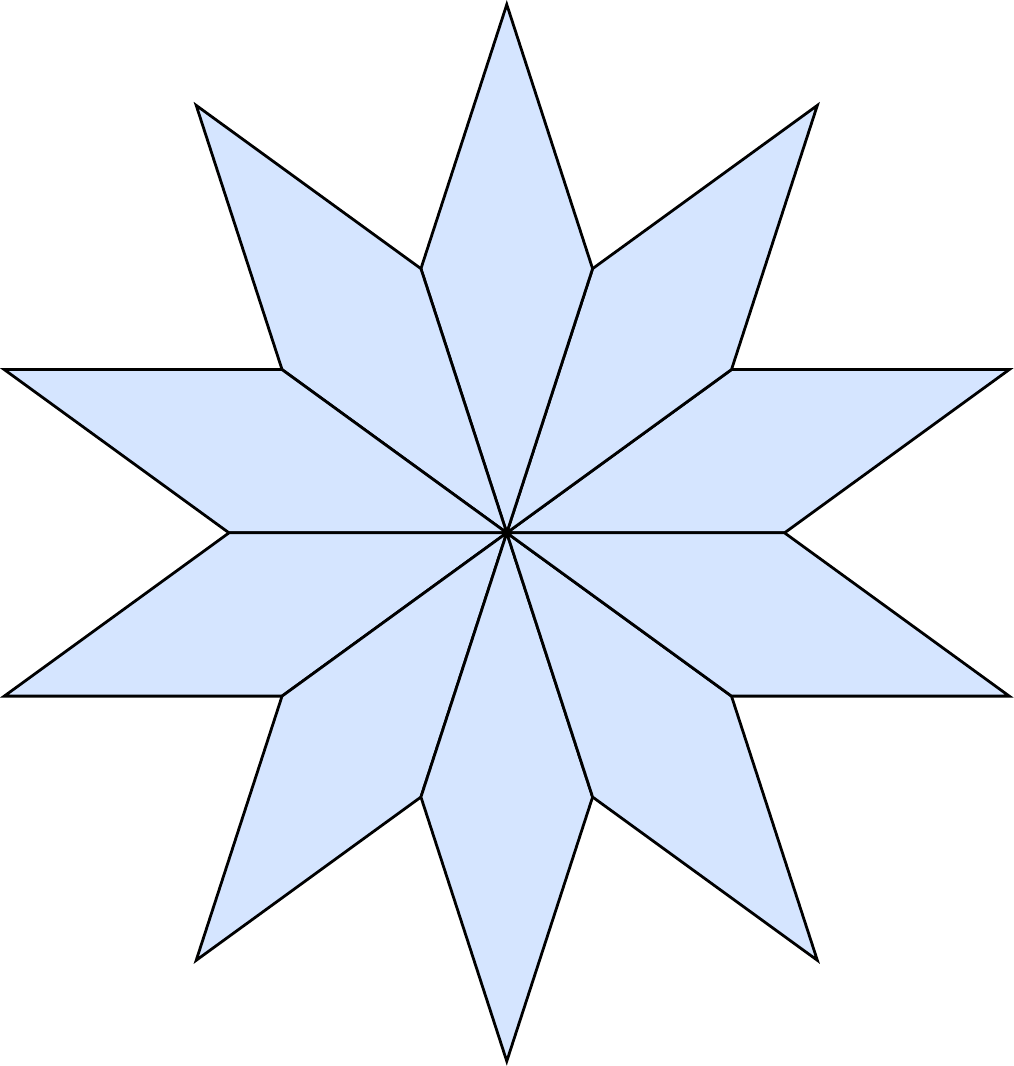}
    \caption{$S_5$}
  \end{subfigure}
  \begin{subfigure}[c]{0.48\textwidth}
    \includegraphics[width=\textwidth]{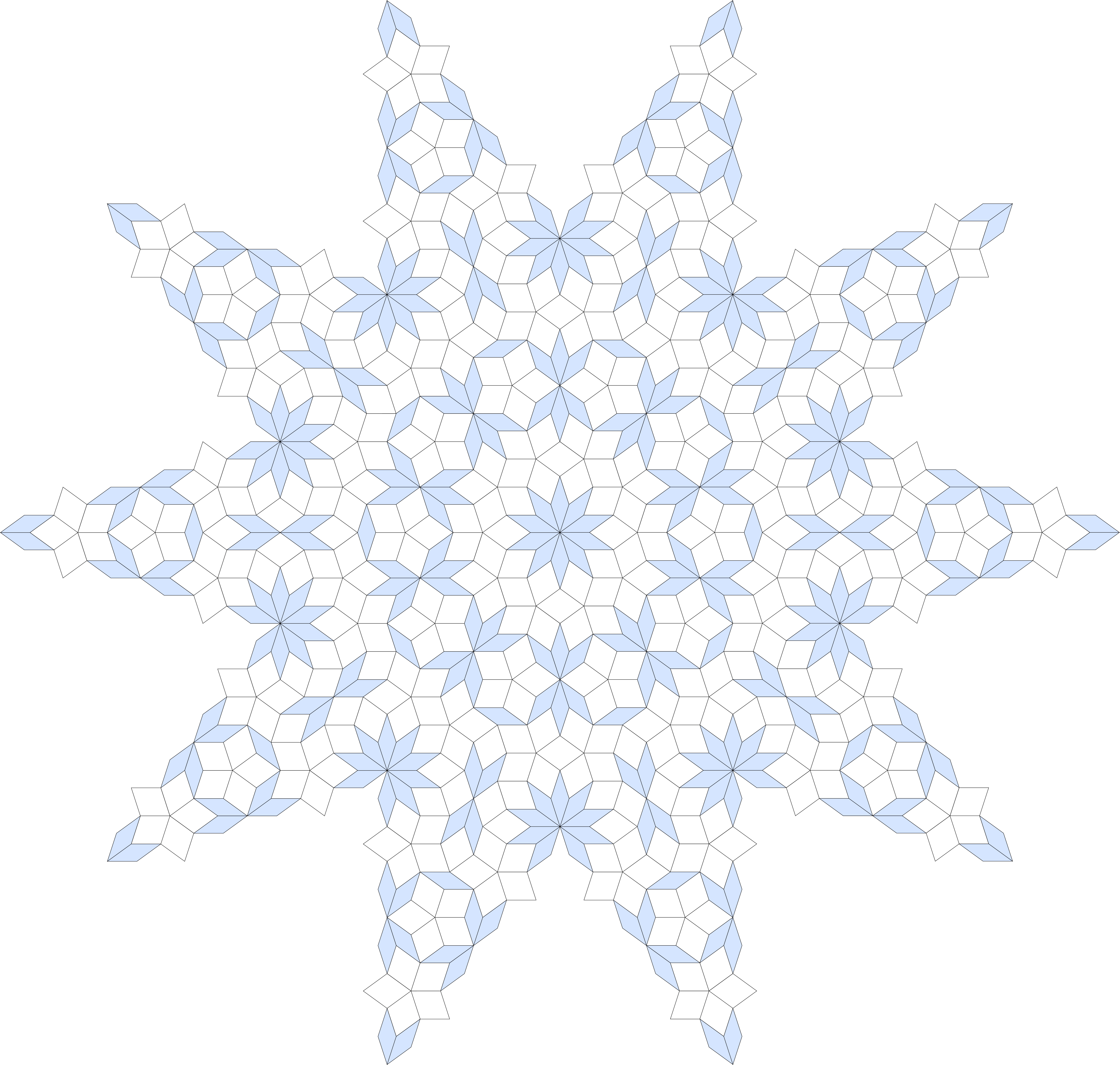}
    \caption{$\sigma_5(S_5)$}
    \end{subfigure}
  \caption{The star $S_5$ and its image by $\sigma_5$.}
  \label{fig:s5_sigma_s5}
\end{figure}

To conclude, let us define one specific tiling that is legal for $\sigma_5$.
Let the star $S_5$ be a corolla of ten narrow rhombuses around a vertex as in Figure \ref{fig:s5_sigma_s5}.
Since a portion of $S_5$ appears in the corner of every metatile (see Figure \ref{fig:sigmafive}), $S_5$ appears in $\sigma_5^2(r_0)$ so that it is a legal pattern for $\sigma_5$. Furthermore $S_5$ appears at the centre of $\sigma_5(S_5)$, and by immediate recursion $\sigma_5^n(S_5)$ appears at the centre of $\sigma_5^{n+1}(S_5)$ for any $n$.

We define the Sub Rosa 5 substitution tiling as $$\mathcal{T}^\infty:= \lim\limits_{n\to\infty} \sigma_5^n(S_5).$$ $\mathcal{T}^\infty$ is a well defined infinite tiling and it has $\sigma_5^n(S_5)$ as its central patterns (see Figure \ref{fig:Tinfty}).

\begin{figure}[h]
  \center \includegraphics[width=\textwidth]{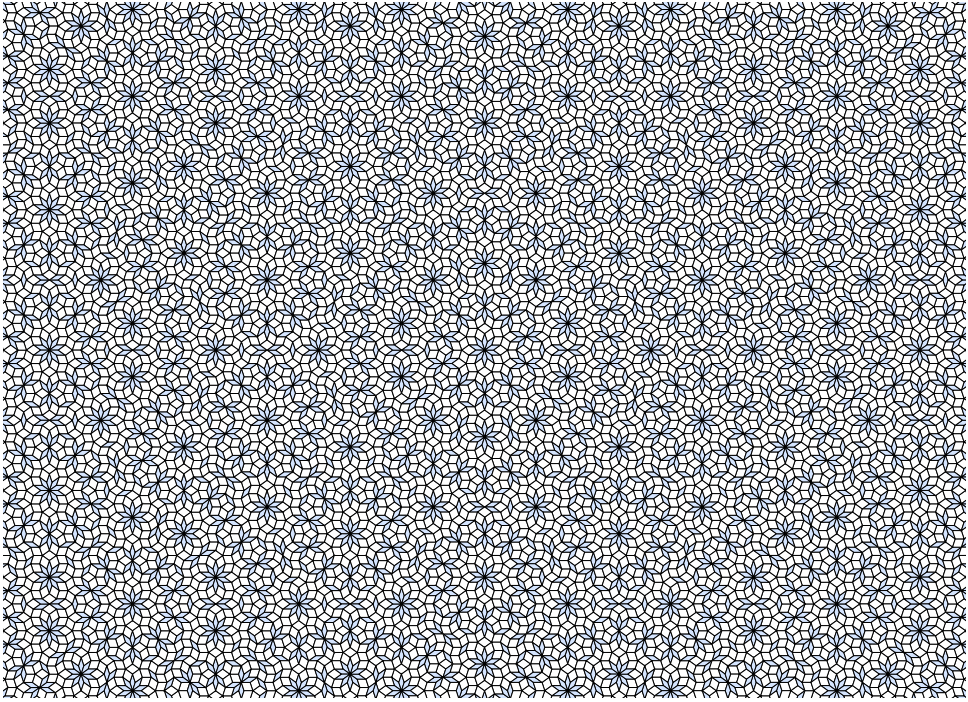}
  \caption{A central patch of $\mathcal{T}^\infty$.}
  \label{fig:Tinfty}
\end{figure}

\begin{proposition}
  $\mathcal{T}^\infty$ is a substitution discrete plane with global 10-fold rotational symmetry.
\end{proposition}
\begin{proof}
 We decompose the proof in three independent parts:
  \begin{itemize}
  \item $\mathcal{T}^\infty$ is legal for the substitution. Indeed, any finite patch of $\mathcal{T}^\infty$ is in some $\sigma_5^n(S_5)\subset \sigma_5^{n+2}(r_0)$, with $r_0$ the narrow rhombus, as discussed above.
  \item $\mathcal{T}^\infty$ is a discrete plane. Indeed, the tiling $\mathcal{T}^\infty$ is legal for $\sigma_5$, which has the abelianized edgeword $[u_5]=(4,2)$.
   By Proposition \ref{prop:eigenvalue_planarity} it is a discrete plane of slope $\e{0}$ because the eigenvalues of the expansion satisfy the required
   conditions $|\lambda_0|>1$ and $|\lambda_\Delta|,|\lambda_1|<1$.

  \item $\mathcal{T}^\infty$ has global 10-fold rotational symmetry around the origin. Indeed, by construction, any patch is included in some $\sigma_5^n(S_5)$ centered at the origin. This has a 10-fold rotational symmetry around the origin. So the image of the patch under the rotation by angle $\tfrac{\pi}{5}$  around the origin
      is a patch of $\sigma_5^n(S_5)$, which means that it is also in $\mathcal{T}^\infty$.
  \end{itemize}
\end{proof}

For more details on the Sub Rosa 5 substitution tiling see \cite[\S 6]{lutfalla2021thesis}.
\section{Sub Rosa substitution tilings}
\label{sec:subrosa}

In this section we briefly present the construction for the Sub Rosa substitution tilings
for higher values of $n$, as defined in \cite{kari2016}. We then present how to lift the Sub Rosa substitutions in $\R{n}$ and we compute the eigenvalues of the Sub Rosa expansions in $\R{n}$ to prove Theorem \ref{th:subrosa}.

\paragraph{Construction} The Sub Rosa tilings form a family of substitution rhombus tilings with a global $n$-fold rotational symmetry \cite{kari2016}. We will here only consider the case for odd $n$ since the two constructions for odd and even $n$ are somewhat different.
In the Sub Rosa construction the substitution rule is given by the edgeword of the substitution, \emph{i.e.}, the sequence of rhombuses that intersects the edge of the substitution as in Section \ref{sec:planar}. The interior is then tiled using a variant of the Kenyon criterion \cite{kenyon1993}.  The edgeword $\Sigma(n)$ is given by
$$\Sigma(n)= s(n)\cdot\overline{s(3)}\cdot\overline{s(5)} \dots \overline{s(n-2)}\, | \, s(n-2)\dots s(5)\cdot s(3)\cdot \overline{s(n)}$$
where $s(n)=135\dots (n-2)$ with each odd integer $i$ representing the rhombus of angles $\tfrac{i\pi}{n}$ and $\tfrac{(n-i)\pi}{n}$. The notation
  $\overline{u}$ is for the mirror image of $u$, and $|$ denotes the middle of the palindromic word.

The abelianized edgeword $[\Sigma(n)]$ is $(n-1, n-3, n-5, \dots 2)$, that is, there are $n-1$ rhombuses of type 1, $n-3$ rhombuses of type $3$ etc.

\begin{figure}[h]
  \caption{Table of $\Sigma(n)$ for small $n$.}
  \begin{center}
    \begin{tabular}{l c}
      $\Sigma(1)$ & | \\
    $\Sigma(3)$ & 1|1 \\
    $\Sigma(5)$ & 131|131\\
    $\Sigma(7)$ & 135131|131531\\
    $\Sigma(9)$ & 1357131531|1351317531 \\
    $\Sigma(11)$ & 135791315317531|135713513197531 \\
  \end{tabular}
  \end{center}
  \end{figure}

\paragraph{Lifting in $\mathbb{R}^n$}
In the following we consider a vertex-hierarchic substitution $\sigma$ on rhombus tiles with the $n$th roots of unity
$$\vec{v}_j:=\left( \cos\frac{2j\pi}{n}, \sin\frac{2j\pi}{n}\right) = e^{\imag\frac{2j\pi}{n}} \text{, for } j\in\{0,1,\dots,n-1\}$$
as edge directions. Let $\phi$ be the associated expansion. Just as for the example with $n=5$ in Section \ref{sec:planar}, we lift the rhombus tilings to $\mathbb{R}^n$
decomposed into $\lfloor \frac{n}{2} \rfloor$ planes and a line.
We define the plane $\e{j}$ for $0\leq j <  \lfloor \tfrac{n}{2}\rfloor$ by its two generating vectors
$$ \left(\cos\frac{2(2j+1)i\pi}{n}\right)_{0\leq i < n} \qquad\text{and}\qquad \left(\sin\frac{2(2j+1)i\pi}{n}\right)_{0\leq i < n}$$
which we also write as a single complex generating vector $$\left(e^{\imag\frac{2(2j+1)i\pi}{n} } \right)_{0\leq i < n}.$$
We define the line $\Delta = \langle (1,1,\dots,1)\rangle$. Note that spaces $\Delta$ and $\e{j}$ are orthogonal and $\mathbb{R}^n$
is their direct sum.

We also assume that the image of any edge by the substitution is the same up to rotation and translation. This implies that $\phi$ is a cyclic linear function and its
matrix is a circulant matrix, just
as in the $5$-fold case in Section \ref{sec:planar}.
Let $u$ be the edgeword of the substitution, so that $u$ is a palindrome over the alphabet $\{1,3,\dots ,n-2\}$.
In particular, the Sub Rosa substitution $\sigma_n$ fits this setup. Also the Planar Rosa substitution in Section~\ref{sec:planar-rosa}
satisfies these conditions.

Let us denote by $M_{\phi}=\circulant{(m_0,\dots ,m_{n-1})}$ the (circulant) matrix of the linear function $\phi$.
Proposition~\ref{prop:circulant} implies that $\phi$ admits $\Delta$ as an eigenspace with an eigenvalue $\lambda_\Delta$ and the planes $\e{j}$ as complex eigenspaces with complex eigenvalues $\lambda_j$, for $0\leq j <\lfloor \frac{n}{2} \rfloor$.
We define the vector $|\lambda|:=\left(|\lambda_0|,\dots |\lambda_{m-1}|\right)$ and call it the \emph{eigenvalue vector}.

We want to study the eigenvalues of $\phi$. Let us translate this problem into terms of the edgeword $u$.
The $i$th coordinate $[u]_i$ of the abelinized edgeword $[u]$ is the number of rhombuses of type $2i+1$ in the sequence $u$, \ie, the number of rhombuses of angle $(2i+1)\tfrac{\pi}{n}$ in $u$.

\begin{definition}[Elementary matrices]
  \label{def:elementary_matrices}
  Let us define the \emph{elementary matrices} $M_i(n)$ for $0\leq i < \lfloor\tfrac{n}{2}\rfloor$ as the matrices that represent the expansions of substitutions with a single rhombus of angle $2i+1$ in the edgeword.
  More precisely let us define $\fst{i}$ and $\scd{i}$ by
  $$\fst{i}:= i \lceil \tfrac{n}{2}\rceil \mod n, \qquad \scd{i}:= -(i+1)\lceil\tfrac{n}{2}\rceil\mod n.$$
  Then $M_i(n)$ is the matrix with all coefficients to $0$ except for the diagonal $\fst{i}$ which has coefficients $(-1)^i$ and the diagonal $\scd{i}$ which has coefficients $(-1)^{i+1}$. In other words, the first column of the circulant matrix $M_i(n)$ has values $(-1)^i$ and $(-1)^{i+1}$ on rows $\fst{i}$ and $\scd{i}$, respectively, and values $0$ elsewhere.
\end{definition}

For examples of such matrices see the elementary matrices $M_0(5)$ and $M_1(5)$ in Section \ref{sec:planar}.

\begin{lemma}
Definition \ref{def:elementary_matrices} is correct, i.e., the matrix $M_i(n)$ defined is indeed the expansion matrix of the substitution with only one rhombus of angle $2i+1$ on the edgeword.
\end{lemma}
\begin{proof}
  Let us look at the substitution in $\R{2}$ rather than in $\R{n}$ for this proof.
Take the set of the edge directions $\{\vec{v}_i, -\vec{v}_i, 0\leq i < n\}$ in the rotation ordering as in Figure \ref{fig:pmvi}.
In this set  $\vec{v}_i$ has neighbours $-\vec{v}_{i+\left\lceil \frac{n}{2}\right\rceil}$ in the positive orientation and $-\vec{v}_{i-\left\lceil \frac{n}{2}\right\rceil}$ in the negative or clockwise orientation. Note that since $n$ is odd we have $2\ceil{\tfrac{n}{2}} = 1 \mod n$, so we have as expected that the second neighbours of $v_i$ are $v_{i+1}$ and $v_{i-1}$.

  \begin{figure}[t]
    \center
    \includegraphics[width=8cm]{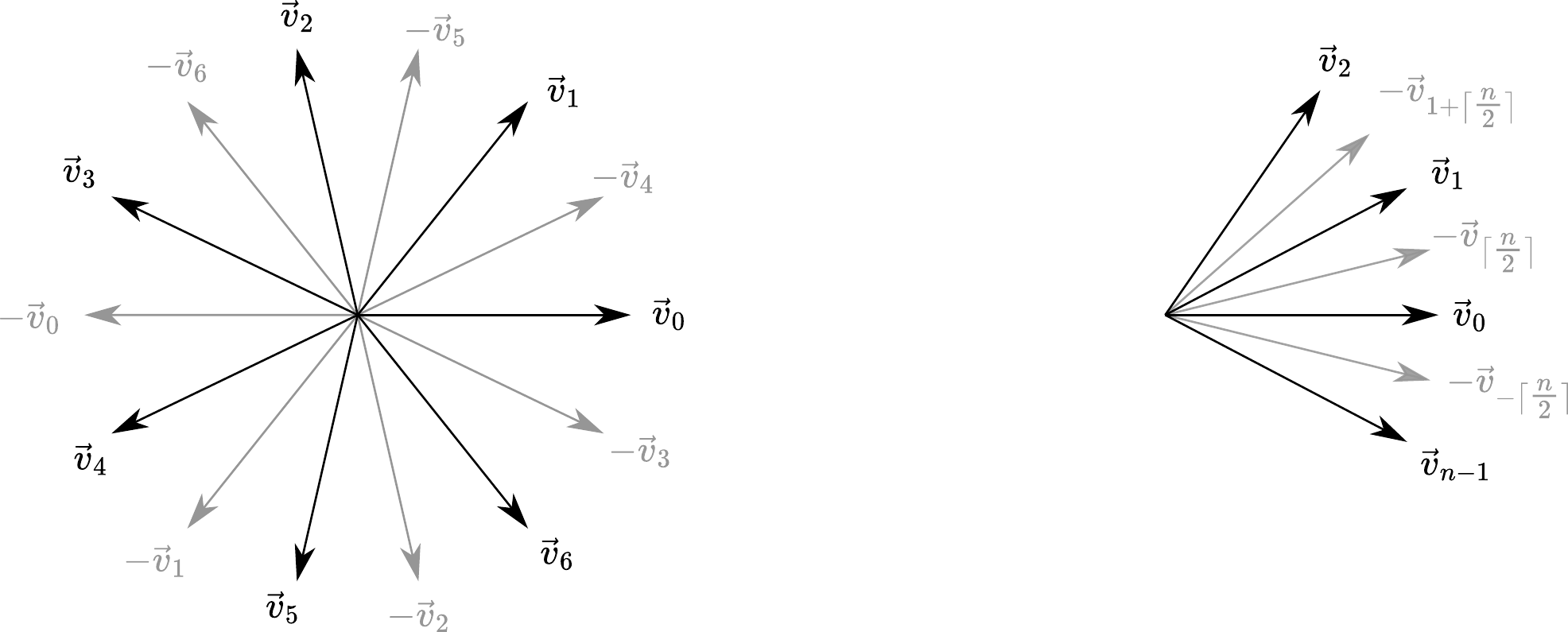}
    \caption{The vectors $\pm \vec{v}_i$ and rotation ordering. On the left for $n=7$ and on the right in the general case.}
    \label{fig:pmvi}
  \end{figure}

  To determine the first column of the matrix $M_i(n)$, consider the image of the edge $\vec{v}_0$ by the substitution with only the single rhombus of type
  $2i+1$ in the edgeword:
  \begin{itemize}
  \item In the case of $i=0$ the image is a rhombus with vectors $\vec{v}_0$ and $-\vec{v}_{ - \left\lceil \frac{n}{2}\right\rceil }$.
  \item In the case of $i=1$ the image is a rhombus with vectors $-\vec{v}_{\left\lceil \frac{n}{2}\right\rceil }$ and $\vec{v}_{-2\left\lceil \frac{n}{2}\right\rceil}$.
  \item In the general case, for a rhombus of type $2i+1$, the image is a rhombus with vectors $(-1)^i\vec{v}_{\fst{i}}$ and $(-1)^{i+1}\vec{v}_{\scd{i}}$.
  \end{itemize}
  See Figure \ref{fig:vect_metaedge} for an illustration of this. Thus, as claimed, the first column of $M_i(n)$ has
  values $(-1)^i$ and $(-1)^{i+1}$ on rows $\fst{i}$ and $\scd{i}$, respectively, and values $0$ elsewhere.

  \begin{figure}[t]
    \includegraphics[width=\textwidth]{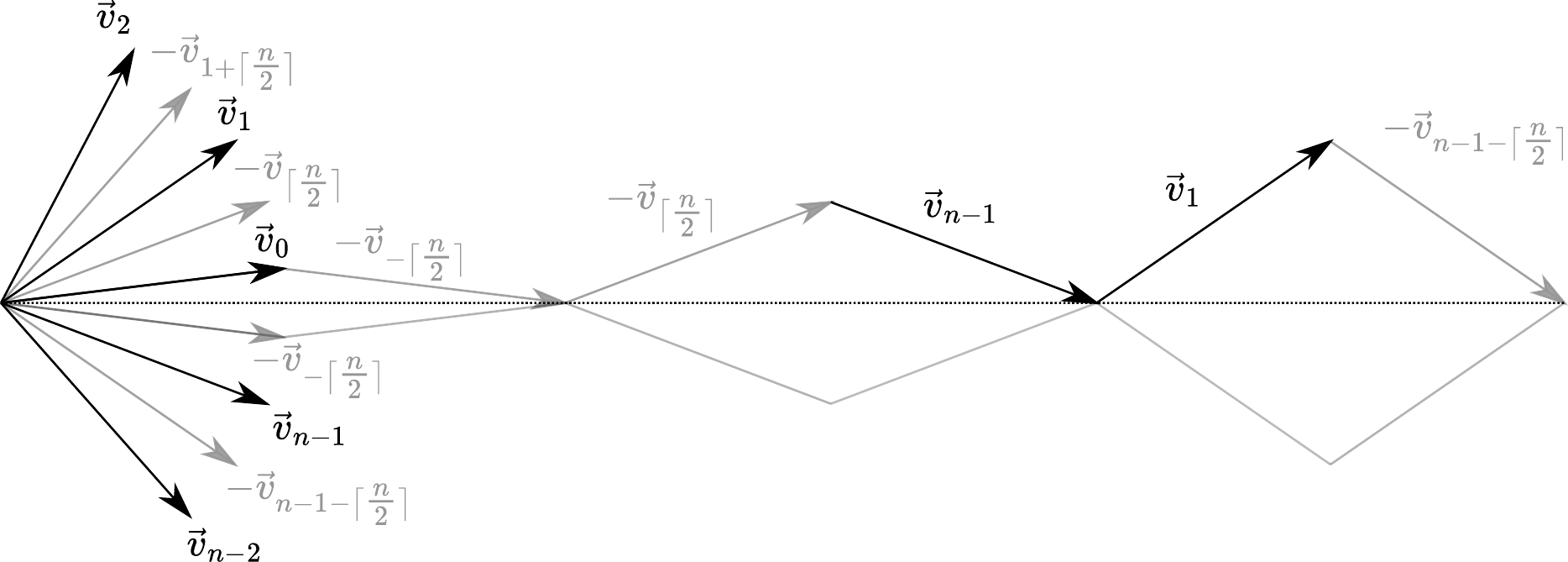}
    \caption{Vectors on the rhombuses of the edgeword.}
    \label{fig:vect_metaedge}
  \end{figure}
\end{proof}

\begin{lemma}[Decomposition as a linear combination of elementary matrices]
  The expansion matrix $M_\phi$ is a linear combination of the elementary matrices $M_i(n)$ with the coefficients from the abelianized edgeword $[u]$:
  $$ M_\phi = \sum\limits_{i=0}^{\lfloor\frac{n}{2}\rfloor -1} [u]_iM_i(n).$$

   \label{lem:coef_rhomb_subrosa}
\end{lemma}

   \begin{proof}
     The proof is reduced to the fact that since the edgeword is $u$ then (up to reordering) $\phi(\vec{v}_0)$ is $[u]_0$ times the vectors of rhombus $1$ plus $[u]_1$ times the vectors of rhombus $3$ etc. This gives
     $$ \phi(\vec{v}_0) = \sum\limits_{i=0}^{\lfloor\frac{n}{2}\rfloor -1} [u]_i\left((-1)^i\vec{v}_{\fst{i}} +(-1)^{i+1}\vec{v}_{\scd{i}}\right),$$
     which when we consider the lifted expansion translates to
     $$ \phi(\vec{e}_0) = \sum\limits_{i=0}^{\lfloor\frac{n}{2}\rfloor -1} [u]_i\left((-1)^i\vec{e}_{\fst{i}} +(-1)^{i+1}\vec{e}_{\scd{i}}\right).$$
From this we get the expected decomposition for $M_\phi$.

\end{proof}

Let us now use this decomposition to get a nice formula to compute the eigenvalues from the edgeword.

\begin{definition}[Eigenvalue matrix]
  \label{def:eigenmatrix}
 Let us define the \emph{eigenvalue matrix} $$N_n := \Bigg( 2\cos\left(\frac{(2i+1)(2j+1)\pi}{2n}\right)\Bigg)_{0\leq i,j< \left\lfloor\frac{n}{2}\right\rfloor}.$$
\end{definition}
Note that the eigenvalue matrix $N_n$ is symmetric: $N_n\trans=N_n$. This means that it does not matter whether we multiply
it with a vector from the left or from the right.

\begin{lemma}[Eigenvalues]
\label{lemma:planarrosa_eigenvalues}
  The expansion matrix $M_\phi$ has $\e{0}, \e{1}, \dots ,\e{\lfloor \tfrac{n}{2}\rfloor -1}$ and $\Delta$ as eigenspaces.
  We denote by $\lambda_{j}$ the eigenvalue of  $M_\phi$ on the eigenspace $\e{j}$, and by $\lambda_\Delta$ on the eigenspace $\Delta$.
  We denote by $|\lambda|$ the vector of the moduli of the eigenvalues $\lambda_{j}$.
  We have $\lambda_\Delta = 0$ and  $$ |\lambda|\trans = \left| \eigenmatrix{n}\cdot [u]\trans \right|,$$
  where $\eigenmatrix{n}$ is the eigenvalue matrix and $u$ is the edgeword of $\phi$. More precisely we have
  $$\lambda_{j} = \left(\sum\limits_{i=0}^{\lfloor\tfrac{n}{2}\rfloor -1} [u]_i2\cos\left(\frac{(2j+1)(2i+1)\pi}{2n}\right)\right)e^{- \imag \frac{(2j+1)\pi}{2n}}.$$
\end{lemma}
\begin{proof}
  We use the  decomposition
   \[ M_\phi = \sum\limits_{i=0}^{\lfloor\frac{n}{2}\rfloor -1} [u]_i\cdot M_i(n) \]
   of the expansion matrix and use the eigenvalues of the elementary matrices.
  The matrix $M_\phi$ and the elementary matrices $M_i(n)$ are circulant so by Proposition~\ref{prop:circulant}
  they have eigenspaces $\e{0}, \e{1}, \dots , \e{\lfloor \tfrac{n}{2}\rfloor -1}$ and $\Delta$.
  By the same Proposition, the eigenvalue of the elementary matrix
  $M_i(n)$ corresponding to the eigenspace $\Delta$ is $\lambda_{i,\Delta}=0$, and the eigenvalue $\lambda_{i,j}$ corresponding to
  the eigenspace $\e{j}$ is
   \begin{align*}\lambda_{i,j}(n) &= (-1)^i\left(e^{\imag\frac{2(2j+1)\fst{i}\pi}{n}}-e^{\imag\frac{2(2j+1)\scd{i}\pi}{n}}\right) \\
    &= 2\cos\left(\frac{(2j+1)(2i+1)\pi}{2n}\right) e^{- \imag \frac{(2j+1)\pi}{2n}}. \end{align*}
  (See Lemma \ref{lemma:appendices_eigen_elem_mat} in the Appendices for details on this manipulation of exponentials.)
  Let us remark that the argument of the complex number $\lambda_{i,j}(n)$ does not depend on $i$.

  From this we have that the required eigenvalues of $M_\phi$ are $\lambda_\Delta = 0$ and
  \begin{align*}
    \lambda_j &= \sum\limits_{i=0}^{\left\lfloor\frac{n}{2}\right\rfloor} [u]_i \lambda_{i,j}(n) \\
    &= \left(\sum\limits_{i=0}^{\lfloor\tfrac{n}{2}\rfloor -1} [u]_i2\cos\left(\frac{(2j+1)(2i+1)\pi}{2n}\right)\right)e^{- \imag \frac{(2j+1)\pi}{2n}}. \end{align*}
  This further implies that $$|\lambda|\trans = \left| \eigenmatrix{n}\cdot [u]\trans \right|.$$
\end{proof}

Let us now apply what we have learned to the Sub Rosa substitution $\sigma_n$. Recall that the edgeword is denoted by $\Sigma(n)$, and the
abelianized edgeword $[\Sigma(n)]$ is $(n-1, n-3, n-5, \dots 2)$, \ie, $[\Sigma(n)]_i=n-(2i+1)$.
Denote by $\lambda_{j}(n)$ the eigenvalue of the Sub Rosa $n$ substitution on the eigenspace $\e{j}$, and by $\lambda_\Delta$ on the eigenspace $\Delta$.

\begin{lemma}[Eigenvalues of Sub Rosa]
  For all $n\geq 7$, $|\lambda_{0}(n)|>1$ and $|\lambda_{1}(n)|>1$.
  \label{lem:weirdtrigsums}
  \end{lemma}

\begin{proof}
  We study the sequence $|\lambda_{j}(2k+1)|$ with a fixed $j$ and prove that it is increasing with $k$. Then we only have to calculate that $|\lambda_{0}(7)|$ and $|\lambda_{1}(7)|$ are both greater than 1.

  Let us take two integers $j<k$, and denote $\theta_{j,k}=\frac{(2j+1)\pi}{2(2k+1)}$ and
  $$
  C_{j,k} = \sum\limits_{i=0}^{k-1}4(k-i)\cos\left((2i+1)\theta_{j,k}\right).
  $$
  Then by Lemma~\ref{lemma:planarrosa_eigenvalues}
  $$
  |C_{j,k}| = \left| \sum\limits_{i=0}^{k-1} ((2k+1)-(2i+1))2\cos\left(\frac{(2j+1)(2i+1)\pi}{2(2k+1)}\right)\right | = |\lambda_{j}(2k+1)|.
  $$
 Let us prove that  $$C_{j,k}\cdot \sin^2(\theta_{j,k}) = \cos(\theta_{j,k})$$
  We will write $\theta$ for $\theta_{j,k}$ to avoid clutter:

  \begin{align}
    & C_{j,k} \sin^2\left(\theta\right) = \sum\limits_{i=0}^{k-1}4(k-i)\cos\left((2i+1)\theta\right)\sin^2\left(\theta\right) \label{eq:trigsum1}\\
    &= \sum\limits_{i=0}^{k-1}(k-i)\left( 2\cos((2i+1)\theta) - \cos((2i+3)\theta) - \cos((2i-1)\theta)  \right) \label{eq:trigsum2}\\
    &= \cos\theta. \label{eq:trigsum3}
  \end{align}

  From line (\ref{eq:trigsum1}) to (\ref{eq:trigsum2}) we rewrite cosines and sines as sums of exponentials. Then we expand the product and pair the exponential terms by argument and find three cosines terms.
  From line (\ref{eq:trigsum2}) to (\ref{eq:trigsum3}) we split the sum in three and reindex to have $\cos((2i+1)\theta)$ terms in each sum. Then we merge back and the terms of the sum cancel out. We end up only with boundary terms which sum up to $\cos\theta$.
  For full details of these trigonometric manipulations see Lemma \ref{lemma:appendices_weird_trig_sum} in the Appendices.

Let us remark that $\theta_{j,k} \in (0,\frac{\pi}{2})$ for $j<k$, so we have $$C_{j,k} = \frac{\cos\theta_{j,k}}{\sin^2\theta_{j,k}} > 0,$$
and therefore $C_{j,k}=|\lambda_{j}(2k+1)|$.

Let us fix $j\in\mathbb{N}$ and consider the function $$f: k \to \frac{\cos\theta_{j,k}}{\sin^2\theta_{j,k}} \qquad \text{with } k>j.$$
Note that since $j$ is fixed we can rewrite $\theta_{j,k}$ as $$ \theta_{j,k} = \frac{(2j+1)\pi}{2(2k+1)} = \frac{1}{a\cdot k + b},$$ with $a:= \tfrac{4}{(2j+1)\pi}$ and $b:= \tfrac{2}{(2j+1)\pi}$.
For simplicity we change our variable to $x := a\cdot k + b$ and we consider instead $$h: x \to \frac{\cos\tfrac{1}{x}}{\sin^2\tfrac{1}{x}} \qquad \text{with } x>\tfrac{2}{\pi}$$
We can now compute the derivative $h'$ for $x>\tfrac{2}{\pi}$ to be
$$h'(x) = \frac{\frac{1}{\tan^2\frac{1}{x}}+\frac{1}{\sin^2\frac{1}{x}}}{x^2\sin\tfrac{1}{x}}.$$  We have $h(x)>0$ and $h'(x)>0$ for $x>\tfrac{2}{\pi}$ so $h$ is a positive increasing function. Additionally, as $\lim_{x\rightarrow2/\pi^+}h(x) =0$ and $\lim_{x\rightarrow\infty}h(x) = \infty$, there exists a value $x_1:=h^{-1}(1)$. Then $x>x_1 \Leftrightarrow h(x)>1$.

Now we can translate these results on $h$ to results on $f$. We have $f(k)>0$ for $k>j$ and $f$ is an increasing function.
This means that, at fixed $j$, $|\lambda_j(n)|$ is increasing with $n$.
A direct calculation shows that $|\lambda_{0}(7)|>1$, and $|\lambda_{1}(7)|> 1$.
So for $n\geq 7$ we have $|\lambda_{0}(n)|,|\lambda_{1}(n)|>1$.
See Table \ref{table:subrosaeigenvalues} for the eigenvalues for small $n$.
\begin{table}[h]
  \caption{Approximate moduli of the Sub Rosa eigenvalues for  small $n$.}
  \center
    \begin{tabular}{l| c c c c c}
      n & $|\lambda_{0}(n)|$ & $|\lambda_{1}(n)|$ & $|\lambda_{2}(n)|$ & $|\lambda_{3}(n)|$ & $|\lambda_{4}(n)|$ \\
      \hline
      1 & - & - & - & - & - \\
      3 & 3.46 & - & - & - & - \\
      5 & 9.96 & 0.90 & - & - & - \\
      7 & 19.69 & 2.01 & 0.53 & - & - \\
      9 & 32.66 & 3.46 & 1.09 & 0.39 & - \\
      11 & 48.87 & 5.27 & 1.76 & 0.76 & 0.30 \\
      \end{tabular}
  \label{table:subrosaeigenvalues}
  \end{table}
\end{proof}
This lemma directly implies Theorem \ref{th:subrosa} since for $n\geq 7$ the Sub Rosa substitution $\sigma_n$ admits planes $\e{0}$ and $\e{1}$ as eigenspaces with complex eigenvalues of moduli strictly greater than 1, which in turn implies that tilings admissible for $\sigma_n$ are not discrete planes and not cut-and-project tilings.

For more detailed examples of Sub Rosa tilings studied lifted in $\R{n}$ see \cite[\S 6]{lutfalla2021thesis}.

\section{Tileability conditions}
\label{sec:tileability}

In Sections \ref{sec:planar} and \ref{sec:subrosa} we presented how substitutions can be lifted in $\R{n}$ and  how the boundary of the substitutions are related to planarity. We also presented the family of Sub Rosa substitution tilings which is a family of promising substitution tilings with $n$-fold rotational symmetry. However, it turned out that they are not discrete planes for odd $n\geq 7$.

Now we will consider the tileability of metatiles defined by their boundary before presenting a new construction for substitution tilings with $n$-fold rotational symmetry that are also discrete planes. The idea is that the relation between the edgeword of the substitution and the eigenvalues of the expansion gives us candidates for substitution discrete planes. But the substitutions are only defined on their boundaries so we need to prove that the interior is tileable.
The question is now: given the edges of a metatile, can the metatile be tiled with unit rhombuses?
To address this problem we first use the work of Kenyon \cite{kenyon1993} on tiling a polygon with parallelograms. In our case all edges are of unit length and all angles are multiples of $\frac{\pi}{n}$.
The main result is Proposition \ref{prop:cs_tileability} which is a sufficient condition for the tileability of the metatiles.

We define a \emph{pseudo-substitution} as a substitution that is only defined on the edges of the metatiles. A pseudo-substitution can be extended to a substitution when its metatiles are  tileable. We say that a substitution is a tiled pseudo-substitution. The interior of the metatile of a pseudo-substitution is a polygon (with unit length edges, see Figure \ref{fig:polygon}) and our goal is now to tile it with parallelograms (rhombuses in our case).
\begin{figure}[h]
  \center  \includegraphics[width=\textwidth]{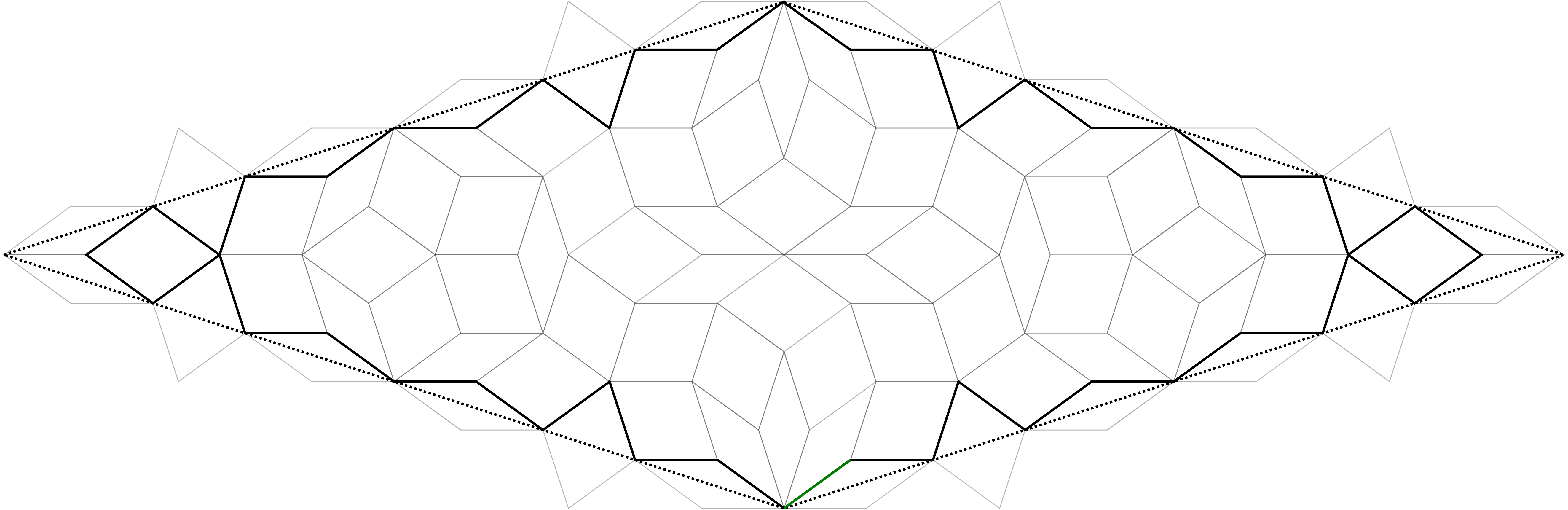}
  \caption{Metatile and parallelogram, in thick line the parallelogram to tile and in thin lines the tiled metatile.}
  \label{fig:polygon}
\end{figure}

The first step of the Kenyon method is to fix a starting vertex on the polygon to tile. Let us write $(\vec{a}_1,\dots,\vec{a}_m)$ for the sequence of oriented edges when going around the boundary of the polygon counterclockwise from the starting point and back. All $\vec{a}_j$ are in the set of directions of the tiling, \ie, $\vec{a}_j\in \{\pm \vec{v}_k\ |\ k=1,\dots ,n\}$.
We say that $\vec{v}_k$ and $-\vec{v}_k$ have the same \emph{edge type} but opposite directions.
We denote $\vec{a}_j^\bot = \imag\vec{a}_j$ for the unit vector orthogonal to $\vec{a}_j$ in the counterclockwise direction.

Suppose the interior of the polygon is actually tiled. As seen in Figure \ref{fig:chains}, from each edge of the polygon starts a chain of rhombuses (shaded in Figure \ref{fig:chains}) that share the same edge type. At the two ends of this chain there
are two edges of the polygon with the same edge types but opposite directions.
These chains define a matching of pairs of edges of the polygon with the following properties.
\begin{enumerate}
  \item[K1] Two edges that are matched have the same edge type but opposite orientations.
  \item[K2] Two matched pairs of edges of the same edge types cannot cross each other with respect to the cyclic ordering of the edges. Indeed, two chains with the same edge type cannot cross, as such a crossing would create a ``flat'' rhombus.
  \item[K3] Two matched edges must ``see'' each other in the parallelogram, in the sense that the
  connecting chain is monotonically increasing in the direction $\vec{a}_j^\bot$. Indeed, otherwise it would mean that the chain circles back on itself.
  \item[K4] The matching is peripherally monotonous: for any two matched pairs $\{\vec{a},\vec{a}'\}$ and $\{\vec{b},\vec{b}'\}$ such that in the cyclic ordering of the edges $\vec{a}<\vec{b}<\vec{a}'<\vec{b}'$, we have $\vec{a}^\bot \cdot \vec{b}\trans > 0$. That ensures that at the crossing of the two chains a real rhombus actually exists.
\end{enumerate}

We call \emph{Kenyon matching} a matching on the oriented edges that follow these properties.

 \begin{theorem}[Kenyon, 93]
   \label{thm:kenyon}
  The polygon is tileable by parallelograms if and only if a Kenyon matching exists.
 \end{theorem}

Let us explain how the Kenyon properties are related to tileability. For a proof of the theorem see the original article \cite{kenyon1993}.
  The main idea is that for the parallelogram to be tileable, as explained above, any edge must match to an edge of same absolute type and opposite direction such that there is a chain of parallelograms between the two matched edges (see Figure \ref{fig:chains}). The existence of the chain forces the fact that the matched edges must be of the same edge type and the opposite direction, and that the matched edges must ``see'' each other in the area to tile. When two such chains cross, the crossing represents a rhombus tile whose edges have the types of the two matched pairs. This implies that matched pairs of the same edge type cannot cross and it also implies the ``peripherally monotonous'' condition.
  \begin{figure}
    \center  \includegraphics[width=\textwidth]{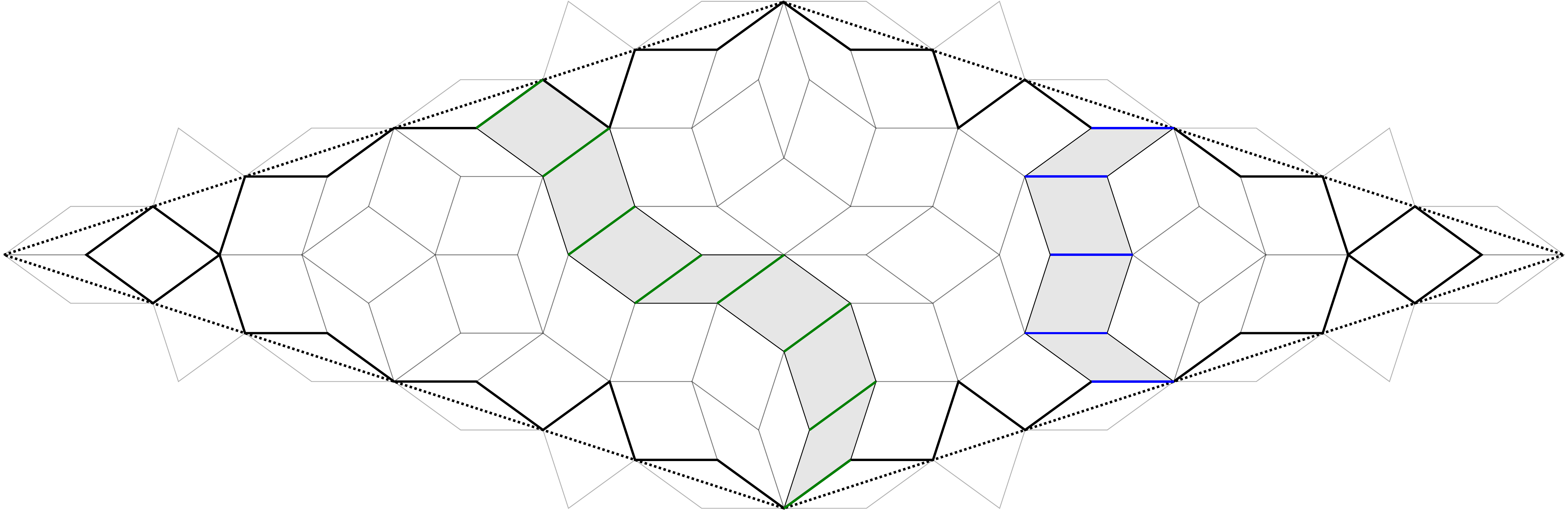}
    \caption{In shaded two chains of rhombuses.}
    \label{fig:chains}
  \end{figure}

In the case of our substitution tilings, we study the tileability of the polygons which are defined by the sequence of rhombuses on the edges of the metatiles,
given by the edgeword of the substitution. The tileability of the metatiles is determined by this sequence of rhombuses.
We denote by the letter $k$, with $k$ odd, the rhombus of angles $\frac{k\pi}{n}$ and $\frac{(n-k)\pi}{n}$. An edgeword $u$ is a word with letters in $\{1,3,5,\dots n-2\}$.

Given the properties we imposed on the edges of the substitution (see Section \ref{sec:planar}) the edgeword is always a palindrome and it is the same on all 4 sides of any meta-rhombus. That is why we only consider one word. There is an equivalence between the expansion function and the abelianization of the edgeword. So the expansion defines the edgeword up to reordering. However, for tileability the order of the letters in the edgeword is very important.

The fact that the opposite edges of the metatiles are identical implies that on the path around the boundary they
contain exactly the same edge types but in the opposite directions. Thus there always exists a matching that satisfies the first Kenyon condition K1. 
The fact that a single edge of the metatile cannot contain both orientations of any edge type means that the two directions of an edge type come 
contiguously one after the other on the cyclic  path around the boundary. This in turn implies that there exists a unique way to match the edges so that both the conditions K1 and K2 are satisfied~\cite{kari2016}. In this matching the $i$'th occurrence of $\vec{v}_k$ in its contiguous segment 
is matched with the $i$'th last occurrence of $-\vec{v}_k$ in its contiguous segment, for all $i$ and $k$.
In the following, by the matching of the edges we always mean this unique way to match the edges that satisfies K1 and K2.
To check tileability we then need to verify that this matching satisfies also the conditions K3 and K4. 

For the proofs we will now define the counting functions of an edgeword.
\begin{definition}[Counting function]
  Let $n\geq 3$ be an odd integer and let $u = u_0u_1\dots u_{m-1}$ be an edgeword of length
  $m$ with symbols from the set $\{1,3,\dots ,n-2\}$.
  For a rhombus type $j$ and a position $x\in \{0,\dots m\}$ we define
  $f_j(x) = |u_0\dots u_{x-1}|_j$, the number of letters $j$ in the prefix of length $x$ of the edgeword $u$. We also define $f_j^{-1}(y)$ as the length of shortest prefix of $u$ with $y$ letters $j$. If there is no such prefix then $f_j^{-1}(y)=+\infty$. Also, we define $f_j^{-1}(0)=0$, and for $y<0$ we set  $f_j^{-1}(y)=-\infty$.

  For $j=n$ we define $f_j(x)=0$ and for $n<j<2n$ we define $f_{j}(x) = -f_{2n-j}(x)$. For these cases we do not define $f_j^{-1}$.
\end{definition}

Let us remark a few things regarding this definition:
\begin{itemize}
\item $f_j^{-1}$ is not an inverse function of $f_j$ since $f_j$ is not bijective. However, for any rhombus type $j$ and a number $y$ such that there is at least $y$ occurrences of the letter $j$, we have $f_j(f_j^{-1}(y))=y$. For any $x$ such that $u_{x-1}=j$ we also have $f_j^{-1}(f_j(x)) = x$. But when $u_{x-1}\neq j$ we have $f_j^{-1}(f_j(x)) < x$.
\item The fact that $j$ is a rhombus type means that it is an odd number between $1$ and $n-2$, coding a unit rhombus of angle $\tfrac{j\pi}{n}$.
\item We give a definition of $f_j(x)$ for $n\leq j < 2n$ because these will appear in the proofs, and it is easier to give them a sense than to always check for the special case of $n\leq j < 2n$.
\item We fix $f_j^{-1}(y) = +\infty$ for simplicity if there are fewer than $y$ occurrences of the letter $j$ in the edgeword. This allows us not to check the number of occurrences of the letter $j$ before using $f_j^{-1}$. Remark that we could also just say that the formulas that involve $f_j^{-1}$ are understood only in the well defined (non-infinity) case. The same holds for $f_j^{-1}(y) = -\infty$ in the case that $y$ is negative.
\end{itemize}

The Kenyon conditions will translate into inequalities on the counting functions $f_j$.

\begin{definition}[Almost-balancedness ]
  We say that a word $u$ is $k$-almost-balanced when for any letters $j_1<j_2$ and any subword $v$ of $u$ we have $|v|_{j_1} - |v|_{j_2} \geq -k $.
  \label{def:letters}

\end{definition}

The idea is that the frequency of appearance of $j_1$ is basically
greater than the frequency of appearance of $j_2$ when $j_1<j_2$, so that $|v|_{j_1} - |v|_{j_2}$ should be positive for long words $v$ But for short factors $v$ we can have slightly negative values, e.g., for the length one subword $v=j_2$ of $u$. This quite unusual definition of $k$-almost-balancedness bounds the negative values that $|v|_{j_1} - |v|_{j_2}$ can take. For example, if $u$ is a binary Sturmian word with letters $j_1$ and $j_2$, with the frequency of $j_1$ larger than the frequency of $j_2$, then $u$ would be 1-almost-balanced.

The main result of this section is the following.

\begin{proposition}[Tileability]
  \label{prop:cs_tileability}
  Let us consider a pseudo-substitution with an edgeword $u$ such that $u$ is 2-almost-balanced.
   Suppose that for every odd $j_1<j_2$, and for every position $k_1$ such that the edgeword has symbol $j_1$ in position $k_1-1$,  we have
  \begin{align}
    f_{|j_2-2|}^{-1}\circ f_{j_2}(k_1)&<f_{|j_1-2|}^{-1}\circ f_{j_1}(k_1). \label{eq:k1counting}
  \end{align}
  Then all the metatiles of the pseudo-substitution are tileable, implying that the pseudo-substitution can be extended to a well defined substitution.
\end{proposition}

This proposition gives us a sufficient condition for the tileability of the metatiles and the well-definedness of the substitutions in our constructions. This proposition was formulated as an implication because that is how we use it, but it is actually an equivalence.

We will now present the proof of this proposition with several lemmas and propositions to cut the proof in smaller pieces and to make it more understandable. We first show the consequences
of tileability of the metatiles  on the counting functions $f_j$. After this we consider the converse direction.

\begin{proposition}
  \label{prop:tileability01}
  If the metatile with angle $\frac{k\pi}{n}$ is tileable then for any $j_1<j_2$, and for any position $k_1$ such that
  the edgeword has symbol $j_1$ in position $k_1-1$, we have
  \begin{align} \label{eq:adjacent_edge}  f_{|j_2-2k|}^{-1}\circ f_{j_2}(k_1)<f_{|j_1-2k|}^{-1}\circ f_{j_1}(k_1).
    \end{align}
\end{proposition}
\begin{proof}
  Let us assume the metatile is tileable so that there is a valid Kenyon matching. Let us see how the chains of parallelograms behave in this case.
   Take a chain of edge type $\vec{a}_1$ that links a rhombus of type $j_1$ on the side 1 of the metatile to a rhombus of type $j_1'$ on the adjacent side 2, as in Figure \ref{fig:one-chain}.
  \begin{figure}
    \center
    \includegraphics[width=0.8\textwidth]{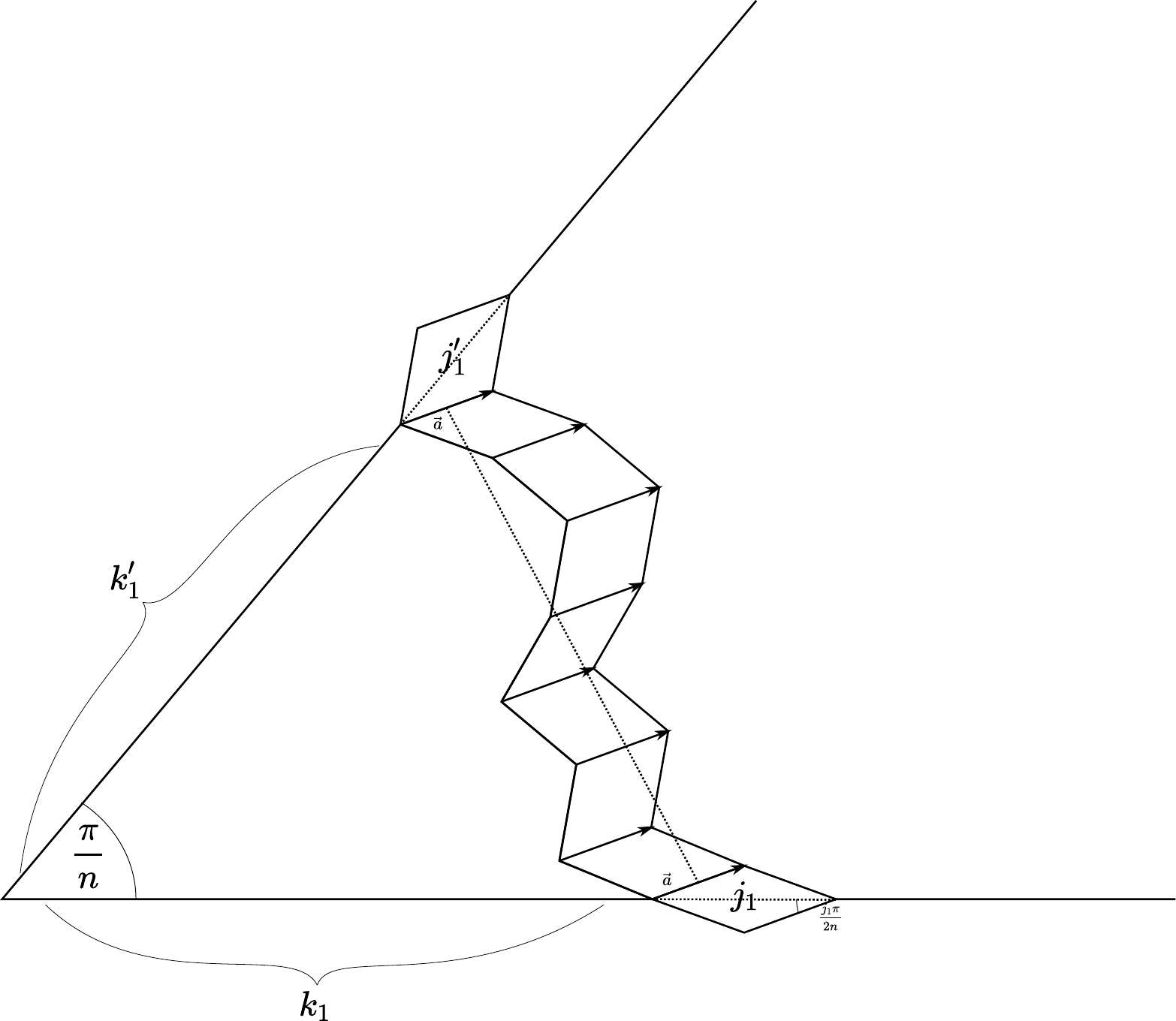}
  \caption{A chain of rhombuses around the narrow corner of the metatile.}
  \label{fig:one-chain}
\end{figure}

On side 1, the direction $\vec{a}_1$ can only be found on the rhombuses of type $j_1$, because the direction of the edges of a rhombus is determined by the general direction of the side and the angle of the rhombus.
Call $k_1-1$ the index of this rhombus on side 1 (starting the indexing from the corner with side 2). There are $f_{j_1}(k_1)$ rhombuses of type $j_1$ between the corner and the chain on side 1.
On side 2, for the same reason the direction $-\vec{a}_1$ only appears on the edges of rhombuses $j_1'$, and there are $f_{j_1'}(k_1')$ such rhombuses between the corner and the chain.
Since the matching is valid, two chains of identical edge types cannot cross each other (condition K2), so every rhombus $j_1$ of side 1 between the corner and the chain is matched to a rhombus of type $j_1'$ of side 2, also between the corner and the chain.
So we have $f_{j_1}(k_1) = f_{j_1'}(k_1')$, which we can reformulate as $f_{j_1'}^{-1}\left(f_{j_1}(k_1)\right)=k_1'$.

The angle of the corner of the metatile is $\frac{k\pi}{n}$, the half-angles of the rhombuses are $\frac{j_1\pi}{2n}$ and $\frac{j_1'\pi}{2n}$, and their sides are parallel.
 From these we can deduce that $j_1'=|j_1-2k|$.

Now let us consider a second chain in this setting.
\begin{figure}
  \center
    \includegraphics[width=0.8\textwidth]{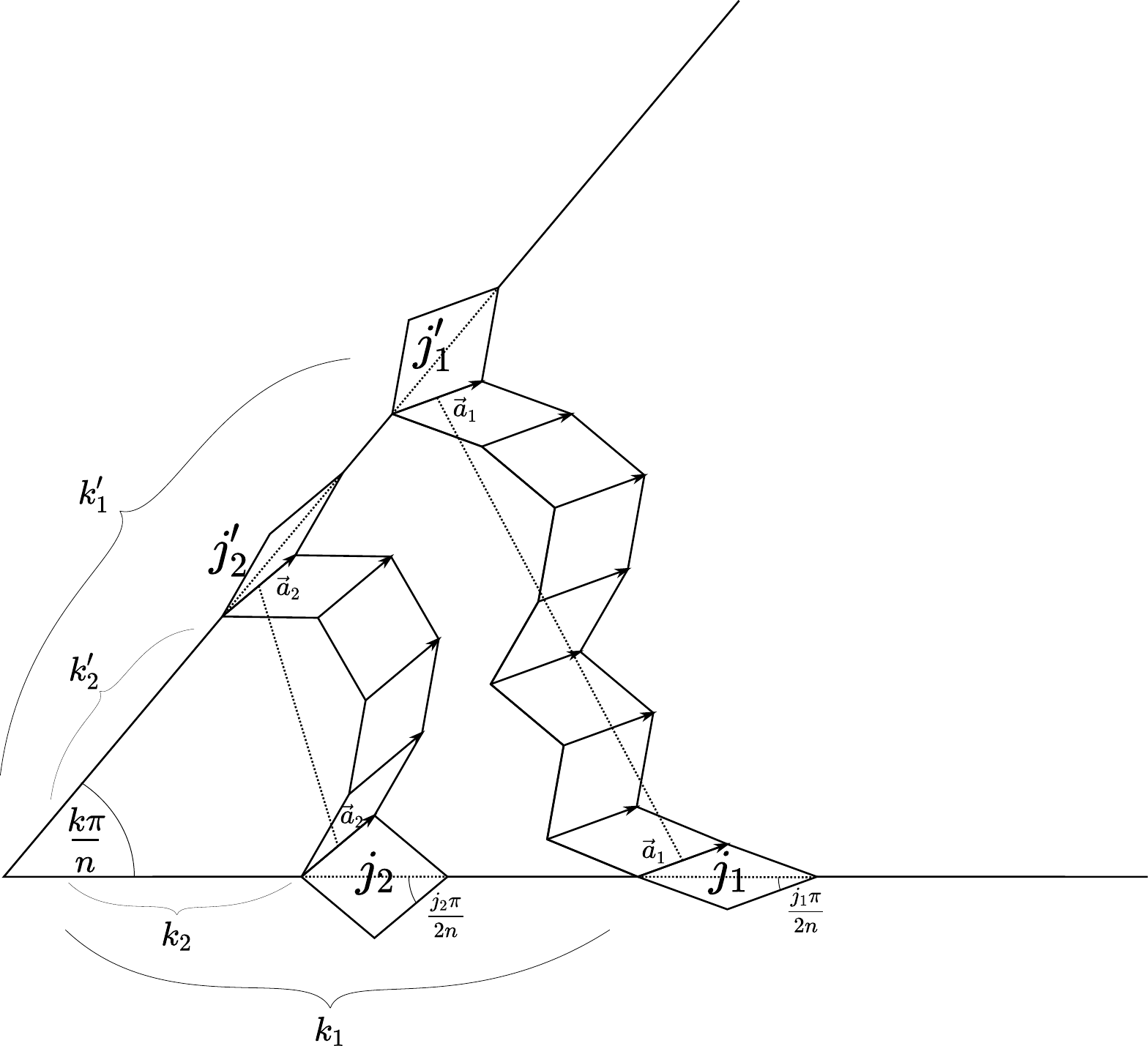}
  \caption{A second chain of rhombuses around the narrow corner.}
  \label{fig:two-chains}
\end{figure}
Take a rhombus type $j_2>j_1$ and denote by $k_2-1$ the last position at which a rhombus $j_2$ appears between the corner and the chain of direction $\vec{a}_1$, as in Figure \ref{fig:two-chains}. Denote by $\vec{a}_2$ the edge of the rhombus facing the corner, and consider the chain that starts at this edge.
As for the first chain, this one is linked to some rhombus of type $j_2'=|j_2-2k|$ on side 2 which is at position $k_2'$ satisfying $f_{j_2'}^{-1}\left(f_{j_2}(k_2)\right) = k_2'$.
But since there is no rhombus of type $j_2$ between $k_2$ and $k_1-1$ we have $f_{j_2}(k_2)=f_{j_2}(k_1)$.

Since $j_2>j_1$ we have $\vec{a}_2^\bot\cdot \vec{a}_1\trans \leq 0$ so the two chains cannot cross (see Figure \ref{fig:two-chains}). This means that $k_2'<k_1'$.
So overall $$ f_{j_2'}^{-1}\circ f_{j_2}(k_1) < f_{j_1'}^{-1}\circ f_{j_1}(k_1).$$
\end{proof}

\begin{figure}
  \center
    \includegraphics[width=0.8\textwidth]{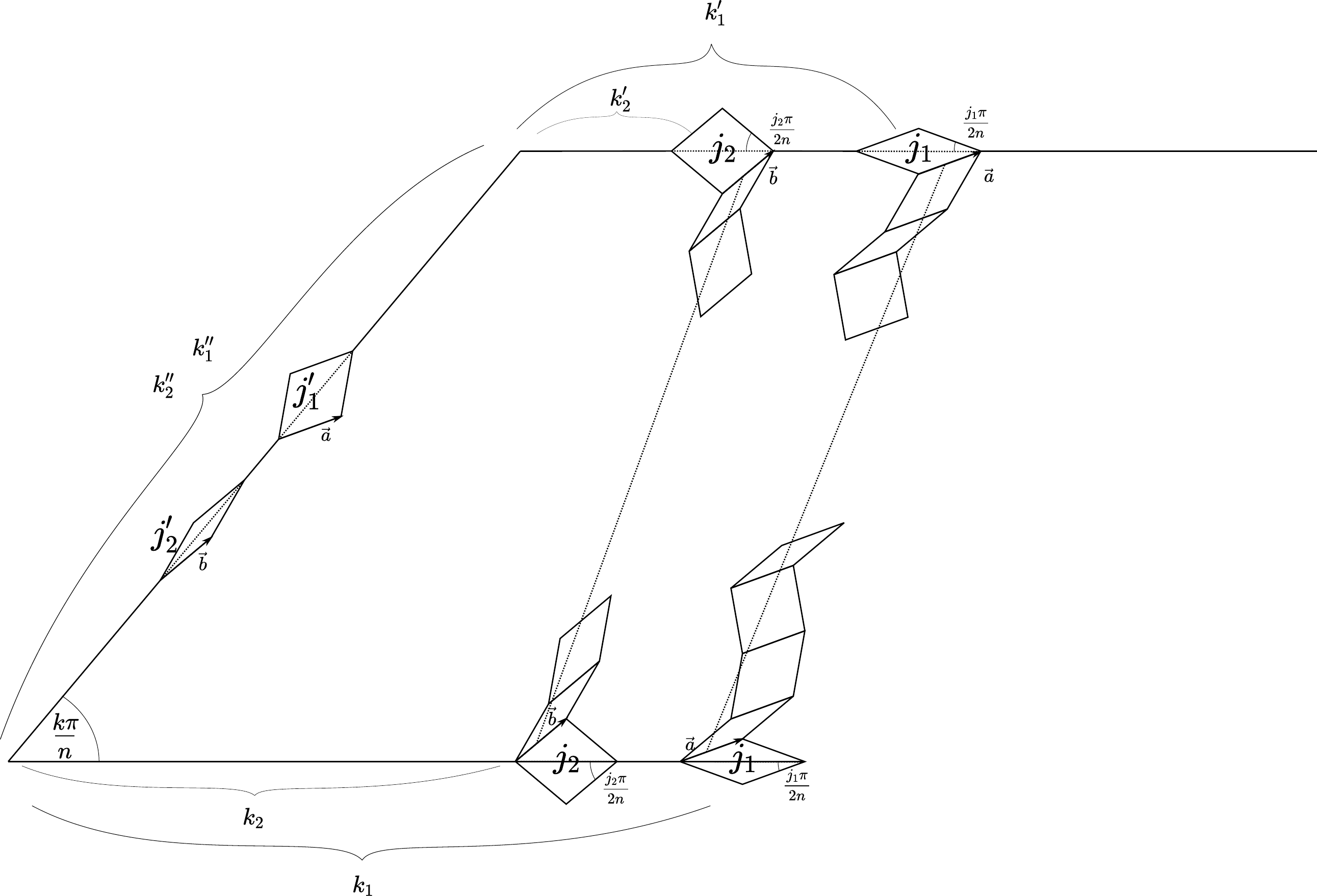}
  \caption{The two chains linking opposite sides.}
  \label{fig:two-chains-opposite}
\end{figure}

Now we modify this situation a little. Suppose that the chains connect to the opposite side as in Figure \ref{fig:two-chains-opposite}. Then we have the following proposition.

\begin{proposition}
  \label{prop:tileability02}
  If the metatile with angle $\frac{k\pi}{n}$ is tileable then for any $j_1<j_2$, and for any position $k_1$ such that
  the edgeword has symbol $j_1$ in position $k_1-1$, we have
  \begin{align} \label{eq:opposite_edge} f_{j_2}^{-1}\Big( f_{j_2}(k_1)-f_{|j_2-2k|}(m)\Big) < f_{j_1}^{-1}\Big( f_{j_1}(k_1)-f_{|j_1-2k|}(m)\Big). \end{align}
\end{proposition}
\begin{proof}
  The quantity $f_{|j_2-2k|}(m)$ is the total number of rhombuses of type $|j_2-2k|$ on the edge adjacent to the angle (see Figure \ref{fig:two-chains-opposite}). So it is the number of rhombuses of type $j_2$ that will be matched to rhombuses of type $|j_2-2k|$ on the adjacent edge. Once we remove those rhombuses, rhombuses of type $j_2$ match to rhombuses of type $j_2$ on the opposite edge, hence (\ref{eq:opposite_edge}).
\end  {proof}

Overall, if the metatile with angles $\tfrac{k\pi}{n}$ and $\tfrac{(n-k)\pi}{n}$ is tileable then for any $j_1<j_2$, and for any position $k_1$ such that the edgeword has symbol $j_1$ in position $k_1-1$, we have (\ref{eq:adjacent_edge}) and (\ref{eq:opposite_edge}) for the both angles. So, in total, four inequalities hold.

Let us see how the converse goes.
\begin{proposition}
  \label{prop:cornersandtileability}
  Let us consider the metatile of angles $\frac{k\pi}{n}$ and $\frac{(n-k)\pi}{n}$ with $k$ odd.
  Assume that $f_1\geq f_3 \geq \dots \geq f_{n-2}$, and assume that for any odd $j_1<j_2$, and for any position $k_1$ such that the edgeword has symbol $j_1$ in position $k_1-1$, we have
  \begin{align}
    f_{|j_2-2k|}^{-1}\circ f_{j_2}(k_1)&<f_{|j_1-2k|}^{-1}\circ f_{j_1}(k_1), \label{eq:counting_adjacent1}\\
    f_{|j_2-2(n-k)|}^{-1}\circ f_{j_2}(k_1)&<f_{|j_1-2(n-k)|}^{-1}\circ f_{j_1}(k_1), \label{eq:counting_adjacent2}\\
    f_{j_2}^{-1}\left(f_{j_2}(k_1)-f_{|j_2-2k|}(m)\right) &< f_{j_1}^{-1}\left(f_{j_1}(k_1)-f_{|j_1-2k|}(m)\right),\label{eq:counting_opposite1}\\
    f_{j_2}^{-1}\left(f_{j_2}(k_1)-f_{|j_2-2(n-k)|}(m)\right) &< f_{j_1}^{-1}\left(f_{j_1}(k_1)-f_{|j_1-2(n-k)|}(m)\right).\label{eq:counting_opposite2}
  \end{align}
  Then the metatile is tileable.
\end{proposition}

\begin{proof}[Proof of Proposition \ref{prop:cornersandtileability}]
  By contradiction suppose the metatile is not tileable. We will prove that one of the inequalities is broken.
  If the metatile is not tileable, then by Theorem \ref{thm:kenyon} there exists an invalid crossing of chains. Any two chains that cross are of one of the types described in Figure \ref{fig:atlascrossings}.
  \begin{figure}
    \center
    \includegraphics[width=\textwidth]{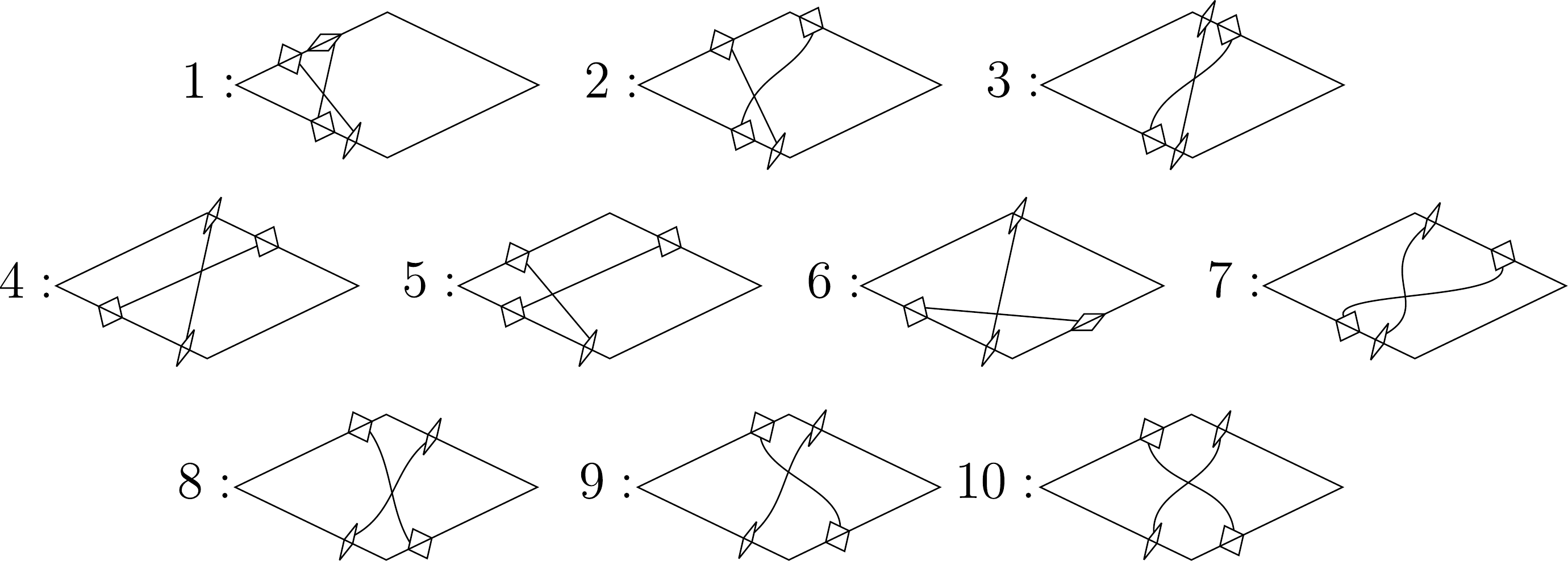}
    \caption{Atlas of possible (but not necessarily valid) crossing types.}
    \label{fig:atlascrossings}
  \end{figure}
  The first case is the one described in the proof of Proposition \ref{prop:tileability01} and in Figure \ref{fig:two-chains}. In particular, as explained in the proof of Proposition \ref{prop:tileability01}, the existence of an invalid crossing of type 1 is equivalent with a broken inequality (\ref{eq:counting_adjacent1}). Case 2 actually also reduces to the same broken inequality. So if the invalid crossing is of type 1 or type 2 then (\ref{eq:counting_adjacent1}) is broken which contradicts the hypothesis.\\
  Type 3 is the case discussed in the proof of Proposition \ref{prop:tileability02} and corresponds to (\ref{eq:counting_opposite1}). So an invalid crossing of type 3 implies a broken inequality (\ref{eq:counting_opposite1}). \\
  Types 4, 5 and 6 are never invalid. Indeed, by their definition it is a crossing between chains that start from faces of rhombuses that face each other so the crossing is always valid.\\
  Type 7: Let us first consider the sub-case where the two pairs of rhombuses whose chains cross are of the same type. This means that we have a situation where the two chains leaving from the two sides of a rhombus are crossing. Let us show it is never the case. For simplicity we assume $k<n-k$, and let $\tfrac{i\pi}{n}$ be the (odd) angle of the rhombus in question. If $n-k-2i < n/2$ then the two directions of the sides of the rhombus have matching rhombuses on the adjacent edges of the metatile. If we call $x$ the position of the rhombus, $x'$ the position of the rhombus to which the left side matches and $x''$ the position of the rhombus to which the right side matches, then $x'<x<x''$.
  But if we have $n-k-2i>n/2$, assuming that the wide angle of the metatile is to the right of the rhombus, we have $x''<x$ because on the adjacent edge of the metatile there are rhombuses of type $2k+i$ that have the exact same edge direction as our rhombus. So the rhombus $x$ matches to a rhombus $x''$ which is more to the left. But the left side of the rhombus matches to $x'<x$, and even $x'<x''<x$, because there are strictly more rhombuses of type $|i-2k|$ (to which the left direction matches) on the left edge of the metatile than there are rhombuses of type $i+2k$ on the right edge of the metatile, so that $x'$ is offset by strictly more positions of rhombuses of type $i$ to the left of $x$ than $x''$. So the two sides of a same rhombus cannot spawn crossing worms.

  Now consider that we have a crossing of type 7 with two different rhombuses. The worm leaving the right side of the left rhombus crosses the worm leaving the right side of the right rhombus (otherwise the two worms leaving the left rhombus would cross). But also the worm leaving the left side of the right rhombus crosses the worm leaving the left side of the left rhombus (for the same reason), and at least one of these two crossings is invalid and of Type 2 or 3. So by the previous cases a crossing of type 7 is impossible.

  For the following cases, consider that the left rhombus is of type $i$, the right rhombus is of type $j$, and the edges of interest are $\vec{a}$ for the left rhombus and $\vec{b}$
  for the right one.\\
  Type 8: Such a crossing is invalid only if $ \vec{a}^\bot\cdot\vec{b}\trans < 0$. This is only possible if $i>k$ or $j>k$, if we consider that the angle on the metatile between the two rhombuses is $\tfrac{k\pi}{n}$. For simplicity, assume $i>k$. In this case the edge direction $a$ is matched on the adjacent edge of the metatile to rhombuses of type $|2k-i|<i$. There are more vectors of type $|2k-i|$ than of type $i$, so the rhombus cannot be matched to a rhombus of the opposite edge of the metatile. So such a crossing cannot be invalid.\\
  Type 9: Such a crossing is invalid only if $\vec{a}^\bot\cdot\vec{b}\trans < 0$. This is only possible if $j>k$, if we consider that the angle on the metatile between the two rhombuses is $\tfrac{(n-k)\pi}{n}$. In that case the edge direction $b$ is matched on the adjacent edge of the metatile to rhombuses of type $|2k-j|<j$. There are more vectors of type $|2k-j|$ than of type $j$, so the rhombus cannot be matched to a rhombus of the opposite edge of the metatile. So such a crossing cannot be invalid.\\
  Type 10: Such a crossing is invalid only if $ \vec{a}^\bot\cdot\vec{b}\trans < 0$. This is only possible if $i>k$ or $j>k$, if we consider that the angle on the metatile between the two rhombuses is $\tfrac{(n-k)\pi}{n}$. For simplicity, let us assume $i>k$. In this case the edge direction $a$ is matched on the adjacent edge of the metatile to rhombuses of type $|2k-i|<i$. There are more vectors of type $|2k-i|$ than of type $i$ so the rhombus cannot be matched to a rhombus of the opposite edge of the metatile. So such a crossing cannot be invalid.

\end{proof}

We can actually do better than this: under the almost-balancedness assumption we only need to prove (\ref{eq:counting_adjacent1}) with $k=1$.
First let us observe that the (\ref{eq:counting_adjacent1}) with $k=1$ implies that there are more rhombuses $j_1$ than rhombuses $j_2$ whenever $j_1<j_2$.

\begin{lemma}
  \label{lemma:f1geqf3}
  Assume that for any odd $j_1<j_2$, and for any position $k_1$ such that the edgeword has symbol $j_1$ in position $k_1-1$,
  \begin{align}
    f_{|j_2-2|}^{-1}\circ f_{j_2}(k_1)&<f_{|j_1-2|}^{-1}\circ f_{j_1}(k_1). \label{eq:kequals1}
  \end{align}
  Then we have $f_1\geq f_3 \geq \dots \geq f_{n-2}$.
\end{lemma}
\begin{proof}[Proof of Lemma \ref{lemma:f1geqf3}]
  Let us first recall that the edgeword $u$ is a palindrome. This result would not hold otherwise.
  Take any odd $j<n-2$ and assume that $f_1\geq f_3\geq \dots \geq f_j$. Let us prove that $f_j\geq f_{j+2}$. By contradiction, suppose that there exists $k_2$ such that $f_{j+2}(k_2) \geq 1+f_j(k_2)$. We take $k_2$ to be the smallest index at which this holds. We have $u_{k_2-1}=j+2$. Take $k_1$ to be the smallest index greater than $k_2$ at which $u_{k_1-1}=j$. Such an index exists because by (\ref{eq:kequals1}) the first occurrence of $j$ is before the first occurrence of $j+2$, and so the last occurrence of $j$ is after the last occurrence of $j+2$.
  By definition, $f_j(k_1)=f_j(k_2)+1$ and $f_{j+2}(k_1)\geq f_{j+2}(k_2)$, so $f_j(k_1)\leq f_{j+2}(k_1)$.
  But by $f_{|j-2|} \geq f_j$ we have $f_{|j-2|}^{-1}(f_j(k_1))\leq k_1$.
  This is impossible because when we apply (\ref{eq:kequals1}) with $j_1=j$ and $j_2=j+2$ on $k_1$, we get
  $$ f_{j}^{-1}\circ j_{j+2}(k_1) < f_{j-2}^{-1}\circ f_j(k_1) \leq k_1,$$ which implies $f_{j+2}(k_1) < f_j(k_1)$. \\
\end{proof}
From this lemma we can obtain this nice result on the counting functions.
\begin{lemma}
  \label{lemma:k1allk}
  If (\ref{eq:counting_adjacent1}) holds for $k=1$ then (\ref{eq:counting_adjacent1}) and (\ref{eq:counting_adjacent2}) hold for any $1\leq k < n$.\\
\end{lemma}
\begin{proof}
  Let us prove that if (\ref{eq:counting_adjacent1}) holds for $1\leq k < k_0$ then it also holds for $k_0+1$.
  In other words, we assume that for any $1\leq k < k_0$, any $j_1 < j_2$ and any $k_1$ such that $u_{k_1-1  }=j_1$ we have $$ f_{|j_2-2k|}^{-1}\circ f_{j_2}(k_1)<f_{|j_1-2k|}^{-1}\circ f_{j_1}(k_1),$$
  and we want to prove that for any $j_1<j_2$ and any $k_1$ such that $u_{k_1-1}=j_1$ we have $$ f_{|j_2-2(k_0+1)|}^{-1}\circ f_{j_2}(k_1)<f_{|j_1-2(k_0+1)|}^{-1}\circ f_{j_1}(k_1).$$
  Let us take $j_1,j_2,k_1$ such that $j_1<j_2$ and $u_{k_1-1}=j_1$.
  Let us consider four cases: $j_1 < j_2 \leq 2k_0$, $j_1\leq 2k_0$ and $j_2 = 2k_0+1$, $j_1\leq 2k_0$ and $2k_0+1 < j_2$, and  $2k_0 \leq j_1 < j_2$.
  \begin{itemize}
  \item Case $j_1 < j_2 \leq 2k_0$. We have $|j_1 - 2(k_0 +1)| = 2(k_0+1)-j_1$ and $|j_2-2(k_0+1)| = 2(k_0+1)-j_2$. From this we get that $|j_1 - 2(k_0 +1)| > |j_2 - 2(k_0 +1)|$ and from Lemma \ref{lemma:f1geqf3} we get that $f_{j_1} \geq f_{j_2}$ and $f_{|j_1 - 2(k_0 +1)|} \leq f_{|j_2 - 2(k_0 +1)|}$. We can then deduce from it that $f^{-1}_{|j_1 - 2(k_0 +1)|} \geq f^{-1}_{|j_2 - 2(k_0 +1)|}$ which in turn implies that $f^{-1}_{|j_2 - 2(k_0 +1)|}\circ f_{j_2} \leq f^{-1}_{|j_1 - 2(k_0 +1)|}\circ f_{j_1} $. In particular, if $u_{k_1-1} = j_1$ then we get $  f^{-1}_{|j_2 - 2(k_0 +1)|}\circ f_{j_2}(k_1) < f^{-1}_{|j_1 - 2(k_0 +1)|}\circ f_{j_1}(k_1)$.

  \item Case $j_1 \leq 2k_0$ and $j_2 = 2k_0+1$. We have $|j_1 - 2(k_0 +1)| = 2(k_0+1)-j_1 > |j_1-2k_0|$ and  $|j_2-2(k_0+1)| = 1 = |j_2-2k_0|$. From this (and Lemma \ref{lemma:f1geqf3}) we get $f^{-1}_{|j_1-2(k_0+1)|} \geq f^{-1}_{|j_1-2k_0|}$ and $f^{-1}_{|j_2-2(k_0+1)|} = f^{-1}_{|j_2-2k_0|}$. From this we get
  $f^{-1}_{|j_2 - 2(k_0 +1)|}\circ f_{j_2}(k_1) < f^{-1}_{|j_1 - 2(k_0 +1)|}\circ f_{j_1}(k_1)$.

  \item Case $j_1\leq 2k_0$ and $2k_0+1 < j_2$. We have $|j_1 - 2(k_0 +1)| = 2(k_0+1)-j_1 > |j_1-2k_0|$ and $|j_2 - 2(k_0 +1)| = j_2 - 2(k_0+1) < |j_2-2k_0|$. From this (and Lemma \ref{lemma:f1geqf3}) we get $f^{-1}_{|j_1-2(k_0+1)|} \geq f^{-1}_{|j_1-2k_0|}$ and $f^{-1}_{|j_2-2(k_0+1)|} \leq f^{-1}_{|j_2-2k_0|}$. From this we get $f^{-1}_{|j_2 - 2(k_0 +1)|}\circ f_{j_2}(k_1) < f^{-1}_{|j_1 - 2(k_0 +1)|}\circ f_{j_1}(k_1)$.

  \item Case $2k_0 \leq j_1 < j_2$. We have  $|j_1 - 2(k_0 +1)| = | |j_1-2k_0| - 2|$ and $|j_2 - 2(k_0 +1)| = | |j_2-2k_0| - 2|$. We then decompose $$f_{|j_1-2(k_0+1)|}^{-1}\circ f_{j_1}(k_1) =f_{||j_1-2k_0|-2|}^{-1}\circ f_{|j_1-2k_0|}\circ f_{|j_1-2k_0|}^{-1}\circ f_{j_1}(k_1).$$ We can then apply our assumption for $k=k_0$ and for $k=1$ to obtain $f^{-1}_{|j_2 - 2(k_0 +1)|}\circ f_{j_2}(k_1) < f^{-1}_{|j_1 - 2(k_0 +1)|}\circ f_{j_1}(k_1)$.
  \end{itemize}

\end{proof}
To also get (\ref{eq:counting_opposite1}) and (\ref{eq:counting_opposite2}) we need the almost-balancedness assumption.
Under the assumption that $u$ is 2-almost-balanced we always have (\ref{eq:counting_opposite1}) and (\ref{eq:counting_opposite2}).

\begin{lemma}
  Assume that $u$ is 2-almost-balanced (with our Definition \ref{def:letters}) and that $f_1(m)>f_3(m)>\dots > f_{n-2}(m) > 0$ with $m=|u|$ the length of the whole word. Then for any $0<k<n$, for any $j_1<j_2$ and for any $k_1$ such that $u_{k_1-1}=j_1$ and $f_{j_1}(k_1)>f_{|j_1-2k|}(m)$ we have
  $$ f_{j_2}^{-1}\left(f_{j_2}(k_1)-f_{|j_2-2k|}(m)\right) < f_{j_1}^{-1}\left(f_{j_1}(k_1)-f_{|j_1-2k|}(m)\right). $$
  \label{lemma:counting_opposite}
\end{lemma}
\begin{proof}
 Let us take such $k,j_1,j_2,k_1$.
  Let us first solve two special cases:
  \begin{itemize}
  \item If $j_1\geq |j_1-2k|$ then we never have $f_{j_1}(k_1)>f_{|j_1-2k|}(m)$. Indeed, we have $f_{j_1}(k_1) \leq f_{j_1}(m) \leq f_{|j_1-2k|}(m)$ by $j_1\geq |j_1-2k|$, so there exist no such $k,j_1,j_2,k_1$.
  \item If $|j_1-2k|> n-2$ then there is no rhombus of angle $|j_1-2k|$ in the edgeword, and either $|j_1-2k|=n$ in which case $f_{|j_1-2k|}(m)=0$, or $|j_1-2k|> n$ in which case  $f_{|j_1-2k|}(m)= - f_{2n-|j_1-2k|}(m) < 0$. This second special case is actually very similar to the the first one because $f_1(m)>f_3(m)>\dots > f_{n-2}(m)> f_n(m)> f_{n+2}(m) > \dots > f_{2n-2}(m)$ with the definition of $f_i(m)=-f_{2n-i}(m)$ for $i>n$.
  \end{itemize}
  If we are not in these cases then  $|j_1-2k|>j_1$. This means that $|j_1-2k|=2k-j_1$ and $j_1<k$.

  If $j_2\geq|j_2-2k|$ then the inequality is simple with
  \begin{align*} & f_{j_2}(k_1)-f_{|j_2-2k|}(m)\leq 0 \Rightarrow  \\ &\qquad f_{j_2}^{-1}\left(f_{j_2}(k_1)-f_{|j_2-2k|}(m)\right) = 0 < f_{j_1}^{-1}\left(f_{j_1}(k_1)-f_{|j_1-2k|}(m)\right). \end{align*}
  Otherwise we have $j_1<j_2<k<|j_2-2k|<|j_1-2k|\leq n-2$. In that case we have $f_{|j_2-2k|}(m)>f_{|j_1-2k|}(m)$.
  Since $f_{j_1}(k_1)-f_{|j_1-2k|}(m)>0$ there exists a $k_1'<k_1$ such that $u_{k_1'}=j_1$ and  $f_{j_1}(k_1')= f_{j_1}(k_1)-f_{|j_1-2k|}(m)$.\\
  Let us now remark that $f_{j_2}(k_1)-f_{|j_2-2k|}(m) \leq f_{j_2}(k_1)-(f_{|j_1-2k|}(m)+1) \leq f_{j_2}(k_1')$
  because $|u_{k_1'}u_{k_1'+1}\dots u_{k_1-1}|_{j_2} \leq 2+|u_{k_1'}u_{k_1'+1}\dots u_{k_1-1}|_{j_1}$ (by Definition \ref{def:letters})
  so $|u_{k_1'}u_{k_1'+1}\dots u_{k_1-1}u_{k_1}|_{j_2} \leq 1+|u_{k_1'}u_{k_1'+1}\dots u_{k_1-1}u_{k_1}|_{j_1}$ and \\
  \begin{align*}
    f_{j_2}(k_1)-f_{j_2}(k_1') &= |u_{k_1'}u_{k_1'+1}\dots u_{k_1-1}u_{k_1}|_{j_2} \\
    &\leq 1+|u_{k_1'}u_{k_1'+1}\dots u_{k_1-1}u_{k_1}|_{j_1} \\
    &\leq 1 + f_{j_1}(k_1)-f_{j_1}(k_1') \\
    &\leq 1 + f_{|j_1-2k|}(m).
    \end{align*}
  So overall
  \begin{align*}f_{j_2}^{-1}\left(f_{j_2}(k_1)-f_{|j_2-2k|}(m)\right) & \leq f_{j_2}^{-1}\circ f_{j_2}(k_1')  \\ &< f_{j_1}^{-1}\circ f_{j_1}(k_1) = f_{j_1}^{-1}\left(f_{j_1}(k_1)-f_{|j_1-2k|}(m)\right).\end{align*}
\end{proof}

\section{Planar Rosa: substitution discrete planes with $2n$-fold rotational symmetry}
\label{sec:planar-rosa}
In this Section we present the construction of the Planar Rosa substitution tiling that is also discrete planes. We then give the proof of Theorem \ref{th:main}.
Let $n$ be an odd integer greater than three.

\paragraph{Lifting to $\mathbb{R}^n$}
As throughout this article, we consider vertex-hierarchic substitutions $\sigma$ such that the image of any edge by the substitution is the same up to rotation and translation.
Just as in Section \ref{sec:subrosa} we lift the tilings to $\R{n}$, which we decompose in $\lfloor \frac{n}{2} \rfloor$ planes $\e{j}$ and a line $\Delta$.
As seen in Section \ref{sec:subrosa}, the expansion $\phi$ admits the planes $\e{j}$ and the line $\Delta$ as eigenspaces, with eigenvalues $\lambda_j$ and $\lambda_\Delta$
given by Lemma~\ref{lemma:planarrosa_eigenvalues}.

The substitutions we study are defined by their edgeword $u = u_0\dots u_l$ with letters in the alphabet of odd numbers $\{1,3,5,7,\dots n-2\}$, where the symbol $k\in 2\mathbb{N}+1$
represents the rhombus with angles $\frac{k\pi}{n}$ and $\frac{(n-k)\pi}{n}$.
As seen in Section \ref{sec:planar}, the abelianized edgeword $[u] = ( |u|_1 , |u|_3 , \dots, |u|_{n-2} ) = ( [u]_0, [u]_1 , \dots , [u]_{n-1} )$ determines $\phi$.

We want $\sigma$ to be planar and this translates to conditions on eigenvalues of $\phi$, as stated in the following proposition. 
For more details and illustrations on the proof we refer to~\cite[\S 3]{lutfalla2021thesis}.

\begin{proposition}[Eigenvalues and planarity]
\label{prop:eigenvalue_planarity}
Let a substitution $\sigma$ have the expansion $\phi$.
If the eigenvalues of $\phi$ are such that $|\lambda_0|>1$,\, $|\lambda_i|<1$ for $1\leq i <\lfloor \tfrac{n}{2}\rfloor$ and $|\lambda_\Delta|<1$ then $\sigma$ is planar of slope $\e{0}$, \ie, tilings legal for $\sigma$ are discrete planes of slope $\e{0}$.
\end{proposition}

\begin{proof}
  Let us prove that there exists $\delta$ such that for all $n\in\mathbb{N}$ and for any tile $t$, the metatile $\sigma^n(t)$ has $\e{0}^\bot$-diameter less than $\delta$.
  Here, the $S$-diameter of a set $X$, noted $\diam_S(X)$, is the diameter of the orthogonal projection of $X$ on the subspace $S$.
  Since $\R{n}$ is the orthogonal direct sum of $\Delta$ and $\e{j}$ over $0\leq j <\lfloor \tfrac{n}{2}\rfloor$,
  we can prove separately that $\sigma^n(t)$ has a finite $\Delta$-diameter and finite $\e{j}$-diameters for $j\neq 0$.
  These are the subspaces where the corresponding eigenvalue has modulus less than one.

  Let $\ea$ be one of the spaces $\Delta$ and $\e{j}$, $j\neq 0$, and let $\lambda$ be the corresponding complex eigenvalue. We have that $|\lambda|<1$.
  The set  $\mathbf{T}$ of prototiles is finite so we can define the number
  $$\delta_{\ea} := \max\{ \diam_{\ea}(\sigma(t))\ |\ t\in\mathbf{T}\}.$$
  Consider any $t\in\mathbf{T}$ and any $n\in\mathbb{N}$. There exist $x_0,y_0\in\sigma^{n+1}(t)$ such that
  $$\diam_{\ea}(\sigma^{n+1}(t)) = \sup \{\|\Pi_{\ea}(x-y)\ |\ x,y \in \sigma^{n+1}(t)\} = \|\Pi_{\ea}(x_0 - y_0)\|,$$
  where $\Pi_{\ea}$ denotes the orthogonal projection operator onto the space $\ea$.
  Since $x_0$ is in $\sigma^{n+1}(t)$ it is in some metatile $\sigma(t')$ for $t'\in \sigma^{n}(t)$. So we decompose $x_0=x_1+x_2$  with $x_1 \in \phi(\sigma^n(t))$
  and $x_2\in \sigma(t')$ for some $t'\in\mathbf{T}$. Here, $x_1$ is the corner of a metatile of order 1 to which $x_0$ belongs and $x_2$ is the relative position of $x_0$ in this metatile.
  Similarly we decompose  $y_0 = y_1+y_2$ for $y_1 \in \phi(\sigma^n(t))$
  and $y_2\in \sigma(t'')$ for some $t''\in\mathbf{T}$. Now
  \begin{align*}
  \diam_{\ea}(\sigma^{n+1}(t)) &=
    \|\Pi_{\ea}(x_0-y_0)\| \\  &\leq \|\Pi_{\ea}(x_1-y_1)\| + \|\Pi_{\ea}(x_0-x_1)\| + \|\Pi_{\ea}(y_0-y_1)\| \\ &\leq |\lambda|\diam_{\ea}(\sigma^n(t)) + 2\delta_{\ea}.
    \end{align*}
  Iterating this gives that $\diam_{\ea}(\sigma^{n+1}(t)) \leq \frac{2\delta_{\ea}}{1-|\lambda|}$ because $0\leq |\lambda|<1$. Thus we have a desired finite bound on the $\ea$-diameters.

  Overall, we now have that the $\e{0}^\bot$-diameter of $\sigma^n(t)$ is bounded by the sum
  $$
  \delta=\frac{2\delta_{\Delta}}{1-|\lambda_\Delta|}+\sum_{j=1}^{\lfloor\tfrac{n}{2}\rfloor-1} \frac{2\delta_{\e{j}}}{1-|\lambda_j|}
  $$
  of the $\Delta$-diameter and the $\e{j}$-diameters.

 To finish the proof, take a tiling $\tiling$ legal for $\sigma$, and suppose that there exists no $r$ such that $\tiling$ is a discrete plane of slope $\e{0}$ and thickness $r$. This means that there exists a sequence $\delta_n$ such that $\delta_n \to \infty$ and there exist vertices $v_n$ and $v_n'$ of the tiling such that $d_{\e{0}^\bot}(v_n,v_n')=\delta_n$. Let $B_n$ be a patch in $\tiling$ that contains $v_k$ and $v_k'$ for all $k<n$. By the definition of a tiling admissible for $\sigma$, there exists a tile $t_n$ and an integer $m_n$ such that the pattern $B_n$ appears in $\sigma^{m_n}(t_n)$. Therefore there exist a sequence of metatiles with unbounded $\e{0}^\bot$-diameter. This is impossible from the result we just proved.
 So any tiling admissible for $\sigma$ is a discrete plane of slope $\e{0}$.
\end{proof}

\paragraph{Rhombus frequencies}
Let us consider a substitution with word $u=u_0\dots u_m$ and eigenvalues (in modulus) $|\lambda|= \left(|\lambda_i|\right)_{0\leq i < \lfloor\tfrac{n}{2}\rfloor}$.
What constraints on $u$ would ensure that $\sigma$ is admissible for planarity?

\begin{definition}[Optimal frequency vector $\gamma$]
  Let us define the optimal frequency vector $\gamma$ as $$\gamma := \left(\cos(\tfrac{(2i+1)\pi}{n})\right)_{0\leq i < \floor{\frac{n}{2}}}.$$
\end{definition}

\begin{lemma}[Approximating $\langle\gamma\rangle$]
  There exists $\epsilon>0$ such that $\forall x\geq 2,\ \forall u \in \{1,3, \dots ,n-2\}^*$ \[ d([u],\tfrac{x\gamma}{2\|\gamma\|^2}) < \epsilon \implies d( \underbrace{ | \eigenmatrix{n}\cdot[u]\trans |}_{|\lambda|}, (x,0,\dots,0)) < 1  \implies \begin{cases}|\lambda_0|>1 \\ |\lambda_1| < 1 \\ \vdots\\ |\lambda_{\lfloor \frac{n}{2}\rfloor-1}| < 1 \end{cases} \]
      \label{lem:approxgamma}
\end{lemma}
\begin{proof}
  Recall that $|\lambda| = \left|  \eigenmatrix{n}\cdot[u]\trans \right|$ by Lemma \ref{lemma:planarrosa_eigenvalues}.

  The first implication is a direct consequence of the uniform continuity of linear functions (here a matrix-vector product) in $\mathbb{R}^{\lfloor \frac{n}{2}\rfloor}$ and $$ \eigenmatrix{n}\cdot\gamma\trans = 2\|\gamma\|^2(1,0,\dots,0)\trans.$$ This is due to the fact that $\eigenmatrix{n}$ is orthogonal (up to a scaling vector, see Lemma \ref{lemma:orthogonal} in the Appendices), and to the fact that $2\gamma$ is equal to the first row of $\eigenmatrix{n}$.

    The second implication is simply a consequence of $x\geq 2$.
\end{proof}

\paragraph{Tracking $\langle\gamma\rangle$}
The idea now is to find a sequence of edgewords $u$ such that $[u]$ approximates or tracks the line $\langle\gamma\rangle$.
It is always possible to track a line with integer points, but let us define
a specific sequence of points that approximates the line $\langle \gamma \rangle$.

\begin{definition}[Billiard line and sequence]
\label{def:seq_p_and_omega}
Let $\Gamma$ be the line $\langle \gamma \rangle$ and let $\Gamma_{\frac{1}{2}}$ be the line $\langle \gamma \rangle + (\tfrac{1}{2},\dots \tfrac{1}{2})$.

Let $\omega$ be the bi-infinite billiard word of line $\Gamma_{\frac{1}{2}}$ centred on $(\tfrac{1}{2},\dots \tfrac{1}{2})$. This means that we build the bi-infinite word $\omega$ by travelling the line and adding a letter $2i+1$ each time it crosses an hyperplane of type $H_{i,k}:= \{ \vec{x}\in \mathbb{R}^{\lfloor \frac{n}{2} \rfloor} \ |\ \vec{x} \cdot \vec{e_i}\trans = k \}$ with $k\in \mathbb{Z}$ and $\vec{e_i}$ a vector of the canonical basis of $\mathbb{R}^{\lfloor \frac{n}{2} \rfloor}$. (Note that when the line crosses an hyperplane of normal $\vec{e_i}$ we add a letter $2i+1$ instead of the classical letter $i$ because the letters of $\omega$ represent rhombuses on the boundary of the substitution's metatiles.)

Let $(p_i)_{i \in \mathbb{Z}}$ be the sequence of points of $\mathbb{Z}^{\lfloor \frac{n}{2} \rfloor}$ associated to word $\omega$ with $p_0 = 0$. This means that for all $j$, $$\omega_j = 2i+1 \qquad \Longrightarrow \qquad p_{j+1}-p_j = \vec{e_i}.$$
\end{definition}

The choice of using the billiard word $\omega$ of the line $\Gamma_{\frac{1}{2}}$ instead of the line $\Gamma$ is motivated by the fact that this exact sequence $\omega$ of rhombuses  appears in the DeBruijn multigrid tiling $P_n(\tfrac{1}{2},\dots \tfrac{1}{2})$, as seen in Theorem \ref{thm:Gn_regular} below. This will give us a tool to prove tileability, as seen in Proposition \ref{prop:cone_tileable}.

Before we prove that the sequence $(p_i)$ approximates the line $\gamma$, let us present the exact position of the intersection points of the line $\Gamma_{\frac{1}{2}}$ with the hyperplanes.

\begin{lemma}
For a vector $\vec{e_i}$ of the canonical basis and an integer $k$, the unique element of $ \Gamma_{\frac{1}{2}}\cap H_{i,k}$ is $(\tfrac{1}{2},\dots \tfrac{1}{2}) + \frac{k-\tfrac{1}{2}}{\cos\left(\frac{(2i+1)\pi}{2n}\right)}\gamma$.
\end{lemma}

\begin{proof}
This comes from the fact that the $i^{\text{th}}$ coordinate of vector $\gamma$ is $\gamma_i = \cos\left(\frac{(2i+1)\pi}{2n}\right)$ and that the vector from $(\tfrac{1}{2},\dots \tfrac{1}{2})$ to $H_{i,k}$ is $(k-\tfrac{1}{2})\vec{e_i}$.
\end{proof}

\begin{proposition}
\label{prop:tracking_gamma}
The sequence $(p_i)_{i\in\mathbb{N}}$ approximates the line $\Gamma$, which means that for any positive $\epsilon$ there are infinitely many points $p_i$ that are $\epsilon$-close to the line:
$$\forall \epsilon > 0,\ \exists^\infty i > 0,\ d(p_i,\Gamma)< \epsilon.$$

\end{proposition}
\begin{proof}
  In this proof we can see $(p_i)_{i\in\mathbb{N}}$ either simply as a billiard word or as a cut-and-project line. We will present first a proof using the fact that it is cut-and-project line. However, since the reader might not be familiar with cut-and-project sets we will also present a proof that only uses the definition of billiard words and their links with dynamical systems on the torus.
\begin{enumerate}
  \item We can see $\{p_i\ |\ i \in\mathbb{N}\}$ as a cut-and-project line of slope $\langle\gamma\rangle$ in which case we can use the Theorem 7.2 of \cite{baake2013}. We
      rephrase the theorem in the formalism of \cite{harriss2004}, which is more adapted to our specific setting.
\begin{theorem}[\cite{baake2013}]
Let $\Lambda$ be a canonical cut and project set with the cut-and-project scheme $(\mathcal{V}, \mathcal{W}, \mathcal{R}, \mathbb{Z}^n)$ where $\mathcal{V}\oplus \mathcal{W} \oplus \mathcal{R}= \mathbb{R}^n$.
Let $(x_i)_{i\in \mathbb{N}}$ be an exhaustive sequence of the points of $\Lambda$ ordered by increasing norm, \emph{i.e.}, $\Lambda = \{ x_i, i \in \mathbb{N}\}$ and $\forall i,\ \|x_i+1\| \geq \|x_i\|$.
Then the sequence $\left(\Pi_{\mathcal{V}^\bot}(x_i)\right)_{i\in \mathbb{N}}$ is uniformly distributed in the window $W$.
\end{theorem}
From the fact that $p_0$ projects to $0$ we get that $0\in W$, so that from the  uniform distribution of $\left(\Pi_{\mathcal{V}^\bot}(p_i)\right)_{i\in \mathbb{N}}$ we
can conclude that $(p_i)_{i\in\mathbb{N}}$ approximates $\langle \gamma \rangle$.

\item Now let us consider $(p_i)_{i\in\mathbb{N}}$ simply as a billiard line and use a known result on billiard lines.
  \begin{figure}[!b]
    \center
  \begin{subfigure}{0.35\textwidth}
    \includegraphics[width=\textwidth]{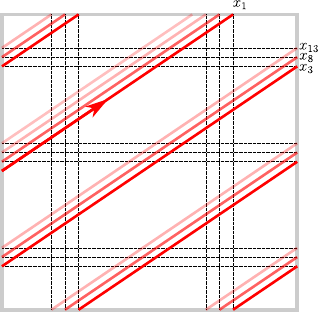}
    \caption{Folded line.}
    \label{subfig:foldedline}
  \end{subfigure}
  \qquad
  \begin{subfigure}{0.55\textwidth}
    \includegraphics[width=\textwidth]{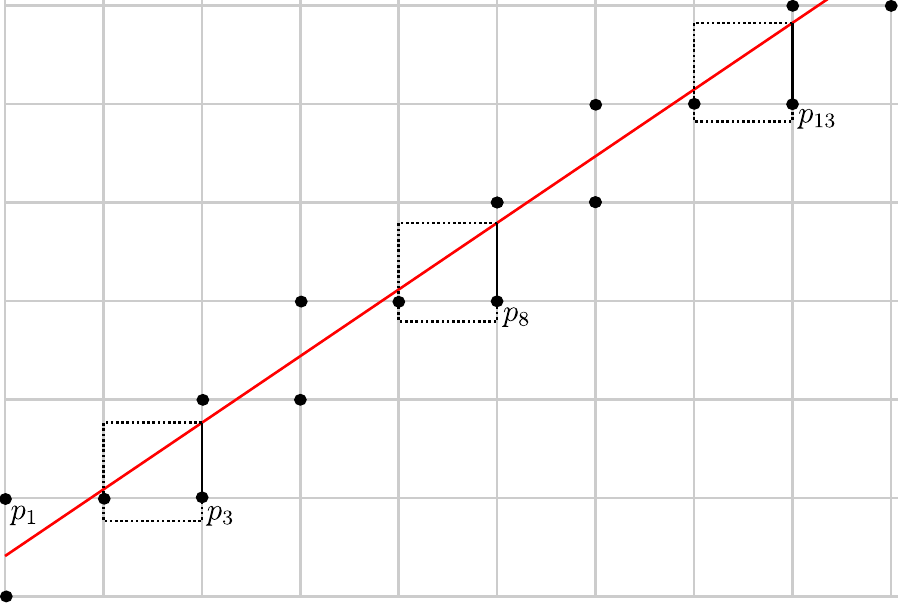}
    \caption{Billiard sequence.}
  \end{subfigure}
  \caption{Ideas for the proof of Theorem \ref{th:folk}.}
  \label{fig:billiard_line}
  \end{figure}

\begin{theorem}[folk.]
  Let $(p_i)_{i\in\mathbb{N}}$ be a billiard sequence of line $\Gamma$ and $\pi$ be the orthogonal projection onto $\Gamma^\bot$.
  Every projected point $\pi(p_i)$ is an accumulation point of the projected sequence $\left( \pi(p_i) \right)_{i \in \mathbb{N}}$.
  \label{th:folk}
\end{theorem}

\begin{proof}[Proof of Theorem \ref{th:folk}]
The idea is to first consider the folded line $\Gamma$ in the torus as seen in Figure \ref{subfig:foldedline}, and consider the sequence $(x_i)_{i\in\mathbb{N}}$ of the intersection points of $\Gamma$ with the boundary of the torus. This sequence has the property that each of its points is an accumulation point of the sequence, which can be proved using the Poincaré Recurrence Theorem applied to the translation on the torus.
Then we only need to remark that the sequence $\left(\pi(p_i)\right)$ and $\left(x_i\right)$ are
very similar: if there are two points $x_{i_1}$ and $x_{i_2}$ that are on the same hyper-facet of the torus and that are $\epsilon$-close, then the corresponding $\pi(p_{i_1})$ and $\pi(p_{i_2})$ are also $\epsilon$-close as is illustrated with points $x_3,x_8,x_{13}$ and $p_3,p_8,p_{13}$ in Figure \ref{fig:billiard_line}. The result follows.
\end{proof}
Now with $\pi(p_0) = 0$ we get that $0$ is an accumulation point of $\left( \pi(p_i) \right)_{i \in \mathbb{N}}$, which means that $(p_i)_{i\in\mathbb{N}}$ approximates $\langle\gamma\rangle$.

\end{enumerate}
\end{proof}

\begin{definition}[Candidate substitutions $\candidate{j}$]
  Now let us define the sequence of words $(u_{(j)})_{j\in \nn}$ with $u_{(j)} := pref_j(\omega)\overline{pref_j(\omega)}$ where $pref_j(\omega)$ is the prefix of length $j$ of $\omega$, \ie, $pref_j(\omega) = \omega_0\omega_1\dots \omega_{j-1}$.

  For $j\in\mathbb{N}$ we define the candidate pseudo-substitution $\candidate{j}$ as having the edgeword $u_{(j)}$.

\end{definition}
By definition the word $u_{(j)}$ is a palindrome so that the pseudo-substitution is well defined on the edges: whenever two tiles are neighbour there will be no conflict when we apply the substitution to them.
\begin{proposition}
For any $j$, the word $u_{(j)}$ is 2-almost-balanced.
\end{proposition}
\begin{proof}
Let us first remark that the infinite word $\omega$ is 1-almost-balanced, \emph{i.e.},
for any $j_1<j_2$ and any finite factor $v$ of $\omega$ we have $|v|_{j_1}-|v|_{j_2} \geq -1$.
Indeed, for $j_1<j_2$ define the projection $P_{j_1,j_2}$ from the alphabet $\{1,3,\dots n-2\}$ to the alphabet $\{a,b\}$ such that $j_1\mapsto a,\ j_2\mapsto b$ and all other letters are erased.
For any finite factor $v$ of $\omega$ we have that $P_{j_1,j_2}(v)$ is a finite factor of $P_{j_1,j_2}(\omega)$ with
$|v|_{j_1}=|P_{j_1,j_2}(v)|_a$ and $|v|_{j_2} = |P_{j_1,j_2}(v)|_b$.
The infinite word $P_{j_1,j_2}(\omega)$ is a binary billiard word over alphabet $\{a,b\}$ with frequencies $\gamma_a>\gamma_b$. So $P_{j_1,j_2}(\omega)$ is 1-balanced in the usual definition, which is that for any finite factors $v,w$ of $P_{j_1,j_2}(\omega)$ such that $|v|=|w|$ we have $\left| |v|_a - |w|_a \right| \leq 1$ and $\left| |v|_b - |w|_b\right| \leq 1$. So if there existed a finite factor $v$ of $P_{j_1,j_2}(\omega)$ such that $|v|_a - |v|_b \leq -2$, then for all factors $w$ of $P_{j_1,j_2}(\omega)$ of the same length we have
$|w|_a-|w|_b \leq 0$, which is impossible with $\gamma_a > \gamma_b$. So $|u|_{j_1}-|u|_{j_2} = |P_{j_1,j_2}(u)|_a- |P_{j_1,j_2}(u)|_b \geq -1$, and overall $\omega$ is 1-almost-balanced.

Let us now recall that $u_{(j)}=pref_j(\omega))\overline{pref_j(\omega)}$. Any factor $v$ of $u_{(j)}$ can be decomposed as $v=v'v''$ with $v'$ a (possibly empty) factor of $pref_j(\omega)$ and $v''$ a (possibly empty) factor of $\overline{pref_j(\omega)}$. So $|v|_{j_1}-|v|_{j_2} = \left(|v'|_{j_1} - |v'|_{j_2}\right) + \left(|v''|_{j_1} - |v''|_{j_2}\right) \geq -2$ because $pref_j(\omega)$ and $\overline{pref_j(\omega)}$ are 1-almost-balanced. Thus $u_{(j)}$ is 2-almost-balanced.
\end{proof}

\begin{proposition}
  There exist $j \in \mathbb{N}$ such that $\candidate{j}$ is a planar substitution of slope $\e{0}$ and such that there exist a discrete plane of slope $\e{0}$ legal for $\candidate{j}$ which has global $2n$-fold rotational symmetry.
\end{proposition}
Theorem \ref{th:main} follows from this proposition. The proof of the proposition is given in three lemmas:
We prove the planarity in Lemma~\ref{lem:biglemma1} and the tileability in Lemma~\ref{lem:biglemma2}.
The third Lemma~\ref{lem:biglemma3} shows the primitivity and the rotational symmetry when the substitution
is known to be well-defined.

\begin{lemma}
  \label{lem:biglemma1}
  There exist infinitely many $j\in\mathbb{N}$ such that the candidate substitution $\candidate{j}$ is planar of slope $\e{0}$, \ie, tilings legal for $\candidate{j}$ are discrete planes of slope $\e{0}$.
\end{lemma}

\begin{proof}[Proof of Lemma \ref{lem:biglemma1}]

  By construction we have $$ \forall j,\ [u_{(j)}] = 2[pref_j(w)] = 2\cdot p_{j+1},$$
  and therefore $$d([u_{(j)}],\langle\gamma\rangle)=2d(p_{j+1},\langle\gamma\rangle).$$
  By Lemma \ref{lem:approxgamma} and Proposition \ref{prop:eigenvalue_planarity} there exists $\epsilon>0$ such that the implication
  $$d([u_{(j)}],\langle\gamma\rangle)<\epsilon \implies \candidate{j} \text{ planar}$$
  holds for large enough $j$.
  With Proposition \ref{prop:tracking_gamma} we obtain the existence of infinitely many $j\in\mathbb{N}$ such that $\ d(p_{j+1},\langle\gamma\rangle)<\tfrac{\epsilon}{2}$, so that $d([u_{(j)}],\langle\gamma\rangle)<\epsilon$, \ie, $\candidate{j}$ is planar.
\end{proof}

\begin{lemma}
  \label{lem:biglemma2}
  There exists $N\in\mathbb{N}$ such that for all $j>N$, the candidate pseudo-substitution
  $\candidate{j}$ exteds to a well-defined substitution, meaning that the metatiles defined by $\candidate{j}$ are tileable.
\end{lemma}

\begin{proof}[Proof of Lemma \ref{lem:biglemma2}]
  For a $\candidate{j}$ to be well defined we need the metatiles of the substitution to be tileable.
  Let us recall that the edgeword of $\candidate{j}$ is $u_j=v \bar{v}$ with $v=pref_j(\omega)$.
  This means that in every corner of the metatiles there is a portion of an infinite cone with the word $v$ on both sides, as in Figure \ref{fig:metatile_schema}.
  \begin{figure}
  \center \includegraphics[width=8cm]{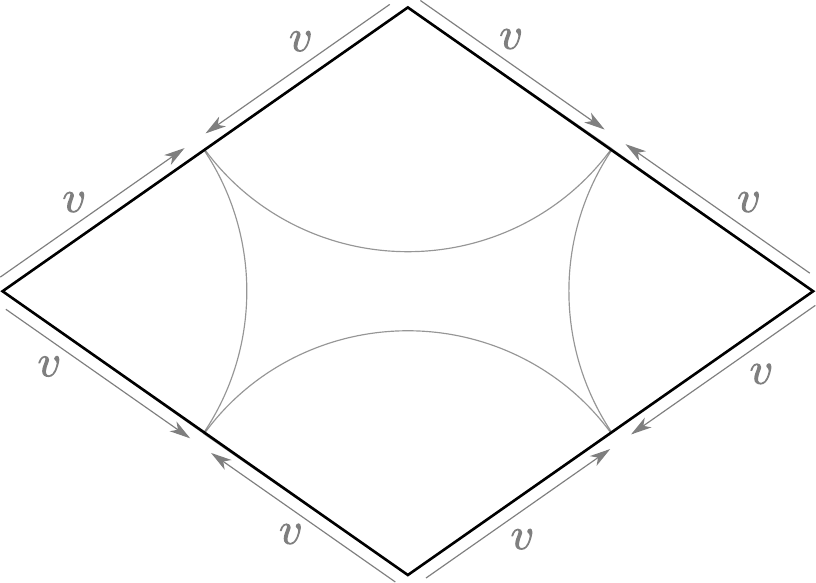}
  \caption{Sketch of the metatile, since the edgeword is a palindrome we can see that around each corner we have the same start of edgeword $v$ on both sides.}
  \label{fig:metatile_schema}
  \end{figure}
  We will first consider tileability around the corners and then we will consider tileability in the middle of the metatile. As
  $v$ is a prefix of $\omega$, let us consider the infinite cone with edgeword $\omega$.
  \begin{proposition}
  \label{prop:cone_tileable}
  For any $0<k\leq n$ the infinite cone of angle $\tfrac{k\pi}{n}$ with the edgeword $\omega$ on both sides is tileable.
  \end{proposition}
  This is actually a corollary to a result on De Bruijn multigrid tilings, let us first define the multigrid construction and then state the result.
  Let $\zeta:=\exp^{\imag\tfrac{2\pi}{n}}$, and for $\gamma_i\in\mathbb{R}$ define
  $H(\zeta^i,\gamma_i) := \left\{ z \in \mathbb{C}\ |\ \text{Re}\left(z\cdot \bar{\zeta^i}\right)+\gamma_i \in \mathbb{Z}\right\}$, called the grid of orientation $\zeta^i$ and offset $\gamma_i$. Set $H(\zeta^i,\gamma_i)$ consists of equidistant parallel lines orthogonal to $\zeta^i$.
  Let the  \emph{multigrid} of order $n$ and offset $(\gamma_0,\gamma_1,\dots , \gamma_{n-1})$ be $$G_n(\gamma_0,\gamma_1,\dots , \gamma_{n-1}) := \bigcup\limits_{i=0}^{n-1} H(\zeta^i,\gamma_i).$$
      The multigrid $G_n(\gamma_0,\gamma_1,\dots , \gamma_{n-1})$ is called \emph{regular} when no more than two lines intersect in any point, \ie, if $H(\zeta^i,\gamma_i) \cap H(\zeta^j,\gamma_j) \cap H(\zeta^k,\gamma_k) = \emptyset$ for any distinct $0\leq i,j,k < n$. Otherwise the multigrid is called \emph{singular}.
      Every multigrid is dual to an edge-to-edge polygonal tiling of the plane by the following dualization process \cite{debruijn1981, baake2013}.
To define this dualization we need the function
$f$ from $\mathbb{C}$ to $\mathbb{C}$ defined by

$$ f(z):= \sum\limits_{i=0}^{n-1}\left\lceil \text{Re}\left(z\cdot \bar{\zeta^i}\right)+\gamma_i \right\rceil \zeta^i. $$
We can remark that $f(z)$ is constant on the interior of every cell of the multigrid, so it associates to each cell a single vertex in $\mathbb{C}$.
   These vertices form the vertex set of the dual tiling where two vertices are linked by an edge when the corresponding cells of the multigrid are adjacent along an edge.
   This dual tiling denoted by $P_n(\gamma_0,\dots \gamma_{n-1})$ is a rhombus tiling whenever the multigrid is regular.\\
  \begin{theorem}[\cite{lutfalla2021}]
  \label{thm:Gn_regular}
   Let $n\geq 3$ be an odd integer.
   Let $\gamma_0,\gamma_1,\dots \gamma_{n-1}$ be non-zero rational numbers.
   The multigrid $G_n(\gamma_0,\gamma_1,\dots \gamma_{n-1})$ is regular.
   In particular, the De Bruijn multigrid $G_n(\tfrac{1}{2},\tfrac{1}{2},\dots,\tfrac{1}{2})$ is regular so the dual tiling $P_n(\tfrac{1}{2},\tfrac{1}{2},\dots,\tfrac{1}{2})$ is a rhombus tiling.
  \end{theorem}

  Let us now use Theorem \ref{thm:Gn_regular} to prove Proposition \ref{prop:cone_tileable} which in turn is used to prove Lemma \ref{lem:biglemma2}.
  Let us remark that the tiling $P_n(\tfrac{1}{2},\tfrac{1}{2},\dots,\tfrac{1}{2})$ contains a cone with edgewords $\omega$. In fact, there are $2n$ such cones around the origin. Indeed,
  the vertical half line in $G_n(\tfrac{1}{2},\tfrac{1}{2},\dots,\tfrac{1}{2})$  starting at 0 is a succession of intersection points of type $H(\zeta^i,\tfrac{1}{2})\cap H(\zeta^{n-i}, \tfrac{1}{2})$ with $0<i<n$. No other type of intersection appears. In the dual tiling $P_n(\tfrac{1}{2},\tfrac{1}{2},\dots,\tfrac{1}{2})$ it is a succession of rhombuses joined by their extremal vertices and with a common diagonal direction. On this vertical half line, crossings of lines corresponding to rhombuses with angles $\pi\tfrac{2i+1}{n}$ appear at positions $\tfrac{2k-1}{2\cos\pi(2i+1/2n)}$, see Figure \ref{fig:debruijn_vertical}. So their ordering is as follows: if until point $X$ there have been $n_i$ rhombuses of type $2i+1$ for each $i$ then the next rhombus is of index $j= \operatorname{argmin} \frac{n_j+\frac{1}{2}}{\cos\pi\frac{2j+1}{2n}}$.

\begin{figure}
\center
\includegraphics[width=6cm]{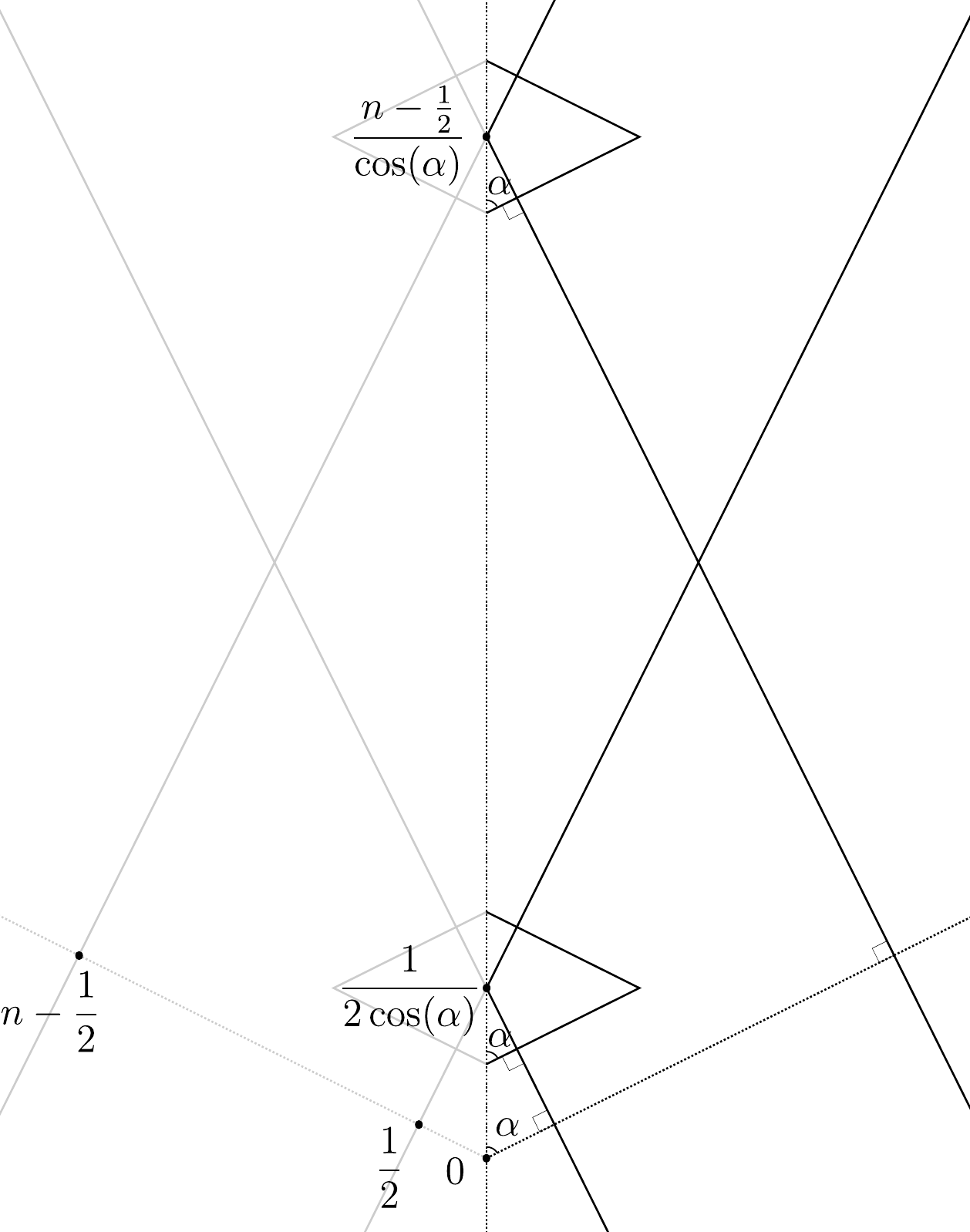}
\caption{Positions of the crossing points of type $\alpha$ on the vertical line in the DeBruijn multigrid $G_n(\frac{1}{2})$.}
\label{fig:debruijn_vertical}
\end{figure}

with the edgeword $u_{(j)}=v\bar{v}$, 

This builds exactly the same sequence as in Definition \ref{def:seq_p_and_omega}. So the sequence $\omega$ of rhombuses is on the vertical half-line. By the $2n$-fold  rotational symmetry of the grid (and the dual tiling) a cone of angle $\tfrac{\pi}{n}$ with the sequence $\omega$ of rhombuses  on both edges repeats $2n$ times around the origin in the dual tiling $P_n(\tfrac{1}{2},\tfrac{1}{2},\dots,\tfrac{1}{2})$.

Let us now use Proposition \ref{prop:cone_tileable} to complete the proof of Lemma \ref{lem:biglemma2}.
The tileability of the infinite cone with the edgeword $\omega$ implies that the 
inequality~(\ref{eq:adjacent_edge}) in Proposition~\ref{prop:tileability01} holds for the cone. (The proposition is stated for tileable 
bounded regions but it holds just as well for the infinite cone.) As (\ref{eq:adjacent_edge}) clearly implies 
the condition (\ref{eq:k1counting}) of Proposition \ref{prop:cs_tileability}, we have that for any fixed $K$ and all sufficiently large $j$ the 
pseudo-substitution $\candidate{j}$ satisfies
inequality (\ref{eq:k1counting})  for $k_1\leq K$. This is because the edgeword of $\candidate{j}$ has a common prefix of length $j$ with the edgeword
$\omega$ of the cone.
To get the tileability of the metatiles we now need to find a constant $K$ such that (\ref{eq:k1counting}) is satisfied also for
$k_1>K$.

Let $0<j_1<j_2<n$ be odd integers, and let us define $$g_{j_1,j_2}(k) := f_{|j_1-2|}^{-1}\circ f_{j_1}(k) -  f_{|j_2-2|}^{-1}\circ f_{j_2}(k).$$
Since the frequency of appearance of $j_1$ in the word $\omega$ (and in the word $u$) is strictly greater than the frequency of appearance of $j_2$, the idea is that $g_{j_1,j_2}(k)$ has a general increasing trend and (if $u$ is long enough) there is $K$ such that for any $k>K$ holds $g_{j_1,j_2}(k)>0$. This means the inequality (\ref{eq:k1counting}) is satisfied.

Let us recall that with $\gamma= \left( \cos(\frac{(2i+1)\pi}{2n})\right)_{0\leq i< \left\lfloor\frac{n}{2}\right\rfloor}$ we have $f_{2i+1}(k) \approx k\cdot \frac{\gamma_i}{\| \gamma \|}$ by the construction of the word $\omega$. For simplicity we define $\tilde{\gamma}_{2i+1} = \frac{\gamma_i}{\| \gamma \|}$, so that we have $f_j(k) \approx k\cdot \tilde{\gamma}_j$ for all odd $j$. This approximation is actually quite good in the sense that there exists $\delta$ dependent only on the dimension $n$ -- for example $\delta=2\sqrt{n}$ works -- such that for each $j$ and $k$ holds $|f_{j}(k)-k\cdot\tilde{\gamma}_j|< \delta$. This means that
$$|f_{|j_1-2|}^{-1}\circ f_{j_1}(k) - k\tfrac{\tilde{\gamma}_{j_1}}{\tilde{\gamma}_{|j_1-2|}}| < \delta + \delta/\tilde{\gamma}_{|j_1-2|}$$ so that $$|g_{j_1,j_2}(k)-k(\tfrac{\tilde{\gamma}_{j_1}}{\tilde{\gamma}_{|j_1-2}|}-\tfrac{\tilde{\gamma}_{j_2}}{\tilde{\gamma}_{|j_2-2|}})|<2\delta + \delta/\tilde{\gamma}_{|j_1-2|} + \delta/\tilde{\gamma}_{|j_2-2|}.$$
An easy calculation shows that the sequence $\gamma_1/\gamma_1, \gamma_3/\gamma_1, \gamma_5/\gamma_3, \dots$ is strictly decreasing so that $\tfrac{\tilde{\gamma}_{j_1}}{\tilde{\gamma}_{|j_1-2}|}-\tfrac{\tilde{\gamma}_{j_2}}{\tilde{\gamma}_{|j_2-2|}}$ is positive when $j_1<j_2$.
 Thus, for $$k> \frac{2\delta + \delta/\tilde{\gamma}_{|j_1-2|} + \delta/\tilde{\gamma}_{|j_2-2|}}{\tfrac{\tilde{\gamma}_{j_1}}{\tilde{\gamma}_{|j_1-2|}}-\tfrac{\tilde{\gamma}_{j_2}}{\tilde{\gamma}_{|j_2-2|}}}$$ we have that $g_{j_1,j_2}(k)>0$.
Now take
$$K:=\max\limits_{j_1<j_2} \frac{2\delta + \delta/\tilde{\gamma}_{|j_1-2|} + \delta/\tilde{\gamma}_{|j_2-2|}}{\tfrac{\tilde{\gamma}_{j_1}}{\tilde{\gamma}_{|j_1-2|}}-\tfrac{\tilde{\gamma}_{j_2}}{\tilde{\gamma}_{|j_2-2|}}}.$$
For any $k>K$ we have $g_{j_1,j_2}(k)>0$, which means that the inequality (\ref{eq:k1counting}) is satisfied.
So $\sigma_{(j)}$ is well-defined for all sufficiently large $j$: its metatiles are tileable.

  \end{proof}

\begin{lemma}
  \label{lem:biglemma3}
  There exists $N \in \mathbb{N}$ such that for all $j>N$, if $\candidate{j}$ is well defined then $\candidate{j}$ is a primitive substitution and $\candidate{j}$ admits a tiling with global $2n$-fold rotational symmetry.
\end{lemma}

\begin{proof}[Proof of Lemma \ref{lem:biglemma3}]
  Assume that $j$ is an integer such that $\candidate{j}$ is well-defined and $j> N$ where $N$ is the length of the shortest prefix of $\omega$ that contains at least one occurrence of each letter.
  \begin{itemize}
  \item The substitution is primitive.
    Indeed, take a tile $t$. Along each side of the boundary of $\candidate{j}(t)$ we have at least one rhombus of each type (but in only one orientation).
    On the first metaedge of the boundary we have all directions of edges but one (the one which is perpendicular to the edge of the metatile). On the second metaedge
    of the boundary  we have all directions except the one perpendicular to this second edge of the metatile. So altogether the two sides  contain all directions of edges.
    Now, if we look at the rhombuses on the boundary of $\candidate{j}^2(t)$ we have every rhombus in every orientation. See Figure
    \ref{fig:primitivity} for the illustration for $n=5$.
    So for any rhombus tile $t$, the image $\candidate{j}^2(t)$ contains every rhombus tile in every orientation,
    which means that $\candidate{j}$ is primitive of order 2.
    \begin{figure}
      \includegraphics[width=\textwidth]{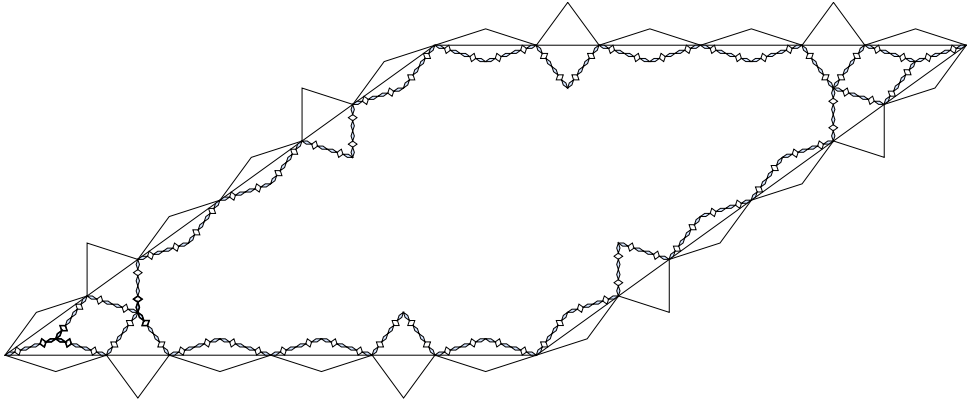}
      \caption{Primitivity for $n=5$, in thick lines one occurrence of each orientation of each rhombus.}
      \label{fig:primitivity}
    \end{figure}
    This implies that any tiling admissible for $\candidate{j}$ is uniformly recurrent.
    Note that this is indifferent to the interior of the metatiles and is only due to the edgeword of the substitution.
  \item Let the $n$-star denoted by $S_n$ be the pattern of a corolla of rhombuses of angle $\tfrac{\pi}{n}$ around a vertex, as in Figure \ref{fig:star_patterns_579}.
    \begin{figure}
    \center  \includegraphics[width=10cm]{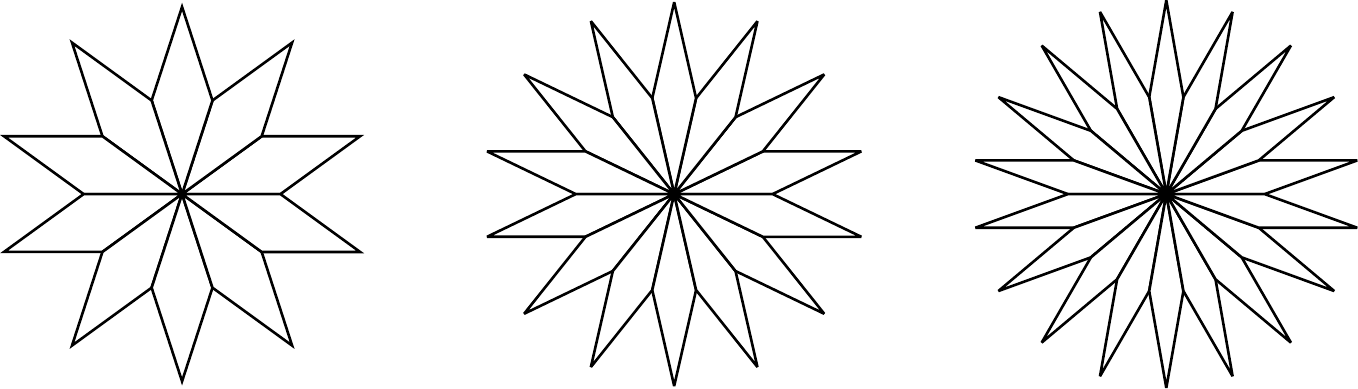}
    \caption{The star patterns $S_5$, $S_7$ and $S_9$.}
    \label{fig:star_patterns_579}
    \end{figure}
    First of all, $S_n$ appears at the centre in $\candidate{j}(S_n)$. Indeed, in the narrow angle of the rhombus of type 1 (\ie, of angle $\tfrac{\pi}{n}$), there is a portion of the star as illustrated in Figure \ref{fig:star_rhombus}.
    By immediate recurrence it follows that $\candidate{j}^k(S_n)$ is at the centre of $\candidate{j}^{k+1}(S_n)$.
    \begin{figure}
      \center  \includegraphics[width=10cm]{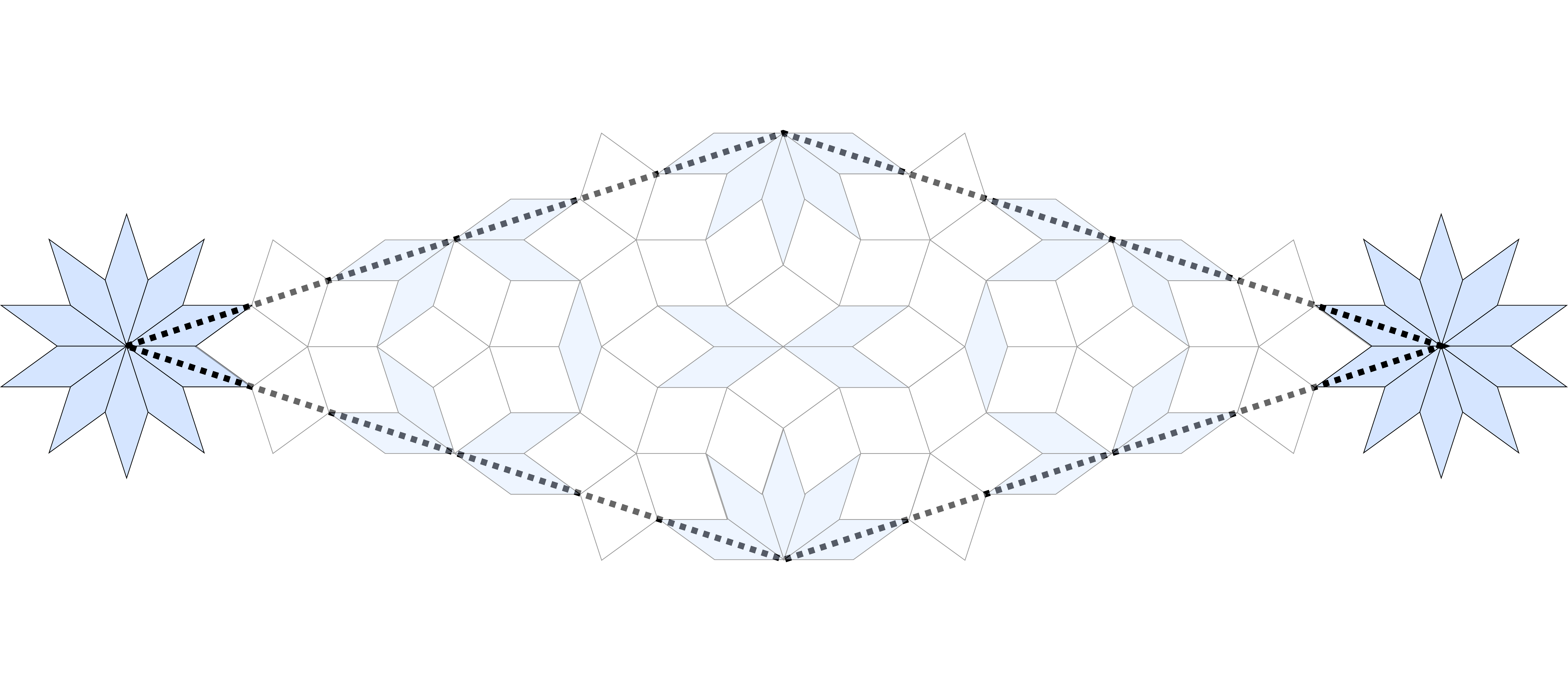}
      \caption{In the narrow corner of the image by $\candidate{j}$ of the rhombus of angle $\tfrac{\pi}{n}$ we always have a portion of the star $S_n$, here for $n=5$.}
      \label{fig:star_rhombus}
    \end{figure}
    But this is not yet enough to prove that $S_n$ is admissible for $\candidate{j}$. We need $S_n$
     to appear in some (high order) metatile.

    Observe that the narrow rhombus $r_0$ is always the first and last rhombus in the edgeword.
    It implies that at a vertex of a meta-rhombus at least two rhombuses $r_0$ meet, corresponding to the two edges of the metatile. In addition, if the angle of the meta-rhombus is $\tfrac{k\pi}{n}$, there are some rhombuses whose angles add up to $\tfrac{(k-1)\pi}{n}$. See Figure \ref{fig:metatiles_corners} for an illustration.
    \begin{figure}
    \center
    \includegraphics[width=\textwidth]{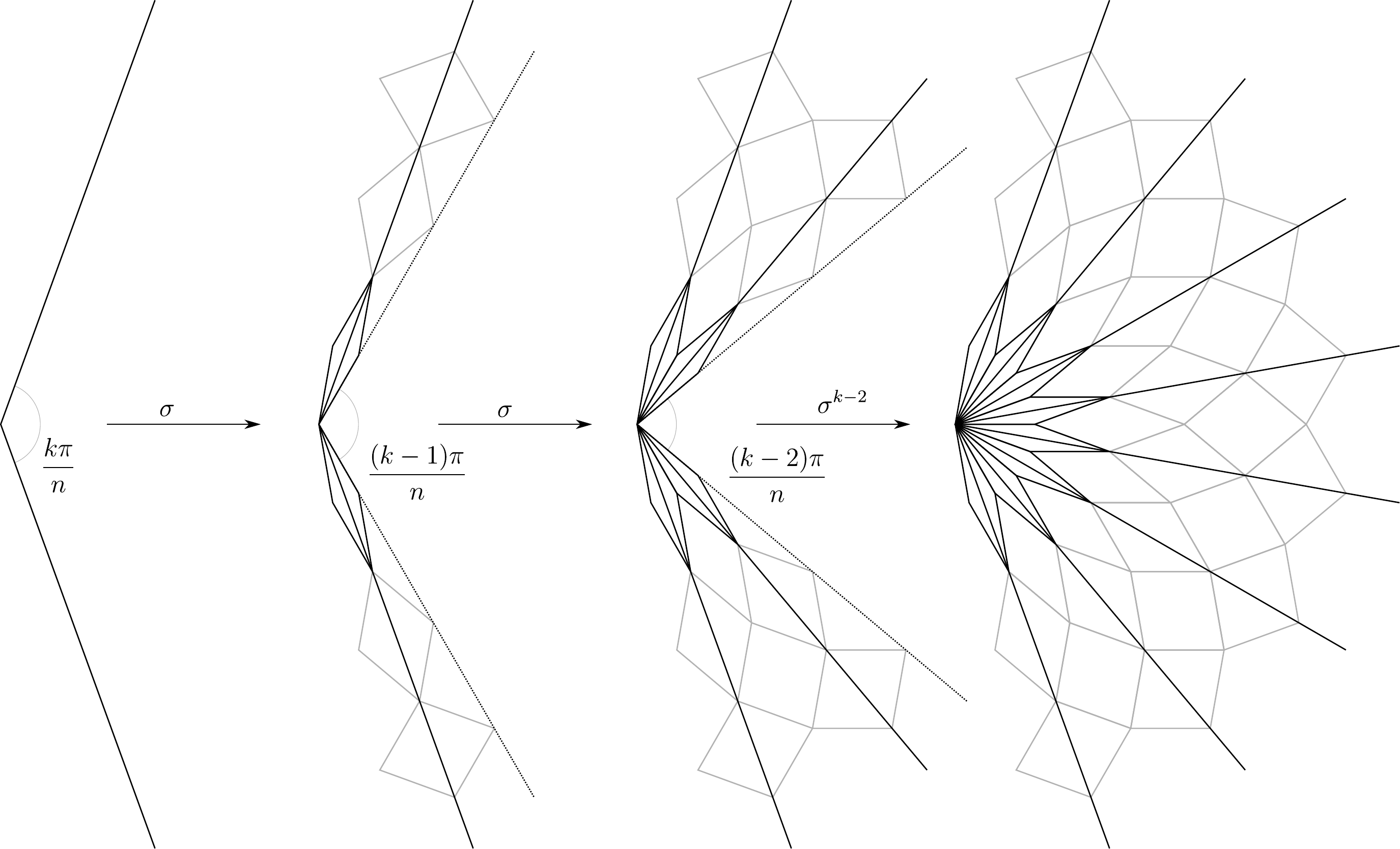}
      \caption{The angles in the corner of metatiles under the iterations of the substitution.}
      \label{fig:metatiles_corners}
    \end{figure}
    If we iterate the substitution, in the next step
    we have four rhombuses $r_0$ and some rhombuses of sum of angles at most $\tfrac{(k-2)\pi}{n}$.
    Eventually, for any rhombus, in any corner of $\candidate{j}^{n-1}$ there are only narrow rhombuses. This means that in $\candidate{j}^n(r_0)$ the star $S_n$ appears. Indeed, take any interior vertex in $\sigma(r_0)$. This vertex is surrounded by rhombuses. Now iterate $n-1$ times the substitution centered on this vertex. The vertex becomes the centre of the corolla $S_n$. So $S_n$ is in $\candidate{j}^n(r_0)$.

    Overall the star $S_n$ is legal for $\candidate{j}$ and stable under $\candidate{j}$.

  \item
    Let us define the limit tiling from seed $S_n$ as
    $$\tiling_{\infty} := \lim\limits_{k\to\infty} \candidate{j}^k(S_n).$$
    The tiling $\tiling_{\infty}$ is well defined because the sequence $\left(\candidate{j}^k(S_n)\right)_{k\in\mathbb{N}}$ is increasing for the inclusion, and it is a tiling such that for any $k, \ \candidate{j}^k(S_n)$ is the central patch of the tiling.
    And since every $\candidate{j}^k(S_n)$ has $2n$-fold rotationnal symmetry around its center, $\tiling_\infty$ also has global $2n$-fold rotationnal symmetry.

  \end{itemize}

\end{proof}

  For a step-by-step description of the Planar Rosa construction for small $n$ see \cite[\S 6]{lutfalla2021thesis}.


\appendix
\section*{Appendices}

\begin{lemma}[Eigenvalues of the elementary matrices]
  \label{lemma:appendices_eigen_elem_mat}
  Let $n$ be an odd integer and $i\in\{0,\dots, \lfloor\frac{n}{2}\rfloor -1\}$.
  The elementary matrix $M_i(n)$ defined in Definition \ref{def:elementary_matrices} has eigenspaces $\Delta$ and $\e{j}$ for $0\leq j < \lfloor\frac{n}{2}\rfloor$ with eigenvalue $\lambda_\Delta=0$ and
  $$\lambda_{i,j}(n)= 2\cos\left(\frac{(2j+1)(2i+1)\pi}{2n}\right)e^{- \imag \frac{(2j+1)\pi}{2n}} $$
\end{lemma}

\begin{proof}
  Let $z$ be a complex number such that $z^n=1$ and $Z$ be the vector $(z^k)_{0\leq k < n}$.
  We have \[ Z\cdot M_i(n) = (-1)^i(z^{\fst{i}} - z^{\scd{i}})\cdot Z\]
  With $z=(e^{\imag\frac{2(2j+1)\pi}{n}})$ we get that $\e{j}$ is eigenspace of $M_i(n)$ with eigenvalue $\lambda_{i,j}(n)$
  \begin{align*}\lambda_{i,j}(n) &= (-1)^i\left(e^{\imag\frac{2(2j+1)\fst{i}\pi}{n}}-e^{\imag\frac{2(2j+1)\scd{i}\pi}{n}}\right) \\
    &= (-1)^ie^{\imag\frac{2(2j+1)i\frac{n+1}{2}\pi}{n}} + (-1)^{i+1}e^{-\imag\frac{2(2j+1)(i+1)\frac{n+1}{2}\pi}{n}} \\
    &= (-1)^ie^{\imag\frac{(2j+1)i(n+1)\pi}{n}} + (-1)^{i+1}e^{-\imag\frac{(2j+1)(i+1)(n+1)\pi}{n}} \\
    &= (-1)^ie^{\imag(2j+1)i\pi}e^{\imag\frac{(2j+1)i\pi}{n}} +(-1)^{i+1}e^{-\imag(2j+1)(i+1)\pi}e^{-\imag\frac{(2j+1)(i+1)\pi}{n}}\\
    &= (-1)^i(-1)^ie^{\imag\frac{(2j+1)i\pi}{n}} + (-1)^{i+1}(-1)^{i+1}e^{-\imag\frac{(2j+1)(i+1)\pi}{n}}\\
    &= e^{\imag\frac{(2j+1)i\pi}{n}} + e^{-\imag\frac{(2j+1)(i+1)\pi}{n}}\\
    &= \left(2\cos\left(\frac{(2j+1)(2i+1)\pi}{2n}\right) \right)e^{- \imag \frac{(2j+1)\pi}{2n}} \end{align*}
  With $z=1$ we get that $\Delta$ is an eigenspace with eigenvalue $0$.
\end{proof}

\begin{lemma}[Trigonometric manipulations]
  \label{lemma:appendices_weird_trig_sum}
  Let $k$ be an integer and $i\in\{0,\dots k-1\}$.
  Let us define $C_{j,k}$ by
  $$C_{j,k}:=\sum\limits_{i=0}^{k-1}4(k-i)\cos\left((2i+1)\theta_{j,k}\right),$$
  with $\theta_{j,k}:=\frac{(2j+1)\pi}{2(2k+1)}$.
  We have
  $$C_{j,k}\cdot \sin^2(\theta_{j,k}) = \cos(\theta_{j,k}).$$

\end{lemma}

\begin{proof}
    Let us write $\theta$ for $\theta_{j,k}$ for the sake of simplicity.
    \begingroup
    \allowdisplaybreaks
  \begin{align*}
    &C_{j,k} \sin^2\left(\theta\right) \\
    &= \sum\limits_{i=0}^{k-1}4(k-i)\cos\left(\frac{(2j+1)(2i+1)\pi}{2(2k+1)}\right)\sin^2\left(\theta\right)\\
    &= \sum\limits_{i=0}^{k-1}4(k-i)\frac{e^{\imag(2i+1)\theta} + e^{-\imag(2i+1)\theta}}{2}\left( \frac{e^{\imag\theta} - e^{-\imag\theta}}{2i}\right)^2\\
    &= \sum\limits_{i=0}^{k-1}\tfrac{k-i}{2}\left(e^{\imag(2i+1)\theta} + e^{-\imag(2i+1)\theta}\right)\left( 2 - e^{\imag2\theta} - e^{-\imag2\theta}\right)\\
    &= \sum\limits_{i=0}^{k-1}\tfrac{k-i}{2} \Big( 2e^{\imag(2i+1)\theta} + 2e^{-\imag(2i+1)\theta} -e^{\imag(2i+3)\theta} \\
    & \qquad \qquad \qquad -e^{-\imag(2i-1)\theta}  -e^{\imag(2i-1)\theta}  -e^{-\imag(2i+3)\theta}\Big)\\
    &= \sum\limits_{i=0}^{k-1}(k-i)\left( 2\cos((2i+1)\theta) - \cos((2i+3)\theta) - \cos((2i-1)\theta)  \right)\\
    &= \sum\limits_{i=0}^{k-1}(k-i)2\cos((2i+1)\theta) - \sum\limits_{i=0}^{k-1}(k-i)\cos((2i+3)\theta) \\
    & \qquad -\sum\limits_{i=0}^{k-1}(k-i)\cos((2i-1)\theta)\\
    &= \sum\limits_{i=0}^{k-1}(k-i)2\cos((2i+1)\theta) - \sum\limits_{i=1}^{k}(k+1-i)\cos((2i+1)\theta) \\
    & \qquad  -\sum\limits_{i=-1}^{k-2}(k-1-i)\cos((2i+1)\theta)\\
    &= \sum\limits_{i=0}^{k-1}(k-i)2\cos((2i+1)\theta) \\
    &\qquad - \sum\limits_{i=0}^{k-1}(k+1-i)\cos((2i+1)\theta) +(k+1)\cos\theta - \cos\tfrac{(2j+1)\pi}{2} \\
    &\qquad -\sum\limits_{i=0}^{k-1}(k-1-i)\cos((2i+1)\theta) + 0\cos((2k-1)\theta) - k\cos\theta\\
    &= \sum\limits_{i=0}^{k-1}\left(2(k-i)-(k+1-i)-(k-1-i)\right)\cos((2i+1)\theta) \\
    &\qquad +(k+1-k)\cos\theta\\
    &= \cos\theta
  \end{align*}
  \endgroup

\end{proof}

\begin{lemma}[The eigenvalue matrix $\eigenmatrix{n}$ is orthogonal up to a scalar]
  \label{lemma:orthogonal}
  ~\\
  Let $n$ be an odd integer, the  matrix $\eigenmatrix{n}$ from Definition \ref{def:eigenmatrix} page \pageref{def:eigenmatrix} is orthogonal up to a scalar.
More precisely,  $\tfrac{1}{\sqrt{n}}\eigenmatrix{n}$ is an orthogonal matrix, \ie, \[ \tfrac{1}{\sqrt{n}}\eigenmatrix{n}\cdot \tfrac{1}{\sqrt{n}}\eigenmatrix{n}^\top = \mathrm{Id}_{\floor{\frac{n}{2}}}.\]
\end{lemma}
\begin{proof}
  Let us first recall that $n$ is an odd integer.
  Let us define three matrices $A$, $B$ and $C$ by:
  \begin{align*}
    A&:= \left(a_i \cos\frac{i(2j+1)\pi}{2n}\right)_{0\leq i,j < n} \qquad \text{with: } a_i = \begin{cases} \sqrt{\frac{1}{n}} \text{ if } i=0\\ \sqrt{\frac{2}{n}} \text{ otherwise}\end{cases}\\
    B&:= \left(\sqrt{\frac{2}{n}}\cos\frac{(2i+1)(2j+1)\pi}{2n}\right)_{0\leq i < \floor{\frac{n}{2}}, 0\leq j < n } \\
    C&:=  \tfrac{1}{\sqrt{n}}\eigenmatrix{n} = \left(\frac{2}{\sqrt{n}}\cos\frac{(2i+1)(2j+1)\pi}{2n}\right)_{0\leq i,j < \floor{\frac{n}{2}}}
  \end{align*}

  Let us first remark that $A$ is a Discrete Cosine Transform matrix, sometimes called DCT-III, which is known to be orthogonal.
  From this we will prove that $B\cdot B^\top = \mathrm{Id}_{\floor{\frac{n}{2}}}$ and then that $C$ is orthogonal.

  Let us prove $B\cdot B^\top = \mathrm{Id}_{\floor{\frac{n}{2}}}$.
  Let us look at the $j,k$ coefficient of $B\cdot B^\top$ for $0\leq j,k < \tfrac{n}{2}$.
  \begin{align*}
    \left(B\cdot B^\top\right)_{j,k} &= \sum\limits_{i=0}^{n-1} \sqrt{\frac{2}{n}}\cos\left(\frac{(2j+1)(2i+1)\pi}{2n}\right) \sqrt{\frac{2}{n}}\cos\left(\frac{(2k+1)(2i+1)\pi}{2n}\right)\\
    &= \left(A\cdot A^\top\right)_{2j+1,2k+1}\\
    &= \begin{cases} 1 \text{ if } j = k \\ 0 \text{ otherwise} \end{cases}
  \end{align*}

  Now from the fact that $B\cdot B^\top = \mathrm{Id}_{\floor{\frac{n}{2}}}$ let us prove that $C$ is orthogonal.
  Let us first remark that for $0\leq i < \floor{\tfrac{n}{2}}$ we have
  $$B_{i,\floor{\frac{n}{2}}} = \sqrt{\frac{2}{n}}\cos\left((2i+1)\frac{\pi}{2}\right) = 0.$$
  This is due to the fact that $n$ is odd, which implies that $n = 2\floor{\tfrac{n}{2}}+1$.\\
  Let us also remark that for $0\leq i,j <\floor{\tfrac{n}{2}}$ we have
  $$\sqrt{2} B_{i,j} = C_{i,j} \qquad \text{ and } \qquad B_{i,n-1-j} = -B_{i, j}.$$
  The first equality is just a reformulation of the definition, and for the second equality let us develop $B_{i,n-1-j}$ as follows:
  \begin{align*}
    B_{i,n-1-j} &= \sqrt{\frac{2}{n}} \cos \frac{(2i+1)(2(n-1-j)+1)\pi}{2n}\\
    &= \sqrt{\frac{2}{n}} \cos \frac{(2i+1)(2n - 2 - 2j + 1)\pi}{2n}\\
    &= \sqrt{\frac{2}{n}} \cos \frac{(2i+1)(2n - (2j + 1))\pi}{2n}\\
    &= \sqrt{\frac{2}{n}} \cos\left( (2i+1)\pi - \frac{(2i+1)(2j + 1)\pi}{2n}\right)\\
    &= \sqrt{\frac{2}{n}} \cos\left( \pi - \frac{(2i+1)(2j + 1)\pi}{2n}\right)\\
    &= -\sqrt{\frac{2}{n}} \cos\left(\frac{(2i+1)(2j + 1)\pi}{2n}\right)\\
    &= -B_{i,j}.
  \end{align*}

  Now let us prove that $\left(C\cdot C^\top\right)_{j,k} = \left(B\cdot B^\top\right)_{j,k}$ for $0\leq j,k < \floor{\tfrac{n}{2}}$.
  \begingroup
    \allowdisplaybreaks
  \begin{align*}
    \left(B\cdot B^\top\right)_{j,k} &= \sum\limits_{i=0}^{n-1} B_{j,i}B_{k,i} \\
    &= \left(\sum\limits_{i=0}^{\floor{\frac{n}{2}}-1} B_{j,i}B_{k,i} \right) + B_{j,\floor{\frac{n}{2}}}B_{k,\floor{\frac{n}{2}}}  + \sum\limits_{i=\floor{\frac{n}{2}}+1}^{n-1} B_{j,i}B_{k,i} \\
    &= \left(\sum\limits_{i=0}^{\floor{\frac{n}{2}}-1} B_{j,i}B_{k,i} \right) + 0  + \sum\limits_{i=\floor{\frac{n}{2}}+1}^{n-1} B_{j,i}B_{k,i} \\
    &= \left(\sum\limits_{i=0}^{\floor{\frac{n}{2}}-1} B_{j,i}B_{k,i} \right) + \sum\limits_{i=0}^{\floor{\frac{n}{2}}-1} B_{j,n-1-i}B_{k,n-1-i} \\
    &= \sum\limits_{i=0}^{\floor{\frac{n}{2}}-1} \left( B_{j,i}B_{k,i}  + (-B_{j,i})\cdot(-B_{k,i}) \right)\\
    &= \sum\limits_{i=0}^{\floor{\frac{n}{2}}-1} 2B_{j,i}B_{k,i}\\
    &= \sum\limits_{i=0}^{\floor{\frac{n}{2}}-1} \sqrt{2}B_{j,i}\sqrt{2}B_{k,i}\\
    &= \sum\limits_{i=0}^{\floor{\frac{n}{2}}-1} C_{j,i}C_{k,i}\\
    &= \left( C\cdot C^\top\right)_{j,k}
  \end{align*}
  \endgroup

  Hence $C$ is orthogonal and $\eigenmatrix{n}$ is orthogonal up to a scalar.

\end{proof}


\bibliographystyle{alpha}
\bibliography{nfold}

\end{document}